\documentclass[12pt,reqno]{amsart}
\usepackage{amsmath,amsfonts,amsthm}
\newtheorem{theorem}{Theorem}[section]
\newtheorem{prop}[theorem]{Proposition}
\newtheorem{lemma}[theorem]{Lemma}
\newtheorem{conj}[theorem]{Conjecture}
\newtheorem{coro}[theorem]{Corollary}
\newtheorem{definition}[theorem]{Definition}
\usepackage[a4paper,
top=24mm,bottom=24mm,left=26mm,right=26mm]{geometry}

\usepackage[breaklinks=true]{hyperref}
\usepackage[numbers,sort]{natbib}

\usepackage[footnotesize,hang,sf]{caption}
\setcaptionwidth{12cm}

\usepackage[]{graphicx}            
 \usepackage{booktabs}
\usepackage{tikz}
\usetikzlibrary{positioning}
\usetikzlibrary{decorations.markings}
\tikzset{->-/.style={decoration={
      markings,
      mark=at position #1 with {\arrow{>}}}, postaction={decorate}}}
\tikzset{->>-/.style={decoration={
      markings,
      mark=at position #1 with {\arrow{>>}}}, postaction={decorate}}}

\usepackage{mathtools}
\usepackage{makecell}
\usepackage{mathrsfs}
\usepackage[T1]{fontenc}
\usepackage[sb]{libertine}
\usepackage[varqu,varl,
]{zi4}  
\usepackage[libertine,bigdelims,vvarbb]{newtxmath}
\usepackage[
 cal=esstix,
 scr=rsfs,
 bb=boondox,
]{mathalfa}
\numberwithin{equation}{section}

\newcommand{\I}{\mathrm{i}}
\newcommand{\E}{\mathrm{e}}

\DeclareMathOperator{\sh}{sh}
\DeclareMathOperator{\ch}{ch}

\DeclareMathOperator{\diag}{diag}
\DeclareMathDelimiter{\Norm}{\mathord}{largesymbols}{"3E}{largesymbols}{"3E}

\DeclareMathOperator{\KBS}{KBS}
\DeclareMathOperator{\Mod}{Mod}
\DeclareMathOperator{\Const}{Const}
\DeclareMathOperator{\SH}{SH}
\allowdisplaybreaks[3]
\begin{document}

\baselineskip 16pt
\parskip 8pt
\sloppy


\title[DAHA and Skein Algebra]{DAHA and Skein Algebra of Surface:
  Double-Torus Knots}


\author[K. Hikami]{Kazuhiro \textsc{Hikami}}

\address{Faculty of Mathematics,
  Kyushu University,
  Fukuoka 819-0395, Japan.}

\email{
  \texttt{khikami@gmail.com}
}

\date{December 27, 2018}


\begin{abstract}
  We study a topological aspect of rank-$1$ double affine Hecke
  algebra (DAHA).
  Clarified is
  a relationship between the DAHA of $A_1$-type
  (resp. $C^\vee C_1$-type) and the skein algebra on a
  once-punctured torus
  (resp. a $4$-punctured sphere), and the $SL(2;\mathbb{Z})$ actions
  of DAHAs are
  identified with
  the Dehn twists on the surfaces.
  Combining these two types of DAHA, we construct the DAHA
  representation for the skein algebra on a genus-two surface,
  and we propose a DAHA polynomial for a double-torus knot, which
  is a simple closed curve on a  genus two Heegaard surface in $S^3$.
  Discussed is 
  a relationship between the DAHA polynomial and the colored Jones polynomial.
\end{abstract}



\subjclass[2000]{
}


\maketitle



\section{Introduction}

The double affine Hecke algebra (DAHA)
was introduced by Cherednik, and it is a powerful tool
in  studies of the
Macdonald polynomials associated with root systems
(see, \emph{e.g.},~\cite{Chered05Book,Macdonald03Book}).
The Macdonald polynomial is ubiquitous in mathematics and physics,
and  an interpretation as a $q$-deformation of a wave-function
of  the quantum Hall effect suggest an
importance of  a topological structure of the
DAHA in studies of topological orders.
The DAHA was recently
applied to quantum topology.
Proposed was a DAHA polynomial
invariant~\cite{IChered13a,IChered16a},
and discussed  was  a relationship with the refined Chern--Simons
invariant and the Khovanov homology~\cite{AganaShaki12a}.
The  construction of the DAHA polynomial~\cite{IChered13a}  is
purely algebraic,
but
the DAHA polynomial  is   limited only to
torus knots and their descendants,
\emph{i.e.},
all are non-hyperbolic.
An attempt~\cite{ArthaShaki17a} was made towards
DAHA for a genus-two surface generalizing DAHA of $A_1$-type,
but  a relationship with the known
quantum polynomial invariants is unclear.

A purpose of this article is to combine two rank-1  DAHAs of $A_1$-type
and $C^\vee C_1$-type to construct the DAHA representation for
double-torus knots.
The double-torus knot~\cite{PHill99a,HillMura00a}
is a simple closed curve on a genus two Heegaard
surface in $S^3$, and a large family of knots such as twist knots
belong to this type.
Originally the DAHAs of $A_1$-type and $C^\vee C_1$-type are for the
Rogers polynomial (or the $q$-ultraspherical polynomial)
and the Askey--Wilson polynomial respectively.
The Askey--Wilson polynomial~\cite{AsWi85} is  on top of the Askey scheme of
classification of orthogonal polynomials of hypergeometric-type, and 
its algebraic structure  receives recent active interests
(see~\cite{Koorn07a,Koorn08a,Zheda91a,Terwil13a}).
Here  we pay attentions to a relationship between the DAHA and
the Kauffman bracket skein
algebra on surfaces.
It is known
that the DAHA of $C^\vee C_1$-type 
represents a quantization of the affine cubic surface which is the  character
variety of a $4$-punctured sphere~\cite{Oblom04a},
while the DAHA of $A_1$-type  is related to the character variety of a
once-punctured torus.
Based on the fact~\cite{Bullo97a,PrzytSikor00a}
that  the coordinate ring of the character varieties is a
specialization of the Kauffman bracket skein algebra,
discussed
also is a relationship with the skein
algebra on the $4$-punctured sphere and the once-punctured
torus~\cite{BerestSamuel14a,BerestSamuel16a}.

For each simple closed curve on the genus-two surface,
we assign a DAHA operator  which represents the skein algebra on the surface.
A benefit of our method   combining  two types of rank-1 DAHAs is that we can make
use of their well-known automorphisms.
Due to the relationship between the DAHA and the skein algebra,
the DAHA automorphisms are regarded as the mapping class group~\cite{Birman74,FarbMarg11Book}, the group of isotopy
classes of orientation-preserving diffeomorphisms of surface.
As the mapping class group
is generated by
the Dehn twists,
we  can clarify the $SL(2;\mathbb{Z})$ actions
of the DAHA of $A_1$-type and
$C^\vee C_1$-type
as the Dehn twists about
curves on each surface.
The $q$-difference DAHA operator is indeed constructed by use of the
automorphisms of DAHA as in the case of torus knots by
Cherednik~\cite{IChered13a}.
Using the DAHA operator assigned to a simple closed curve $\mathbb{c}$
on the
surface, we propose a DAHA polynomial for $\mathbb{c}$.
We  compute explicitly  the DAHA polynomial for double-twist knots,
and discuss a relationship with the colored Jones polynomial.

This paper is organized as follows.
In Section~\ref{sec:torus},  we study the once-punctured torus $\Sigma_{1,1}$.
We recall properties of
the DAHA of $A_1$-type, and establish a relationship with the skein
algebra on $\Sigma_{1,1}$.
The DAHA polynomial proposed by Cherednik is also reviewed.
Section~\ref{sec:sphere} is for the $4$-punctured sphere $\Sigma_{0,4}$.
We recall both 
the DAHA of $C^\vee C_1$-type and the skein algebra on $\Sigma_{0,4}$.
Based on the correspondence,
we associate  a DAHA operator for a simple closed curve with a
rational slope.
In these sections, the $SL(2;\mathbb{Z})$-actions on the rank-1 DAHAs
are interpreted as the Dehn twists about certain curves on the
surfaces $\Sigma_{1,1}$ and $\Sigma_{0,4}$.
In Section~\ref{sec:twice_torus},
as a prototype toward the genus-two surface,
we study the skein algebra on a twice-punctured torus $\Sigma_{1,2}$.
We ``glue'' two types of the rank-1 DAHAs, $A_1$-type and $C^\vee C_1$-type,
using a quantum dilogarithm function, and give the DAHA representation
of the skein algebra on $\Sigma_{1,2}$.
Section~\ref{sec:genus-two} is for the genus-two surface.
We propose the DAHA polynomial for a simple closed curve on the
surface, and study a relationship with the colored Jones polynomial.
In the rest of this section, we collect our notations such as special functions.

\subsection{Preliminaries}
The Kauffman bracket  skein module $\KBS_A (M)$ of a
3-manifold $M$ is defined by
\begin{gather}
  \raisebox{-5mm}{
    \begin{tikzpicture}
      \draw [line width=1.2pt](0,1) --( 1,0) ;
      \draw[line width=10pt,white](0,0)--(1,1);
      \draw[line width=1.2pt](0,0)--(1,1);
    \end{tikzpicture}
  }
  =
  A \,
  \raisebox{-\ht\strutbox}{
    \begin{tikzpicture}
      \draw [line width=1.2pt] (1,1) to[out=-135,in=135] (1,0);
      \draw [line width=1.2pt] (0,1) to[out=-45,in=45] (0,0);
    \end{tikzpicture}
  }
  +A^{-1} \,
  \raisebox{-\ht\strutbox}{
    \begin{tikzpicture}
      \draw [line width=1.2pt] (1,1) to[out=-135,in=-45] (0,1);
      \draw [line width=1.2pt] (1,0) to[out=135,in=45] (0,0);
    \end{tikzpicture}
  }
  ,
  \\[2mm]
  \raisebox{-\ht\strutbox}{
    \begin{tikzpicture}
      \draw [line width=1.2pt] (0,0) circle (0.5) ;
    \end{tikzpicture}
  }
  = -A^2 - A^{-2} .
  \notag
\end{gather}
When $M=\Sigma \times [0,1]$ with an oriented surface $\Sigma$,
we write $\KBS_A(\Sigma)$.
Here,
a multiplication
$\mathbb{x} \, \mathbb{y}$
of curves $\mathbb{x}$ and $\mathbb{y}$
means
that $\mathbb{x}$ is vertically above $\mathbb{y}$,
\begin{gather}
  \mathbb{x} \, \mathbb{y}
  =
  \newcolumntype{C}{>{$}c<{$}}
  \begin{tabular}{|C|}
    \hline
    \hphantom{aa}\mathbb{x}\hphantom{aa}
    \\
    \hline
    \mathbb{y}
    \\
    \hline
  \end{tabular}
\end{gather}

Throughout this paper,
we  use  for simplicity
a variant of hyperbolic functions
\begin{align}
  \label{hyperbolic}
  &\sh(x) = x- x^{-1},
  &
  &\ch(x)=x+x^{-1} .
\end{align}
We recall the standard notations of $q$-calculus.
We use the $q$-Pochhammer symbol
defined by
\begin{gather}
  (x;q)_n =
  \prod_{j=1}^n
  \left( 1 - x \, q^{j-1} \right)
  ,
  \\
  \left( x_1, x_2, \cdots ; q \right)_n
  =
  \left( x_1 ; q \right)_n
  \left( x_2; q \right)_n \cdots
  .
  \notag
\end{gather}
Here we mean $(x;q)_0=1$,  and for negative integers 
\begin{equation}
  (x;q)_{-n} =
  \frac{(x;q)_\infty}{(x\, q^{-n};q)_\infty}
  =
  \frac{1}{\left( x \, q^{-n};q \right)_n} .
\end{equation}
We also use the $q$-hypergeometric series
\begin{equation}
  {}_r\phi_s
  \left[
    \begin{matrix}
      a_1,  \cdots,
      a_r
      \\
      b_1, \cdots,
      b_s
    \end{matrix}
    ;
    q,
    z
  \right]
  =
  \sum_{n=0}^\infty
  \frac{
    (a_1, \cdots, a_r ; q)_n
  }{
    (q, b_1, \cdots, b_s ;q)_n 
  }
  \left(
    (-1)^n q^{\frac{1}{2} n (n-1)}
  \right)^{1+s-r}
  z^n .
\end{equation}
See, \emph{e.g.},~\cite{GaspRahm04} for properties of the hypergeometric
functions.

\section*{Acknowledgments}
The author would like to thank H.~Fuji, A.N.~Kirillov, and H.~Murakami
for communications and comments on a draft of the manuscript.
A part of this work was presented at the workshop ``Volume Conjecture
in Tokyo'' on August 2018,
celebrating the 60th birthday of Jun~Murakami and Hitoshi~Murakami.
Thanks to the organizers and the participants.
This work is supported in part by KAKENHI
JP16H03927,
JP17K05239,
JP17K18781.

\section{Once-Punctured Torus}
\label{sec:torus}
\subsection{Skein Algebra}
We study the skein module on a once-punctured torus $\Sigma_{1,1}$.
We set simple closed curves
$\mathbb{x}$,
$\mathbb{y}$,
$\mathbb{z}$,
and
$\mathbb{b}$ as in Fig.~\ref{fig:xyz_torus}.
See that $\mathbb{b}$ denotes the boundary circle of the puncture.

\begin{figure}[htbp]
  \centering
    \includegraphics[scale=1.0]{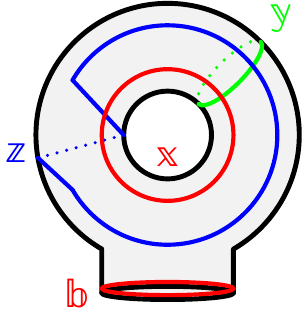}
  \caption{Depicted are simple closed curves on the  once-puncture torus $\Sigma_{1,1}$.}
  \label{fig:xyz_torus}
\end{figure}

\begin{prop}[\cite{BulloPrzyt99a}]
  \label{prop:Sigma11}
  The $\KBS_A(\Sigma_{1,1})$ is generated by
  $\mathbb{x}$,
  $\mathbb{y}$, and
  $\mathbb{z}$,
  satisfying
  \begin{align}
    A \, \mathbb{x} \, \mathbb{y}
    - A^{-1} \,  \mathbb{y} \, \mathbb{x}
    &=
      \left( A^2-A^{-2} \right)  \, \mathbb{z}
      ,
      \notag
    \\
    A  \,  \mathbb{y} \,  \mathbb{z}
    - A^{-1} \,  \mathbb{z} \,  \mathbb{y}
    & =
      \left( A^2-A^{-2} \right) \, \mathbb{x}
      ,
      \notag
    \\
    A  \, \mathbb{z} \, \mathbb{x}
    - A^{-1}\,  \mathbb{x}  \,  \mathbb{z}
    & =
      \left( A^2-A^{-2} \right) \, \mathbb{y}
      .
      \label{torus_xy_z}
  \end{align}
\end{prop}

It is noted
that the boundary circle $\mathbb{b}$ is generated by
\begin{equation}
  \label{boundary_b}
  \mathbb{b}
  =
  A \, \mathbb{x} \, \mathbb{y} \, \mathbb{z} -
  A^2 \, \mathbb{x}^2 - A^{-2} \, \mathbb{y}^2 - A^2 \, \mathbb{z}^2
  +A^2+A^{-2}
  .
\end{equation}



\subsection{DAHA of $A_1$-type}
We collect several
properties of DAHA of $A_1$-type.
Essential references are~\cite{Chered05Book,Macdonald03Book}.

\begin{definition}
  The DAHA of $A_1$-type
  $\mathcal{H}_{q,t}$ is $\mathbb{C}(q,t)$-algebra generated by
  $\mathsf{Y}^{\pm 1}$, $\mathsf{X}^{\pm 1}$, $\mathsf{T}^{\pm 1}$
satisfying
\begin{gather}
  \left( \mathsf{T} + t \right)
  \, \left( \mathsf{T} - t^{-1} \right) = 0
  ,
  \notag
  \\
  \mathsf{T} \, \mathsf{X} \, \mathsf{T}= \mathsf{X}^{-1}
  ,
  \notag
  \\
  \mathsf{T}^{-1} \, \mathsf{Y} \, \mathsf{T}^{-1}
  = \mathsf{Y}^{-1}
  ,
  \notag
  \\
  \mathsf{X} \, \mathsf{Y}
  = q^{-1} \, \mathsf{Y} \, \mathsf{X} \, \mathsf{T}^2
  .
\end{gather}
\end{definition}
We use an idempotent
\begin{equation}
  \mathsf{e}
  = \frac{1}{t+t^{-1}} 
  \left( 
    t +  \mathsf{T}
  \right)
  ,
\end{equation}
which satisfies
\begin{gather}
  \mathsf{e}^2 = \mathsf{e},
  \notag
  \\
  \mathsf{e} \, \mathsf{T} =
  \mathsf{T} \, \mathsf{e}
  =
  t^{-1} \mathsf{e}
  .
\end{gather}

\begin{definition}
  The spherical DAHA $\SH_{q,t}$ is
  \begin{equation*}
    \SH_{q,t}
    =
    \mathsf{e} \,  \mathcal{H}_{q,t} \,  \mathsf{e}
    .
  \end{equation*}
\end{definition}

The automorphisms of $A_1$-DAHA are listed in the following~\cite{Chered05Book}.
\begin{lemma}
  \begin{itemize}
  \item 
    An automorphism $\epsilon: \mathcal{H}_{q,t}\to
    \mathcal{H}_{q^{-1},t^{-1}}$;
    \begin{equation}
      \label{A1_epsilon}
      \epsilon:
      \left(
        \begin{matrix}
          \mathsf{X} 
          \\
          \mathsf{Y} \\
          \mathsf{T}  
        \end{matrix}
      \right)
      \mapsto
      \left(
        \begin{matrix}
          \mathsf{Y} \vphantom{X^{-1}}\\
          \mathsf{X} \vphantom{Y^{-1}}\\
          \mathsf{T}^{-1}\\
        \end{matrix}
      \right)
      .
    \end{equation}

  \item
    An anti-automorphism $\epsilon^\prime:
    \mathcal{H}_{q,t}\to\mathcal{H}_{q,t}$;
    \begin{equation}
      \epsilon^\prime:
      \begin{pmatrix}
        \mathsf{X} \\
        \mathsf{Y} \\
        \mathsf{T}
      \end{pmatrix}
      \mapsto
      \begin{pmatrix}
        \mathsf{Y}^{-1} \\
        \mathsf{X}^{-1} \\
        \mathsf{T}
      \end{pmatrix}
      .
    \end{equation}
  \end{itemize}
\end{lemma}

\begin{lemma}[\cite{Chered05Book}]
  The $SL(2;\mathbb{Z})$ action on $\mathcal{H}_{q,t}$ is generated by
  \begin{align}
    \label{tau_for_A1}
    &
      \begin{pmatrix}
        1 & 1 \\
        0 & 1
      \end{pmatrix}
            \mapsto \tau_R ,
          &
          &
            \begin{pmatrix}
              1 & 0 \\
              1 & 1
            \end{pmatrix}
                  \mapsto \tau_L ,
  \end{align}
  where
  $\tau_{\bullet}: \mathcal{H}_{q,t} \to \mathcal{H}_{q,t}$ is
  \begin{align}
    \tau_R & :
      \left(
        \begin{matrix}
          \mathsf{T} \\
          \mathsf{Y} \\
          \mathsf{X} 
        \end{matrix}
      \right)
      \mapsto
      \left(
        \begin{matrix}
          \mathsf{T} \\
          q^{\frac{1}{2}} \, \mathsf{X} \, \mathsf{Y}
          \\
          \mathsf{X}
        \end{matrix}
      \right)
    ,
        \notag
      \\[2ex]
      \tau_L &:
      \left(
        \begin{matrix}
          \mathsf{T} \\
          \mathsf{Y} \\
          \mathsf{X} 
        \end{matrix}
      \right)
      \mapsto
      \left(
        \begin{matrix}
          \mathsf{T}
          \\
          \mathsf{Y}
          \\
          q^{-\frac{1}{2}} \, \mathsf{Y} \, \mathsf{X}
        \end{matrix}
      \right)
    .
    \label{auto_A1}
  \end{align}
\end{lemma}


We note that,
with~\eqref{A1_epsilon},
we have
\begin{equation}
  \label{A1_L_epsilon}
  \tau_L = \epsilon \, \tau_R \, \epsilon
  .
\end{equation}

\subsection{Polynomial Representation and Macdonald Polynomial}
We recall a representation on ring of the Laurent polynomials
$\mathbb{C}[x^{\pm 1}]$~\cite{Chered05Book}.
We use
an involution $\mathsf{s}$
and
a $q$-difference operator
$\eth$  respectively defined by
\begin{align}
  & \mathsf{s} \, f(x) = f(x^{-1}),
  &
  &
  \eth \, f(x) = f(q\, x) ,
\end{align}
where
$f \in \mathbb{C}[x^{\pm 1}]$.

\begin{prop}
  A polynomial representation in
  $\mathbb{C}[x^{\pm 1}]$
  is given by
  \begin{align}
    \label{A1_represent}
    \mathsf{T}
    & \mapsto
      t^{-1} \mathsf{s} +
      \left( t^{-1} - t \right)
      \frac{1}{x^2-1} \,  \left(\mathsf{s}-1 \right)
      ,
    \\
    \mathsf{X}
    & \mapsto
      x
      ,
      \notag
    \\
    \mathsf{Y}
    & \mapsto
      {\eth} \,  \mathsf{s} \, \mathsf{T}
      .
      \notag
  \end{align}
\end{prop}


Here $\mathsf{Y}$ is called the Dunkl--Cherednik operator.
We see that
$\mathsf{T} f=t^{-1} f$ for the symmetric Laurent polynomials
$f\in \mathbb{C}[x+x^{-1}]$, and that
$\mathsf{e} \, \mathbb{C}[x]
= \mathbb{C}[x+x^{-1}]$.
Thus $\SH_{q,t}$ preserves a symmetric space~$\mathbb{C}[x+x^{-1}]$.
As a $q$-difference operator of $\SH_{q,t}$,
we have the following expression,
\begin{gather}
  \label{A1_Macdonald}
  \left.
    \mathsf{Y}
    +  \mathsf{Y}^{-1}
  \right|_{\text{sym}}
  \mapsto
  \frac{
    t \, x - t^{-1} x^{-1}}{
    x - x^{-1}
  } \, \eth
  + \frac{
    t^{-1} x - t \, x^{-1}}{
    x-x^{-1}
  } \,
  \eth^{-1}
  ,
\end{gather}
where $\left. h \right|_{\text{sym}}$ means that
$h \in \mathcal{H}_{q,t}$ acts on the symmetric
Laurent polynomial space $\mathbb{C}[x+x^{-1}]$.
This operator is known as
the Macdonald operator (see, \emph{e.g.},~\cite{Macdo95}).
One also finds that 
\begin{equation}
  \label{A1_rep_xyz}
  \left.
    q^{\frac{1}{2}} \mathsf{X} \, \mathsf{Y}
    + q^{-\frac{1}{2}} \mathsf{Y}^{-1} \, \mathsf{X}^{-1}
  \right|_{\text{sym}}
  \mapsto
  q^{\frac{1}{2}} x \,
  \frac{t \, x - t^{-1} x^{-1}}{
    x - x^{-1}
  } \, \eth
  +
  q^{\frac{1}{2}} x^{-1} \,
  \frac{t^{-1} x - t \, x^{-1}}{
    x - x^{-1}
  } \,
  \eth^{-1}
  .
\end{equation}


We  have the  non-symmetric Macdonald polynomials $E_m(x;q, t)$
as eigenfunctions of~$\mathsf{Y}$,
\begin{align}
  \mathsf{Y} \, E_{-m}(x; q, t)
  & = t^{-1} q^{-m} E_{-m}(x; q, t)
    ,
    \notag
  \\
  \mathsf{Y} \, E_{m}(x; q,t)
  & = t \, q^m \, E_m(x; q,t)
    .
    \label{A1_nonsymmetric}
\end{align}
Here $m>0$, and
the Laurent polynomials $E_m(x;q, t)$ have forms of
\begin{align}
  E_{-m}(x; q, t) & = x^{-m}
                    + \frac{\left(t-t^{-1} \right) q^m}{
                    t \, q^m - t^{-1} q^{-m}
                    } x^m
                    + \cdots
                    ,
                    \notag
  \\
  E_m(x; q, t)
                  & =  x^m+ \cdots
                    ,
\end{align}
where $\cdots$ means Laurent polynomials
$x^k$ with $|k|<m$.
It is noted that $E_0(x;q,t)=1$ and that
$\mathsf{Y} E_0(x;q,t)=t^{-1} E_0(x;q,t)$.

Symmetric eigenfunctions of~\eqref{A1_Macdonald} are
the  Macdonald polynomials of $A_1$-type
(also known as
the $q$-ultraspherical polynomial, or the Rogers
polynomial~\cite{Macdonald03Book}).
Explicitly,
we have
\begin{equation}
  \label{Macdonald_eigen}
  \left( \mathsf{Y} + \mathsf{Y}^{-1} \right) M_n(x; q, t)
  =
  \left(
    t \, q^n + t^{-1} q^{-n}
  \right) \,
  M_n(x; q, t)    ,
\end{equation}
where
\begin{align}
  \label{A1_Macdo_def}
  M_n(x ; q, t)
  & =
    x^n \cdot
    {}_2 \phi_1
    \left[
    \begin{matrix}
      t^2 , \quad  q^{-2n}
      \\
      t^{-2} q^{2-2n}
    \end{matrix}
  ; ~
  q^2 , q^2 \, t^{-2} x^{-2}
  \right]
  \\
  & =
    \frac{
    (q^2;q^2)_n}{
    (t^2;q^2)_n}
    \sum_{\substack{
    j,k \geq 0 \\
  j+k=n
  }}
  \frac{
  (t^2;q^2)_j (t^2;q^2)_k
  }{
  (q^2;q^2)_j (q^2;q^2)_k
  } \,
  x^{j-k}
  .
  \notag
\end{align}
Here the polynomials are normalized to be
$M_0(x;q, t)=1$ and
\begin{equation*}
  M_n(x; q, t)
  =
  \left(  x^n+ x^{-n} \right)
  +
  \cdots
  ,
\end{equation*}
for $n>0$.
The polynomials $ M_n(x;q, t) $ span the
symmetric Laurent polynomial space
$\mathbb{C}[x+x^{-1}]$.
In terms of the non-symmetric polynomials~\eqref{A1_nonsymmetric}, we have
\begin{align}
  \label{Macdo_sum_non-sym}
    M_m(x;q, t)
    & =
    E_{-m}(x;q, t)
    + \frac{ q^m -q^{-m}}{
      t^2 q^m - q^{-m}
    } \, E_m(x; q, t)
    \\
    & =
    t^{-1}  \left( \mathsf{T} + t \right)
    E_m(x;q,t)
      .
      \notag
\end{align}
Some of them are explicitly written as follows;
\begin{align}
  M_0(x; q, t)
  & = 1
    ,
    \notag
  \\
  M_1 (x; q, t )
  & = x+x^{-1}
    ,
    \label{A1_Macdo_1}
  \\
  M_2(x; q,t )
  &=
    x^2+x^{-2}
    + \frac{\left( 1+q^2 \right) \left( 1-t^2 \right)}{
    1-q^2t^2}
    ,
    \notag
  \\
  M_3(x; q, t)
  &=
    x^3
    +x^{-3}
    +
    \frac{ \left(1-q^6 \right) \left(1-t^2 \right)}{
    \left( 1-q^2 \right)  \left( 1-q^4 t^2 \right)}
    \left( x+ x^{-1} \right)
    .
    \notag
\end{align}
Note that the generating function of the $A_1$-type Macdonald polynomials is
\begin{equation}
  \label{gf_A1_polynomial}
  \sum_{n=0}^\infty
  M_n(x;q, t) \,
  \frac{
    \left( t^2; q^2 \right)_n}{
    \left( q^2;q^2 \right)_n} \,
  z^n
  =
  \frac{
    \left( t^2 x \, z , t^2 x^{-1} z ; q^2 \right)_\infty
  }{
    \left(x \, z , x^{-1} z ; q^2 \right)_\infty
  } .
\end{equation}
One sees that
the Macdonald polynomial~\eqref{A1_Macdo_def} reduces  at $q=t$
to
\begin{equation}
  \label{Macdo_q_equal_t}
  M_n(x;q,q)
  =
  \frac{
    x^{n+1}-x^{-n-1}
  }{
    x -x^{-1}
  }
  =
  S_n(x+x^{-1})
  .
\end{equation}
Here $S_n(z)$ is 
the Chebyshev polynomial of the second kind,
which is also
defined recursively by
\begin{equation}
  \label{recursion_S}
  z \, S_n(z)=S_{n+1}(z)+S_{n-1}(z)  ,
\end{equation}
with
$S_0(z)=1$, $S_1(z)=z$.


We list some identities of the $A_1$-type
Macdonald polynomials.
As a typical property of   orthogonal polynomials,
we have the three-term recurrence relation,
\begin{gather}
  \label{three-term}
  (\mathsf{X}+\mathsf{X}^{-1}) \, M_n(\mathsf{X}; q, t)
  =
  M_{n+1}(\mathsf{X}; q, t) +
  \frac{
    \left( 1-q^{2n} \right) \left(1-q^{2n-2} t^4 \right)
  }{
    \left( 1-q^{2n-2}t^2 \right)
    \left( 1-q^{2n} t^2 \right)
  } \, M_{n-1} (\mathsf{X}; q, t) .
\end{gather}
We also have
\begin{multline}
  \label{three-term_2}
  \left(
    q^{-\frac{1}{2}} \mathsf{Y} \, \mathsf{X}
    +
    q^{\frac{1}{2}} \mathsf{X}^{-1} \mathsf{Y}^{-1}
  \right) M_n( \mathsf{X}; q, t)
  \\
  =
  t \, q^{n+\frac{1}{2}}  \, M_{n+1}(\mathsf{X}; q, t)
  +
  t^{-1}  q^{-n+\frac{1}{2}} 
  \frac{
    \left( 1-q^{2n} \right)  \left( 1-q^{2n-2} t^4 \right)
  }{
    \left( 1-q^{2n-2} t^2 \right) \left(1-q^{2n} t^2 \right) 
  } \, M_{n-1}(\mathsf{X}; q, t) .
\end{multline}

For our later computations, we introduce the
raising and lowering operators of $M_n(x;q,t)$.
\begin{prop}
  \label{prop:raising_lowering}
  We have  the raising operator with a parameter shift, 
  \begin{multline}
    \label{raise}
    \left\{
      \frac{
        \left( 1 - t^2 x^2 \right) \left( 1-q^2 t^2 x^2 \right)
      }{
        q\, t^2 x \left( x^2 -1 \right)
      } \,
      \eth
      -
      \frac{
        \left( t^2-x^2 \right) \left( t^2 q^2 - x^2 \right)
      }{
        q\, t^2 x \left( x^2 -1 \right)
      } \, 
      \eth^{-1}
    \right\} M_m(x; q, q \, t)
    \\
    =
    \left( q^{m+1} t^2 - q^{-m-1}t^{-2} \right) M_{m+1}(x; q, t)
    .
  \end{multline}
  The lowering operator with a parameter shift is given by
  \begin{equation}
    \label{lower}
    \frac{x}{x^2-1}
    \left( \eth - \eth^{-1} \right) M_m(x; q, t)
    =
    \left( q^m -q ^{-m} \right) M_{m-1}(x; q, q \, t )
    .
  \end{equation}
\end{prop}
\begin{proof}
  It can be  proved by calculating actions on
  the generating function~\eqref{gf_A1_polynomial}.
  See also~\cite{KiriNoum96a,Koorn06a}.
\end{proof}

These raising and lowering operators, which 
preserve the  symmetric Laurent polynomial space
$\mathbb{C}[x+x^{-1}]$,
can be rewritten using the generators of DAHA.
For brevity,
we denote the  raising and lowering operators in
Prop.~\ref{prop:raising_lowering}
as
$\mathsf{K}^{(+)}$ and $\mathsf{K}^{(-)}$ respectively.
By use of~\eqref{A1_represent}, we have
\begin{equation*}
  \left. t \, \mathsf{Y} - t^{-1} \mathsf{Y}^{-1} \right|_{\text{sym}}
  =
  \left.
    t^{-1}  \left( t^{-1} - \mathsf{T} \right)
    \eth \right|_{\text{sym}}
  =
  t^{-1} \,
  \frac{t^{-1}x^2-t}{x^2-1} 
  \left( 1 - \mathsf{s} \right) \eth
  \bigr|_{\text{sym}}
  ,
\end{equation*}
which proves~\cite{Chered05Book,Macdonald03Book}
that the lowering operator~\eqref{lower} is written as
\begin{equation*}
  \left. \mathsf{K}^{(-)} \right|_{\text{sym}}
  =
  \frac{t}{
    t^{-1} \mathsf{X} - t \, \mathsf{X}^{-1}
  }
  \left( t \,  \mathsf{Y} - t^{-1} \mathsf{Y}^{-1} \right)
  \Bigr|_{\text{sym}} .
\end{equation*}
Note that $\mathsf{K}^{(-)}$ does not depend on $t$  as operators
on~$\mathbb{C}[x+x^{-1}]$.

Combining the identities~\eqref{raise} and~\eqref{lower}, we have
\begin{equation*}
  \mathsf{K}^{(+)} \mathsf{K}^{(-)}  \, M_m(x;q,t)
  =
  \left\{
  \left( \mathsf{Y} + \mathsf{Y}^{-1} \right)^2
  -
  \left(t+t^{-1}\right)^2
\right\} M_m(x;q,t)
.
\end{equation*}
Using the above expression for $\mathsf{K}^{(-)}$
and the fact that
the Macdonald polynomials $M_m(x;q,t)$ are bases of
$\mathbb{C}[x+x^{-1}]$,
we find
\begin{equation*}
  \left. \mathsf{K}^{(+)} \right|_{\text{sym}}
  =
  \left.
    t^{-1}
    \left( t^{-1} \mathsf{Y} - t \, \mathsf{Y}^{-1} \right)
    \left(
      t^{-1} \mathsf{X} - t \, \mathsf{X}^{-1}
    \right)
  \right|_{\text{sym}}
  .
\end{equation*}
To conclude, we have the following.
We recall that $\sh(x)$ is defined in~\eqref{hyperbolic}.
\begin{prop}
  Both the raising and lowering operators preserve the symmetric Laurent
  polynomial space
  $\mathbb{C}[x+x^{-1}]$, and they are written as
  \begin{gather}
    t^{-1} \sh \left( t^{-1}\mathsf{Y} \right)
    \sh \left( t^{-1}\mathsf{X} \right) \,
    M_m(x;q, q\, t)
    =
    \left( q^{m+1} t^2 - q^{-m-1} t^{-2} \right)
    M_{m+1}(x; q, t)
    ,
    \notag
    \\
    \frac{t}{
      \sh \left( t^{-1}\mathsf{X} \right)
    } \, \sh(t \, \mathsf{Y}) \,
    M_m(x;q, t)
    =
    \left( q^m - q^{-m} \right)
    M_{m-1}(x; q, q\, t )
    .
    \label{lower_raise}
  \end{gather}
\end{prop}




\subsection{Automorphisms as Conjugation}
We revisit the $SL(2;\mathbb{Z})$ action~\eqref{auto_A1} in the
polynomial representation of DAHA.
As a  completion of DAHA~\cite{Chered05Book}, we introduce a
function
\begin{equation*}
  U_R=\exp\left( \frac{ \left( \log \mathsf{X} \right)^2}{
      2 \log q
    }
  \right)
  .
\end{equation*}
As $U_R$ is symmetric in $\mathsf{X} \leftrightarrow \mathsf{X}^{-1}$ and
$\mathsf{s} \, U_R = U_R \, \mathsf{s}$, it commutes with~$\mathsf{T}$.
One easily sees that
$
  \eth \, U_R=
  q^{\frac{1}{2}} \mathsf{X} \, U_R \, \eth
$,
and we obtain
\begin{equation*}
  \mathsf{Y}\, U_R
  =
  \eth \, \mathsf{s} \, \mathsf{T} \, U_R
  =
  q^{\frac{1}{2}} \mathsf{X} \, U_R \, \eth \, \mathsf{s} \,
  \mathsf{T}
  =
  U_R \, q^{\frac{1}{2}} \mathsf{X} \, \mathsf{Y}
  .
\end{equation*}
Trivial is a commutativity between $U_R$ and $\mathsf{X}$, and
thus,
the automorphism~$\tau_R$~\eqref{auto_A1} is identified with
a conjugation by $U_R$.
For $\tau_L$~\eqref{auto_A1}, we recall~\eqref{A1_L_epsilon} to
find the following~$U_L$.

\begin{prop}
  The $SL(2;\mathbb{Z})$ action on $\SH_{q,t}$ is given as
  conjugation.
  In particular,
  the automorphisms $\tau_{\bullet}$~\eqref{tau_for_A1} are written as
  conjugations
  \begin{equation}
    \label{A1_conjugation}
    \tau_{\bullet}: h \mapsto U_{\bullet}^{~-1} \, h \, U_{\bullet}
    ,
  \end{equation}
  where
  \begin{align}
    U_R & =
          \exp \left(
          \frac{ ( \log \mathsf{X} )^2}{2 \log q}
          \right) ,
          &
    U_L & =
          \exp \left(
          -
          \frac{ ( \log \mathsf{Y} )^2}{2 \log q}
          \right) .
        \end{align}
\end{prop}

See~\cite{DiFraKedem17a} where the operators $U_\bullet$ were
introduced for DAHA of  $A_n$-type.

\subsection{Shift Operator}
As seen from the raising~\eqref{raise} and the lowering
operators~\eqref{lower},
it is useful to  introduce a parameter shift operator $\eth_t$
in $\SH_{q,t}$
satisfying \begin{equation}
  \label{shift_t_q}
  \eth_t t= q \, t  \, \eth_t
  .
\end{equation}
The $SL(2;\mathbb{Z})$ actions $\tau_\bullet$~\eqref{tau_for_A1} on
$t$ are trivial,
but 
we have the following action on $\eth_t$.
\begin{prop}
  We have
  $\tau_R: \eth_t \mapsto \eth_t$,
  and
  \begin{equation}
    \label{partial_t_Tau}
    \tau_L
    :
    \eth_t
    \mapsto
    \frac{  1}{
      \sh \left(t^{-1} q^{-\frac{1}{2}}  \mathsf{Y} \, \mathsf{X} \right)
    }
    \sh (t^{-1}\mathsf{X}) \,
    \eth_t
    .
  \end{equation}
\end{prop}

\begin{proof}
  The conjugation~\eqref{A1_conjugation} shows
  an invariance of~$\eth_t$ under~$\tau_R$.
  For $\tau_L$, we use the fact that
  the Macdonald polynomials span the  symmetric polynomial space $\mathbb{C}[x+x^{-1}]$,
  and
  we compute actions on $M_m(x;q, t)$.
  Recalling that $M_m(x;,q,t)$ is a sum of the non-symmetric
  polynomials~\eqref{Macdo_sum_non-sym}
  and that the operator $U_L$ is symmetric in
  $\mathsf{Y}  \leftrightarrow \mathsf{Y}^{-1}$,
  we have the following equalities;
  \begin{align*}
    & U_L^{~-1} \eth_t \, U_L \, M_m(x; q, t)
    \\
    & \overset{\text{\eqref{A1_nonsymmetric}}}{=}
      U_L^{~-1} \eth_t \,
      \E^{ - \frac{ \left( \log \left( t \, q^m \right) \right)^2}{
      2 \log q
      }
      }
      M_m(x;q, t)
      =
      U_L^{~-1}
      \E^{ - \frac{ \left(
      \log \left( t \, q^{m+1} \right) \right)^2}{
      2 \log q
      }
      } \, M_m(x;q, q\, t)
    \\
    & \overset{\text{\eqref{lower_raise}}}{=}
      \E^{ - \frac{ \left( \log \left( t \, q^{m+1} \right) \right)^2}{
      2 \log q
      }
      }
      U_L^{~-1}
      \frac{t}{
      \sh(t^{-1} \mathsf{X})
      } \,
      \frac{1}{
      \sh(t^{-1} \mathsf{Y})
      }
      \left( q^{m+1} t^2 - q^{-m-1} t^{-2}\right)
      M_{m+1}(x; q, t)
    \\
    & \overset{\text{\eqref{auto_A1}}}{=}
      \E^{ - \frac{
      \left( \log  \left( t \, q^{m+1} \right) \right)^2}{
      2 \log q}
      }
      \frac{t}{
      \sh 
      \left(t^{-1} q^{-\frac{1}{2}}\mathsf{Y}\mathsf{X} \right)
      } \,
      \frac{1}{
      \sh(t^{-1} \mathsf{Y} )
      }
      U_L^{~-1}
      \left( q^{m+1} t^2 - q^{-m-1} t^{-2}\right)
      M_{m+1}(x; q, t)
    \\
    & \overset{\text{\eqref{A1_nonsymmetric}}}{=}
      \frac{t}{
      \sh
      \left(t^{-1} q^{-\frac{1}{2}} \mathsf{Y} \, \mathsf{X} \right)
      } \,
      \frac{1}{
      \sh \left( t^{-1} \mathsf{Y} \right)
      }
      \left( q^{m+1} t^2 - q^{-m-1} t^{-2}\right)
      M_{m+1}(x; q, t)
    \\
    & \overset{\text{\eqref{lower_raise}}}{=}
      \frac{1}{
      \sh \left(t^{-1} q^{-\frac{1}{2}} \mathsf{Y} \, \mathsf{X} \right)
      } \,
      \sh \left( t^{-1}\mathsf{X} \right) \,
      M_m(x;q, q\, t) .
  \end{align*}
  This proves the statement.
\end{proof}

\subsection{Algebra Embedding}

We shall  give a DAHA representation
for $\KBS_A(\Sigma_{1,1})$ defined in~\eqref{torus_xy_z}.
For a simple closed curve $\mathbb{c}_{(r,s)}$ with slope $s/r$ on
$\Sigma_{1,1}$,
we assign DAHA operators as
\begin{equation}
  \label{torus_c_and_M}
  \mathbb{c}_{(r,s)} \mapsto  \mathcal{M}_{(r,s)}
  =
  \ch \left( \mathcal{O}_{(r,s)} \right)
  \in \SH_{q,t}
  ,
\end{equation}
where we recall that  $\ch(x)$ is defined in~\eqref{hyperbolic}.
The
curves in 
Fig.~\ref{fig:xyz_torus} are identified as
$\mathbb{c}_{(1,0)}=\mathbb{x}$,
$\mathbb{c}_{(0,1)}=\mathbb{y}$, and $\mathbb{c}_{(1,1)}=\mathbb{z}$.
The algebra $\KBS_A(\Sigma_{1,1})$
was studied in detail in~\cite{FrohGelc00a}
(see also~\cite{BaMuPrWiWa18a}),
and
known is
a ``product-to-sum formula'' for  $\mathbb{c}_{(r,s)}$.
For  example, we can check easily the following skein algebra;
\begin{gather}
  \mathbb{c}_{(1,0)} \, \mathbb{c}_{(0,1)}
  =
  A \,  \mathbb{c}_{(1,1)} +
  A^{-1} \mathbb{c}_{(1,-1)} ,
  \notag
  \\
  \mathbb{c}_{(0,1)} \, \mathbb{c}_{(1,1)}
  =
  A \, \mathbb{c}_{(1,0)} +
  A^{-1} \mathbb{c}_{(1,2)} ,
  \notag
  \\
  \mathbb{c}_{(1,1)} \, \mathbb{c}_{(1,0)}
  =
  A \,  \mathbb{c}_{(0,1)} +
  A^{-1} \mathbb{c}_{(2,1)}
  .
  \label{skein_torus_1001}
\end{gather}
Here
we put
\begin{align}
  \mathcal{M}_{(1,0)}
  & =
    \ch (\mathsf{X} )
    ,
    \notag
  \\
  \mathcal{M}_{(0,1)}
  & =
    \ch (\mathsf{Y})
    ,
    \notag
  \\
  \mathcal{M}_{(1,1)}
  & =
    \ch \left( q^{\frac{1}{2}} \mathsf{X} \, \mathsf{Y} \right)
    =
    \ch \left( q^{-\frac{1}{2}} \mathsf{Y} \, \mathsf{X} \right)
    .
  \label{define_M_xyz}
\end{align}
Recall that $\mathcal{M}_{(0,1)}$ is the Macdonald operator of $\SH_{q,t}$, and
see~\eqref{A1_Macdonald} and~\eqref{A1_rep_xyz} for explicit forms of the
$q$-difference
operators on the symmetric  Laurent polynomial space
$\mathbb{C}[x+x^{-1}]$.
By direct computations, we get a representation for~\eqref{skein_torus_1001}
\begin{gather}
  \mathcal{M}_{(1,0)} \, \mathcal{M}_{(0,1)}
  =
  q^{-\frac{1}{2}} \mathcal{M}_{(1,1)} +
  q^{\frac{1}{2}} \mathcal{M}_{(1,-1)} ,
  \notag
  \\
  \mathcal{M}_{(0,1)} \, \mathcal{M}_{(1,1)}
  =
  q^{-\frac{1}{2}} \mathcal{M}_{(1,0)} +
  q^{\frac{1}{2}} \mathcal{M}_{(1,2)} ,
  \notag
  \\
  \mathcal{M}_{(1,1)} \, \mathcal{M}_{(1,0)}
  =
  q^{-\frac{1}{2}} \mathcal{M}_{(0,1)} +
  q^{\frac{1}{2}} \mathcal{M}_{(2,1)}
  ,
  \label{mm_prod_1}
\end{gather}
where we have defined
\begin{gather}
  \mathcal{M}_{(1,-1)}
  =
  \ch \left( q^{-\frac{1}{2}} \mathsf{X}^{-1} \mathsf{Y} \right)
  ,
  \notag
  \\
  \mathcal{M}_{(1,2)}
  =
  \ch \left( \mathsf{Y}  \, \mathsf{X} \, \mathsf{Y} \right)
  ,
  \notag
  \\
  \mathcal{M}_{(2,1)}
  =
  \ch \left( \mathsf{X} \, \mathsf{Y} \, \mathsf{X} \right)
  .
\end{gather}
Using these relations, we can check~\eqref{torus_xy_z} in
Prop.~\ref{prop:Sigma11}, and we obtain the following.
\begin{theorem}
  \label{thm:embed_A1}
  We have an algebra embedding
  $\KBS_{A}(\Sigma_{1,1}) \to SH_{q,t}$
  with $A=q^{-\frac{1}{2}}$ by
  \begin{equation}
    \label{torus_daha}
    \left(
      \begin{matrix}
        \mathbb{x} \\
        \mathbb{y} \\
        \mathbb{z} 
      \end{matrix}
    \right)
    \mapsto
    \left(
      \begin{matrix}
        \mathcal{M}_{(1,0)} \\
        \mathcal{M}_{(0,1)} \\
        \mathcal{M}_{(1,1)} 
      \end{matrix}
    \right)
    .
  \end{equation}
\end{theorem}


We can also check that
\begin{equation}
  \label{mm_prod_3}
  \mathcal{M}_{(1,-1)}  \, \mathcal{M}_{(1,1)}
  =
  q^{-1} \, \mathcal{M}_{(2,0)_T}
  +
  q \, \mathcal{M}_{(0,2)_T}
  +
  q+q^{-1} -
  \left(
    t^2 q^{-1} + t^{-2} q
  \right)
  ,
\end{equation}
where we follow the notation of~\cite{FrohGelc00a},
\begin{align*}
  \mathcal{M}_{(2,0)_T}
    & =
    T_2(\mathcal{M}_{(1,0)})
    =
    \mathsf{X}^2 + \mathsf{X}^{-2}
    ,
    \\
    \mathcal{M}_{(0,2)_T}
    & =
    T_2( \mathcal{M}_{(0,1)} )
    =
    \mathsf{Y}^2 + \mathsf{Y}^{-2}
      .
\end{align*}
Here
$T_n(z)$ denotes the Chebyshev polynomial of the first kind
defined by
\begin{equation}
  \label{Chebysheff}
  T_n(x+x^{-1})=x^n+x^{-n},
\end{equation}
which satisfies the same recurrence equation~\eqref{recursion_S} with
the second kind polynomial.
From the viewpoint of the skein algebra, the last term
in~\eqref{mm_prod_3}
denotes  the boundary circle
$\mathbb{b}$ in
Fig.~\ref{fig:xyz_torus},
\begin{equation}
  \label{A1_b_and_t}
  \mathbb{b} \mapsto
  -t^2 \, q^{-1} - t^{-2} \, q
  .
\end{equation}
Indeed the identities~\eqref{mm_prod_1} and~\eqref{mm_prod_3}
give
\begin{multline}
  q^{-\frac{1}{2}}
  \mathcal{M}_{(1,0)}   \mathcal{M}_{(0,1)}   \mathcal{M}_{(1,1)}
  \\
  =
  q^{-1}   \mathcal{M}_{(1,0)}^{~2}  +
  q \, \mathcal{M}_{(0,1)}^{~2} +
  q^{-1} \mathcal{M}_{(1,1)}^{~2}
  + \left( -q -q^{-1} \right)
  +\left(- t^2 \,  q^{-1}-t^{-2} \, q \right)
  .
\end{multline}
Recalling the skein algebra relation~\eqref{boundary_b}, we see that the
embedding~\eqref{torus_daha}
induces~\eqref{A1_b_and_t}.

Actions of $\mathcal{M}_{(r,s)}$  on
the Macdonald polynomial $M_n(x;q, t)$ as a basis of
the symmetric Laurent polynomials $\mathbb{C}[x+x^{-1}]$
can be
computed explicitly in principle.
The Macdonald polynomials are eigenpolynomials of
$\mathcal{M}_{(1,0)}=\mathsf{Y}+\mathsf{Y}^{-1}$ as
in~\eqref{Macdonald_eigen}, and
the three-term  recurrence relations~\eqref{three-term}
and~\eqref{three-term_2}  denote respectively
the actions of  $\mathcal{M}_{(1,0)}$ and
$\mathcal{M}_{(1,1)}$.
These relations  reduce to results of representation of $\KBS_A(\Sigma_{1,1})$~\cite{ChoGel14a,MarcPaul15a}
when  $q$ is a root of unity.

Other operators $\mathcal{M}_{(r,s)}$ can be  given explicitly using
the $SL(2;\mathbb{Z})$-action  of $\SH_{q,t}$.
For instance,
when we apply the automorphisms~\eqref{auto_A1} to~\eqref{mm_prod_1},
we get
\begin{gather}
  \label{mm_prod_2}
  \mathcal{M}_{(n,1)} \, \mathcal{M}_{(1,0)}
  =
  q^{-\frac{1}{2}} \mathcal{M}_{(n-1,1)}
  + q^{\frac{1}{2}} \mathcal{M}_{(n+1,1)}
  ,
  \notag
  \\
  \mathcal{M}_{(0,1)} \, \mathcal{M}_{(1,n)}
  =
  q^{-\frac{1}{2}} \mathcal{M}_{(1,n-1)}
  +
  q^{\frac{1}{2}} \mathcal{M}_{(1,n+1)}
  ,
\end{gather}
where
for $n \geq 1$
\begin{gather}
  \mathcal{M}_{(n,1)}
  =
  \ch \left(
    q^{\frac{n}{2}-1} \mathsf{X}^{n-1} \mathsf{Y} \, \mathsf{X}
  \right)
  ,
  \notag
  \\
  \mathcal{M}_{(1,n)}
  =
  \ch \left(
    q^{-\frac{n}{2}+1} \mathsf{Y}^{n-1} \mathsf{X} \, \mathsf{Y}
  \right)
  .
\end{gather}

With our algebra embedding, the
$SL(2;\mathbb{Z})$ actions~\eqref{auto_A1} of DAHA
naturally induce the mapping class group
$\Mod(\Sigma_{1,1}) \cong SL(2;\mathbb{Z})$
(see, \emph{e.g.},~\cite{Birman74,FarbMarg11Book}),
which is generated by the Dehn twists $\mathscr{T}_{\mathbb{x}}$
(resp. $\mathscr{T}_{\mathbb{y}}$)
about the curve
$\mathbb{x}$
(resp. $\mathbb{y}$).
As seen from the fact that
$\tau_R: \mathbb{c}_{(0,1)} \mapsto \mathbb{c}_{(1,1)}$,
$\tau_R: \mathbb{c}_{(1,0)} \mapsto \mathbb{c}_{(1,0)}$, and that the
boundary circle $\mathbb{b}$ is fixed,
the automorphism
$\tau_R$~\eqref{tau_for_A1}
denotes the right Dehn twist
$\mathscr{T}_{\mathbb{x}}^{~-1}$
about $\mathbb{x}=\mathbb{c}_{(1,0)}$.
In the same manner, as we have
$\tau_L: \mathbb{c}_{(0,1)} \mapsto \mathbb{c}_{(0,1)}$
and
$\tau_L: \mathbb{c}_{(1,0)} \mapsto \mathbb{c}_{(1,1)}$,
we can identify
the automorphism $\tau_L$~\eqref{tau_for_A1}
with  the  left Dehn twist
$\mathscr{T}_{\mathbb{y}}$
about
$\mathbb{y}=\mathbb{c}_{(0,1)}$.
For our later use, we summarize the actions of
the Dehn twists  as follows.
\begin{align}
  & \mathscr{T}_{\mathbb{y}} \mapsto
  \tau_L,
  &
  &
    \tau_L^{~ \pm 1}
      :
    \left(
    \begin{matrix}
      \mathsf{T}
      \\
      \mathsf{X}
      \\
      \mathsf{Y}
      \\
      \eth_t
    \end{matrix}
  \right)
  \mapsto
  \left(
  \begin{matrix}
    \mathsf{T}
    \\
    q^{\mp \frac{1}{2}} \mathsf{Y}^{\pm 1} \, \mathsf{X}
    \\
    \mathsf{Y}
    \\
    \frac{1}{\sh \left(t^{-1} q^{\mp \frac{1}{2}} \mathsf{Y}^{\pm 1} \mathsf{X}
      \right)}
    \sh \left( t^{-1} \mathsf{X} \right) \, \eth_t
  \end{matrix}
  \right),
  \notag
  \\
  &
    \mathscr{T}_{\mathbb{x}}
    \mapsto
    \tau_R^{~ - 1} ,
  &
  &
  \tau_R^{~ \pm 1}
    :
    \left(
    \begin{matrix}
      \mathsf{T}
      \\
      \mathsf{X}
      \\
      \mathsf{Y}
      \\
      \eth_t
    \end{matrix}
  \right)
  \mapsto
  \left(
  \begin{matrix}
    \mathsf{T}
    \\
    \mathsf{X}
    \\
    q^{\pm \frac{1}{2}}\mathsf{X}^{\pm 1} \mathsf{Y}
    \\
    \eth_t
  \end{matrix}
  \right)
  .
  \label{auto_and_Dehn}
\end{align}

\subsection{DAHA Polynomial}


The $N$-colored Jones polynomial of knot $K$
is a linear combination of 
the Kauffman bracket polynomials for $k$ parallel copies of knot $K$
($1 \leq k \leq N-1$), so that the polynomial for
an unknot is
$(-1)^{N-1}\frac{A^{2N}-A^{-2N}}{A^2-A^{-2}}$.
We should recall that this is the Chebyshev polynomial of the second kind.
A case of $N=2$ is the Jones polynomial.
As
a simple closed curve $\mathbb{c}_{(r,s)}$
with coprime integers
$(r,s)$
on the genus-one
Heegaard surface in $S^3$
denotes a torus knot,
the DAHA operator $\mathcal{M}_{(r,s)}$~\eqref{torus_c_and_M} associated to
$\mathbb{c}_{(r,s)}$
is expected to be related with
the (colored) Jones polynomial  of the torus knot.
Following Cherednik~\cite{IChered13a,IChered16a},
we define
\begin{equation}
  \label{P_poly_A1}
  P_n
  (x, q,t ; \mathbb{c}_{(r,s)})
  =
  M_{n-1}(\mathcal{O}_{(r,s)} ; q, t)
  (1) .
\end{equation}
Here $\mathcal{O}_{(r,s)}= \tau(\mathsf{X})$ where
the $SL(2;\mathbb{Z})$-action is
$\mathscr{T}\mapsto \tau$~\eqref{auto_and_Dehn}
when the curve is $\mathbb{c}_{(r,s)} = \mathscr{T}(\mathbb{c}_{(1,0)})$.
Note that the case of $n=2$ is simply given from~\eqref{A1_Macdo_1} as
\begin{equation}
  \begin{aligned}[t]
    P_2(x,q,t;\mathbb{c}_{(r,s)})
    &=
    \mathcal{M}_{(r,s)}(1)
    \\
    &=
    \ch ( \mathcal{O}_{(r,s)}) (1) .
  \end{aligned}
\end{equation}

For instance,
we have
\begin{align}
  P_n(x, q,t;\mathbb{c}_{(1,0)})
  &= M_{n-1}(x ; q, t) ,
    \notag
  \\
  P_n(x, q,t;\mathbb{c}_{(0,1)})
  & =
    M_{n-1}(t^{-1}; q, t) .
    \label{A1_for_c01}
\end{align}
As
the Macdonald polynomial reduces
to the Chebyshev polynomial of the second kind~\eqref{Macdo_q_equal_t} at a specific setting $q=t$,
both  polynomials $P_n(x, q, t; \mathbb{c}_{(r,s)})$
for $(r,s)=(0,1)$,~$(1,0)$
are regarded as a deformation of  the $n$-colored Jones
polynomial
for unknot.

We show an explicit result for the curve $\mathbb{c}_{(2k+1,2)}=
\left( \mathscr{T}_{\mathbb{x}}^{~ -k} \circ \mathscr{T}_{\mathbb{y}}^{~2} \right)(\mathbb{c}_{(1,0)})$.
We have
\begin{equation}
  \label{tau_tau_X}
  \mathcal{O}_{(2k+1,2)}
  =
  \left( \tau_R^{~k} \circ \tau_L^2 \right)(\mathsf{X})
  = q^{k-1} (
    \mathsf{X}^k \mathsf{Y}  )^2 \mathsf{X} ,
\end{equation}
and
the DAHA operator associated to  $\mathbb{c}_{(2k+1,2)}$ is given by
\begin{equation*}
  \mathcal{M}_{(2k+1,2)}
  =
  \ch \left(
    q^{k-1}
    ( \mathsf{X}^k \, \mathsf{Y} )^2 \mathsf{X}
  \right)
  .
\end{equation*}
This can be written as the operator
on  the symmetric polynomials
$\mathbb{C}[x+x^{-1}]$ as
\begin{multline}
  \label{M0_2k1}
  \mathcal{M}_{(2k+1,2)}
  \mapsto
  \widehat{M}_{(2k+1,2)}^{(0)}(x; q, t)
  \\
  =
  \left(q \, x \right)^{2k+1}
  \frac{
    \left( 1-t^2 x^2 \right)
    \left(1-q^2t^2x^2 \right)
  }{
    t^2 \left(1-x^2 \right)
    \left(1-q^2x^2\right)
  } \, \eth^2
  +
  \left(q \, x^{-1}\right)^{2k+1}
  \frac{
    \left( t^2-x^2 \right)
    \left(q^2t^2-x^2\right)}{
    t^2 
    \left(1-x^2\right) \left(q^2-x^2 \right)
  } \, \eth^{-2}
  \\
  -
  q \, t^{-2} \left(q^2-t^2 \right) \left(1-t^2 \right)
  \frac{
    x \left(1+x^2 \right)
  }{
     \left( q^2 - x^2 \right)
    \left( 1- q^2 x^2 \right)
  }
  .
\end{multline}
From this expression
we obtain
\begin{equation}
  \label{P2_and_torus_knot}
  \begin{aligned}[t]
    P_2(-q, q,-q; \mathbb{c}_{(2k+1,2)})
    & =
    \left.
      \widehat{M}_{(2k+1,2)}^{(0)}(x;q,t)(1)
    \right|_{x=t=-q}
    \\
    &    =
    1-q^{4k}-q^{4k+2}-
    q^{4k+4}
    \\
    & =
    - q^{6k+3} \left( q+q^{-1} \right) 
    J_2(q^2; T_{(2k+1,2)})
    ,
  \end{aligned}
\end{equation}
where $J_N(q; T_{(s,t)})$ is the colored Jones polynomial for torus
knot~\cite{Mort95a}
normalized to be
$J_N(q; \text{unknot})=1$,
\begin{equation}
  \label{cJones_torus}
  J_N(q; T_{(s,t)})
  =
  \frac{q^{\frac{1}{4} s t (1-N^2)}}{
    q^{\frac{N}{2}} - q^{-\frac{N}{2}}
  }
  \sum_{r=-\frac{N-1}{2}}^{\frac{N-1}{2}}
  \left(
    q^{s t r^2 - (s+t) r + \frac{1}{2}}
    -
    q^{s t r^2-(s-t) r - \frac{1}{2}}
  \right) .
\end{equation}

We note that
the three-term recurrence relation~\eqref{three-term} for the Macdonald
polynomial gives a recursion relation for the DAHA polynomial,
\begin{multline}
  P_{n+1}(x,q,t; \mathbf{c}_{(2k+1,2)})
  =
  \widehat{M}_{(2k+1,2)}^{(0)}(x;q,t) \left(
    P_n(x,q,t;\mathbf{c}_{(2k+1,2)})
  \right)
  \\
  -
  \frac{
    \left( 1-q^{2n-2} \right) \left(1-q^{2n-4} t^4 \right)
  }{
    \left( 1-q^{2n-4}t^2 \right)
    \left( 1-q^{2n-2} t^2 \right)
  } \,
  P_{n-1} (x, q, t ; \mathbf{c}_{(2k+1,2)})
  .
\end{multline}
We find  that this agrees with   the colored Jones polynomial  up to
framing factor at a
specific point,
\begin{equation}
  P_m(-q,q,-q; \mathbb{c}_{(2k+1,2)})
  =
  q^{(2k+1)(m^2-1)}
  \frac{q^m - q^{-m}}{q- q^{-1}}
  J_m(q^2; T_{(2k+1,2)}) .
\end{equation}

See~\cite{IChered13a,IChered16a} for further computations of the
$A_n$
DAHA polynomials.
\section{$4$-Punctured Sphere}
\label{sec:sphere}
\subsection{Skein Algebra}
We set simple closed curves $\mathbb{x}$, $\mathbb{y}$,
$\mathbb{z}$, and $\mathbb{b}_j$ on a  $4$-punctured sphere
$\Sigma_{0,4}$ as in Fig.~\ref{fig:4-hole_sphere}.
The boundary circles $\mathbb{b}_j$ of the punctures are central.

\begin{figure}[htbp]
  \centering
    \includegraphics[scale=1]{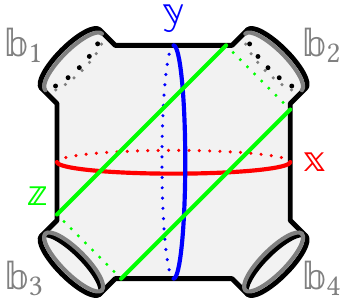}
  \caption{Depicted are simple closed curves on the 4-punctured sphere
  $\Sigma_{0,4}$.}
  \label{fig:4-hole_sphere}
\end{figure}

\begin{prop}[\cite{BulloPrzyt99a}]  
  The Kauffman bracket  skein module
  $\KBS_A(\Sigma_{0,4})$ is generated by
  $\mathbb{x}$, $\mathbb{y}$, $\mathbb{z}$, and $\mathbb{b}_j$ satisfying
  \begin{align}
    A^2 \, 
    \mathbb{x} \,  \mathbb{y}
    -A^{-2} \,
    \mathbb{y}  \,  \mathbb{x}
    & =
      \left( A^4-A^{-4} \right) \,
      \mathbb{z}
      + \left( A^2-A^{-2} \right) \,
      \left(
      \mathbb{b}_2 \,  \mathbb{b}_3
      +
      \mathbb{b}_1 \, \mathbb{b}_4
      \right)
      ,
      \notag
    \\
    A^2 \, 
    \mathbb{y} \,  \mathbb{z}
    -A^{-2} \,
    \mathbb{z}  \,  \mathbb{y}
    & =
      \left( A^4-A^{-4} \right) \,
      \mathbb{x}
      + \left( A^2-A^{-2} \right) \,
      \left(
      \mathbb{b}_1 \,  \mathbb{b}_2
      +
      \mathbb{b}_3 \, \mathbb{b}_4
      \right)
      ,
      \notag
    \\
    A^2 \, 
    \mathbb{z} \,  \mathbb{x}
    -A^{-2} \,
    \mathbb{x}  \,  \mathbb{z}
    & =
      \left( A^4-A^{-4} \right) \,
      \mathbb{y}
      + \left( A^2-A^{-2} \right) \,
      \left(
      \mathbb{b}_1 \,  \mathbb{b}_3
      +
      \mathbb{b}_2 \, \mathbb{b}_4
      \right)
      ,
      \label{algS04}
  \end{align}
  with
  \begin{multline}
    \label{center_x_b}
    A^2 \, \mathbb{x} \, \mathbb{y} \, \mathbb{z} =
    A^4 \,  \mathbb{x}^{2}
    + A^{-4} \, \mathbb{y}^{2}
    + A^4 \, \mathbb{z}^{2}
    \\
    +A^2  \,
    \left( \mathbb{b}_1\, \mathbb{b}_2
      + \mathbb{b}_3  \, \mathbb{b}_4 \right)
    \, \mathbb{x}
    +A^{-2} \,
    \left( \mathbb{b}_1\, \mathbb{b}_3
      + \mathbb{b}_2  \, \mathbb{b}_4 \right)
    \, \mathbb{y}
    +A^2 \,
    \left( \mathbb{b}_1\, \mathbb{b}_4
      + \mathbb{b}_2  \, \mathbb{b}_3 \right)
    \, \mathbb{z}
    \\
    +
    \mathbb{b}_1^{~2}
    +   \mathbb{b}_2^{~2}+  \mathbb{b}_3^{~2}+  \mathbb{b}_4^{~2}
    +
    \mathbb{b}_1\,  \mathbb{b}_2\,  \mathbb{b}_3\,  \mathbb{b}_4
    -\left( A^2+A^{-2} \right)^2
    .
  \end{multline}
\end{prop}

\begin{figure}[htbp]
  \centering
  \includegraphics[scale=.75]{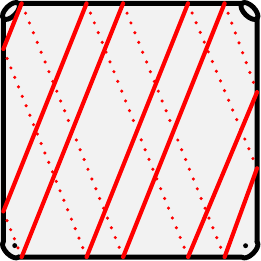}
  \caption[]{A simple closed curve with slope $5/2$ on $\Sigma_{0,4}$
    is given.
    Here punctures are on corners of the square.}
  \label{fig:slope4sphere}
\end{figure}
Essential simple closed curves on $\Sigma_{0,4}$ are parameterized by a
slope $\mathbb{Q} \cup \{\infty\}$.
Curves $\mathbb{x}$, $\mathbb{y}$, and $\mathbb{z}$ are
identified
respectively
with curves $\mathbb{c}_{(1,0)}$,
$\mathbb{c}_{(0,1)}$, and
$\mathbb{c}_{(1,1)}$,
where
$\mathbb{c}_{(r,s)}$ with coprime integers $r$ and $s$
denotes a simple closed curve
with a slope $s/r$
(see Fig.~\ref{fig:slope4sphere}).
The  multiplicative structure of these curves were
investigated in detail~\cite{BaMuPrWiWa18a}, and
a ``product-to-sum
formula''~\cite{FrohGelc00a} was given.
For example, one finds the following skein relation;
\begin{gather}
  \mathbb{c}_{(1,0)} \, \mathbb{c}_{(0,1)}
  =
  A^2 \mathbb{c}_{(1,1)} + A^{-2} \mathbb{c}_{(1,-1)}
    + \left(
      \mathbb{b}_1 \,     \mathbb{b}_4
      +
      \mathbb{b}_2 \,     \mathbb{b}_3
    \right)
    ,
    \notag
    \\
    \mathbb{c}_{(1,0)} \, \mathbb{c}_{(1,1)}
    =
    A^2 \mathbb{c}_{(2,1)} + A^{-2} \mathbb{c}_{(0,1)}
    + \left(
      \mathbb{b}_1 \,     \mathbb{b}_3
      +
      \mathbb{b}_2 \,     \mathbb{b}_4
    \right)
    ,
    \notag
    \\
    \mathbb{c}_{(1,1)} \, \mathbb{c}_{(0,1)}
    =
    A^2 \mathbb{c}_{(1,2)} + A^{-2} \mathbb{c}_{(1,0)}
    + \left(
      \mathbb{b}_1 \,     \mathbb{b}_2
      +
      \mathbb{b}_3 \,     \mathbb{b}_4
    \right)
    .
  \end{gather}

\subsection{$C^\vee C_1$-DAHA}
A generalization of the rank-1 DAHA is known as  DAHA of
$C^\vee C_1$-type (or the Askey--Wilson type),
which has four parameters besides $q$.

\begin{definition}
  The DAHA of $C^\vee C_1$-type
  $\mathcal{H}_{q, t_0, t_1,t_2, t_3}$
  is generated by
  $\mathsf{T}_0^{~\pm 1}$,
  $\mathsf{T}_1^{~\pm 1}$,
  $\mathsf{T}_0^{\vee~\pm 1}$,
  and
  $\mathsf{T}_1^{\vee~\pm 1}$,
  satisfying
  \begin{gather}
    \left( \mathsf{T}_0 - {t_0}^{-1} \right) \,
      \left( \mathsf{T}_0 + t_0  \right)
      =0
      ,
      \notag
      \\
      \left( \mathsf{T}_1 - {t_1}^{-1} \right) \,
      \left( \mathsf{T}_1 +t_1 \right)
      =0
      ,
      \notag
      \\
      \left( \mathsf{T}_0^\vee - {t_2}^{-1} \right) \,
      \left( \mathsf{T}_0^\vee  + t_2 \right)
      =0
      ,
      \notag
      \\
      \left(  \mathsf{T}_1^{\vee}  - {t_3}^{-1} \right)
      \left( \mathsf{T}_1^{\vee}+t_3\right)
      =0
      ,
    \end{gather}
  and
  \begin{equation}
    \mathsf{T}_1^\vee \, \mathsf{T}_1 \, \mathsf{T}_0 \,
    \mathsf{T}_0^\vee = q^{-\frac{1}{2}}
    .
  \end{equation}
\end{definition}

We use an idempotent
\begin{equation}
  \mathsf{e}
  =
  \frac{1}{t_1+t_1^{~-1}} \left(t_1 + \mathsf{T}_1\right)
  .
\end{equation}
satisfying
\begin{gather}
  \mathsf{e}^2 = \mathsf{e}
  ,
  \notag
  \\
  \mathsf{e}  \, \mathsf{T}_1
  = \mathsf{T}_1 \, \mathsf{e}
  =
  t_1^{-1} \, \mathsf{e}
  .
\end{gather}

Hereafter we denote $\boldsymbol{t}=(t_0,t_1, t_2, t_3)$.
\begin{definition}
  The spherical DAHA $\SH_{q, \boldsymbol{t}}$ of $C^\vee C_1$-type
  is defined by
  $\SH_{q,\boldsymbol{t}}= \mathsf{e} \, \mathcal{H}_{q,\boldsymbol{t}} \, \mathsf{e}$.
\end{definition}

We recall some of the known  automorphisms of DAHA.
See, \emph{e.g.},~\cite{NoumiStokm00a,Oblom04a}.
\begin{lemma}
  We have an involutive anti-automorphism $\epsilon^\prime$ defined by
  \begin{align}
    \label{AW_epsilon}
    &
      \epsilon^\prime:
      \begin{pmatrix}
        \mathsf{T}_0 \\
        \mathsf{T}_1 \\
        \mathsf{T}_0^\vee \\
        \mathsf{T}_1^\vee 
      \end{pmatrix}
    \mapsto
    \begin{pmatrix}
      \mathsf{T}_1^\vee \\
      \mathsf{T}_1 \\
      \mathsf{T}_0^\vee \\
      \mathsf{T}_0
    \end{pmatrix}
    ,
    \qquad\qquad
       \begin{pmatrix}
         t_0 \\ t_1 \\ t_2 \\ t_3
       \end{pmatrix}
    \mapsto
    \begin{pmatrix}
      t_3 \\ t_1 \\ t_2 \\ t_0
    \end{pmatrix}
    .
  \end{align}
\end{lemma}

\begin{lemma}
  The $SL(2;\mathbb{Z})$ action on $\mathcal{H}_{q,\boldsymbol{t}}$ is
  generated by
  \begin{align}
    &
    \begin{pmatrix}
      1 & 1
      \\
      0 &1
    \end{pmatrix}
    \mapsto \sigma_{R}
          ,
        &
        &
          \begin{pmatrix}
            1 & 0 \\
            1 & 1
          \end{pmatrix}
                \mapsto \sigma_{L}
                ,
  \end{align}
  where
  the automorphisms are
  \begin{align}
    &
      \sigma_{R}:
      \begin{pmatrix}
        \mathsf{T}_0 \\
        \mathsf{T}_1 \\
        \mathsf{T}_0^\vee \\
        \mathsf{T}_1^\vee 
      \end{pmatrix}
      \mapsto
      \begin{pmatrix}
        \mathsf{T}_0  \, \mathsf{T}_0^\vee \, {\mathsf{T}_0}^{-1} \\
        \mathsf{T}_1 \\
        \mathsf{T}_0 \\
        \mathsf{T}_1^\vee 
      \end{pmatrix}
      ,
      \qquad \qquad
      \begin{pmatrix}
        t_0 \\ t_1 \\ t_2 \\ t_3
      \end{pmatrix}
      \mapsto
      \begin{pmatrix}
        t_2 \\ t_1 \\ t_0 \\ t_3
      \end{pmatrix}
    ,
    \notag
      \\[2ex]
    &
      \sigma_{L}:
      \begin{pmatrix}
        \mathsf{T}_0 \\
        \mathsf{T}_1 \\
        \mathsf{T}_0^\vee \\
        \mathsf{T}_1^\vee 
      \end{pmatrix}
      \mapsto
      \begin{pmatrix}
        \mathsf{T}_0 \\
        \mathsf{T}_1 \\
        \mathsf{T}_1^\vee \\
        {\mathsf{T}_1^\vee}^{-1} \,  \mathsf{T}_0^\vee \, \mathsf{T}_1^\vee
      \end{pmatrix}
      ,
      \qquad \qquad
      \begin{pmatrix}
        t_0 \\ t_1 \\ t_2 \\ t_3
      \end{pmatrix}
      \mapsto
      \begin{pmatrix}
        t_0 \\ t_1 \\ t_3 \\ t_2
      \end{pmatrix}
    .
    \label{auto_AW}
  \end{align}
\end{lemma}

We note that
\begin{equation}
  \label{tau_AW_R_and_L}
  \sigma_{R}= \epsilon^\prime \, \sigma_{L} \, \epsilon^\prime
  ,
\end{equation}
where $\epsilon^\prime$ is an involutive
anti-automorphism~\eqref{AW_epsilon}.


\subsection{Polynomial Representation and Askey--Wilson Polynomial}
We review the representation on the Laurent polynomials
$\mathbb{C}[x^{\pm 1}]$.
See, \emph{e.g.},~\cite{Noum95,SSah97b,NoumiStokm00a}.

\begin{prop}
  A polynomial representation is given by
  \begin{align}
    \mathsf{T}_0
    & \mapsto
      {t_0}^{-1} \, \mathsf{s} \, \eth -
      \frac{
      q^{-1} \left({t_0}^{-1}-t_0\right)  x^2
      + q^{-\frac{1}{2}} \left({t_2}^{-1}-t_2 \right)  x
      }{
      1-q^{-1} x^2
      } \,
      \left( 1 - \mathsf{s} \, \eth \right)
      ,
      \notag
    \\
    \mathsf{T}_1 
    & \mapsto
      {t_1}^{-1} \mathsf{s}+
      \frac{
      \left( {t_1}^{-1}-t_1 \right)
      +
      \left( {t_3}^{-1}-t_3 \right)  x
      }{
      x^2-1} \,
      \left( \mathsf{s} -1 \right)
      ,
      \notag
    \\
    \mathsf{T}_0^\vee
    & \mapsto
      q^{-\frac{1}{2}} {\mathsf{T}_0}^{-1} x
      ,
      \notag
    \\
    \mathsf{T}_1^\vee
    & \mapsto
      x^{-1} {\mathsf{T}_1}^{-1}
      .
      \label{AW_poly_rep}
  \end{align}
\end{prop}

Based on this representation, we define
\begin{align}
  \mathsf{Y}
  & =
    \mathsf{T}_1 \mathsf{T}_0
    ,
    \notag
  \\
  \mathsf{X}
  & =
    \left( \mathsf{T}_1^\vee \, \mathsf{T}_1 \right)^{-1}
    ,
    \label{CC1_XY}
\end{align}
where  $\mathsf{Y}$ is the Dunkl--Cherednik operator for the
Askey--Wilson polynomial.

As in the case of $A_1$-type,
we have
$\mathsf{T}_1 f = t_1^{~-1} f$
for a symmetric Laurent polynomial
$f \in \mathbb{C}[x+x^{-1}]$.
We see that the projection 
$\mathsf{e}$ is
$ \mathbb{C}[x]\to \mathbb{C}[x+x^{-1}]$,
and
$\SH_{q,\boldsymbol{t}}$ preserves
$\mathbb{C}[x+x^{-1}]$.
On this symmetric polynomial space,
the so-called Askey--Wilson operator is explicitly written as
\begin{equation}
  \label{AW_operator}
  \left.
    \mathsf{Y}+ \mathsf{Y}^{-1}
  \right|_{\text{sym}}
  \mapsto
  A(x; \boldsymbol{t}) \, ( \eth -1 ) 
    +
    A(x^{-1}; \boldsymbol{t}) \, (\eth^{-1}-1)
    + t_0 \,t_1
    + \left( t_0 \, t_1 \right)^{-1}
    ,
\end{equation}
where
\begin{equation}
  \label{AW_function_A}
  A(x; \boldsymbol{t})
  =
  t_0 \, t_1 \,
  \frac{
    \left( 1- \frac{1}{t_1 t_3} x \right)
    \left( 1+ \frac{t_3}{t_1 } x \right) 
    \left( 1- \frac{q^{\frac{1}{2}}}{t_0 t_2} x \right) 
    \left( 1+ \frac{q^{\frac{1}{2}} t_2 }{t_0} x \right) 
  }{
    \left( 1-x^2 \right)
    \left( 1-q \, x^2 \right)
  } .
\end{equation}

Eigenfunctions of $\mathsf{Y}$~\eqref{CC1_XY} are called the
non-symmetric Askey--Wilson
polynomial,
\begin{align}
    \mathsf{Y} \, E_m(x; q, \boldsymbol{t})
    & = \left( t_0 \, t_1 \right)^{-1} q^m \,
      E_m(x; q, \boldsymbol{t}),
      \notag
    \\
    \mathsf{Y} \, E_{-m}(x; q, \boldsymbol{t} )
    & = t_0 \, t_1 \, q^{-m} \,
    E_{-m}(x; q, \boldsymbol{t})
      .
\end{align}
Here
$m >0$ and
\begin{align}
    E_m(x; q, \boldsymbol{t}) & =
    x^m +
    \frac{
      \left( t_0^{~2} -1 \right) t_1^{~2}
      +q^m \left( t_1^{~2} -1 \right)
    }{
      \left( t_0 \, t_1 \right)^2
      q^{-m}
      - q^m
    } \, x^{-m}
    + \cdots
    ,
    \\
    E_{-m}(x; q, \boldsymbol{t}) & =
    x^{-m} + \cdots
                                   ,
                                   \notag
\end{align}
where $\cdots$ denote the Laurent polynomials~$x^k$ with~$|k|<m$.
We note that
$E_0(x;q,\boldsymbol{t})=1$ and that
$\mathsf{Y} E_0(x;q, \boldsymbol{t})= \left( t_0 t_1 \right)^{-1}
E_0(x;q,\boldsymbol{t})$.

The eigenfunctions of~\eqref{AW_operator} are the symmetric Askey--Wilson
polynomials~\cite{AsWi85}.
We have
\begin{equation}
  \label{AW_eigenfunction}
  \left( \mathsf{Y} + \mathsf{Y}^{-1} \right)
  P_m(x; q, \boldsymbol{t})
  =
  \left(
    \left(t_0\,  t_1 \right)^{-1} q^m
    + t_0 \, t_1 \,  q^{-m} \right)  \,
  P_m(x; q, \boldsymbol{t})
  .
\end{equation}
Here  we have
\begin{equation}
  \label{def_AW_poly}
  P_m(x ; q, \boldsymbol{t})
  =
  \frac{
    \left( a \,  b, a \ c, a \, d ; q\right)_m
  }{
    a^m
    \left( a \,  b \,  c \,  d \,  q^{m-1} ; q\right)_m
  } \,
  {}_4\phi_3
  \left[
    \begin{matrix}
      q^{-m},
      \
      q^{m-1} a \, b  \,  c  \, d,
      \
      a \, x,
      \
      a  \, x^{-1}
      \\
      a  \, b ,
      \
      a \,  c,
      \
      a \, d
    \end{matrix}
    ; 
    q, q
  \right]
  ,
\end{equation}
where
\begin{equation}
  \label{para_abcd}
  a=\frac{1}{t_1 \, t_3}
  ,
  \qquad
  b= - \frac{t_3}{t_1}
  ,
  \qquad
  c=\frac{q^{\frac{1}{2}}}{t_0 \, t_2}
  ,
  \qquad
  d= -\frac{q^{\frac{1}{2}} t_2}{t_0 }
  .
\end{equation}
Note that
we have normalized the polynomials so that
\begin{equation*}
  P_m(x; q, \boldsymbol{t}) =
  \left( x^m+x^{-m} \right)+
  \cdots
  ,
\end{equation*}
and $P_0(x; q, \boldsymbol{t})=1$.
These are written in terms of the non-symmetric polynomials as
\begin{equation}
  \begin{aligned}[t]
    P_m(x;q,\boldsymbol{t})
    & =
    E_m(x; q, \boldsymbol{t})
    + \frac{
      \left( q^m -1 \right)
      \left( t_0^{~2} +q^m \right) t_1^{~2}
    }{
      q^{2m} - \left( t_0 \, t_1\right)^2
    } \,
    E_{-m}(x; q, \boldsymbol{t})
    \\
    & =
    t_1 \left(
      \mathsf{T}_1 + t_1
    \right) E_{-m}(x;q, \boldsymbol{t})
    .
  \end{aligned}
\end{equation}
Some of them are explicitly written as
\begin{align}
  &
    P_0(x; q, \boldsymbol{t}) = 1,
    \notag
  \\
  &
    P_1(x;q, \boldsymbol{t})
    =
    x+x^{-1}
    +
    \frac{
    q^{\frac{1}{2}} t_0
    \left( 1+t_1^{~2} \right)  \left( 1-t_2^{~2} \right)  t_3
    +
    \left( q+t_0^{~2} \right)  t_1 \,  t_2
    \left( 1-t_3^{~2} \right)
    }{
    \left(
    q-t_0^{~2} t_1^{~2}
    \right)   t_2 \, t_3
    }
    .
\end{align}
Higher order polynomials are  generated from the three-term recurrence
relation.
It is read as
(see, \emph{e.g.},~\cite{GaspRahm04})
\begin{equation}
  \label{AW_three_1}
  \left( \mathsf{X}+\mathsf{X}^{-1} \right)
  P_m(x ; q, \boldsymbol{t})
  =
  P_{m+1}(x ; q, \boldsymbol{t})
  + B_m \, P_m(x ; q, \boldsymbol{t})
  + C_m  \, P_{m-1}(x ; q, \boldsymbol{t})
  .
\end{equation}
Here using~\eqref{para_abcd} we have
\begin{gather}
  \begin{multlined}
    C_n
    =
    \frac{
      1-a b c d q^{n-2}
    }{
      \left( 1-a \, b \, c \, d \, q^{2n-3} \right)
      \left( 1-a \,  b \, c \, d \, q^{2n-2} \right)
    }
    \cdot
    \frac{
      1-q^n
    }{
      \left( 1-a \, b \, c \, d \, q^{2n-2} \right)
      \left( 1-a \, b \, c \, d \, q^{2n-1} \right)
    }
    \\
    \times
    \left( 1-a \, b \, q^{n-1} \right)
    \left( 1-a \, c \, q^{n-1} \right)
    \left( 1-a \, d \, q^{n-1} \right)
    \left( 1-b \, c \, q^{n-1} \right)
    \left( 1-b \, d \, q^{n-1} \right)
    \left( 1-c \, d \, q^{n-1} \right)
    ,
  \end{multlined}
  \notag
  \\[1ex]
  \begin{multlined}
    B_n
    =a+a^{-1}
    - 
    \frac{1-a \, b \, c \, d \, q^{n-1}}{
      \left( 1-a \, b  \, c  \, d  \, q^{2n-1} \right)
      \left( 1-a \, b \, c \,  d \, q^{2n} \right)
    } \, a^{-1}
    \left( 1-a \,  b \,  q^n \right)
    \left( 1-a \, c \, q^n \right)
    \left( 1-a \, d \, q^n \right)
    \\
    -
    C_n
    \frac{
      \left( 1-a \, b \, c \, d \, q^{2n-3} \right)
      \left( 1-a \, b \, c \, d \, q^{2n-2} \right)
    }{
      1-a \, b \, c \, d \, q^{n-2}
    }
    \frac{
      a
    }{
      \left( 1-a  \, b \, q^{n-1} \right)
      \left( 1-a \, c \, q^{n-1} \right)
      \left( 1-a \, d \, q^{n-1} \right)
    }
    .
  \end{multlined}
\end{gather}
It is noted that we also  have
\begin{multline}
  \label{AW_three_2}
  \left( \mathsf{T}_1 \, \mathsf{T}^\vee_0
    + \left(\mathsf{T}_1 \, \mathsf{T}^\vee_0 \right)^{-1}
  \right)
  P_m(x ; q, \boldsymbol{t})
  =
  q^{m+\frac{1}{2}}
  \left(  t_0 \, t_1 \right)^{-1}
  P_{m+1}(x ; q, \boldsymbol{t})
  \\
  +
  \left(
    B_m
    -
    q^{m-\frac{1}{2}}
    \frac{
      \left( t_0 - q^{\frac{1}{2}} t_1 \right)
      \left( 1+q^{\frac{1}{2}} t_0 \, t_1 \right)
      \left( t_2- t_3 \right)
      \left( 1+t_2 \,  t_3 \right)
    }{
      \left( q^{m-\frac{1}{2}}+t_0  \, t_1 \right)
      \left( q^{m+\frac{1}{2}}+t_0 \, t_1 \right)
      t_2 \,  t_3}
  \right) \, P_m(x ; q, \boldsymbol{t})
  \\
  +
  t_0 \,  t_1 \,
  q^{-m+\frac{1}{2}} C_m \,
  P_{m-1}(x ; q, \boldsymbol{t})
  .
\end{multline}

\subsection{Automorphisms as Conjugation}
We  study the $SL(2;\mathbb{Z})$ action~\eqref{auto_AW} of $\SH_{q,\boldsymbol{t}}$
under the polynomial representation~\eqref{AW_poly_rep}.
We introduce
\begin{equation*}
  V_R = \exp \left( - \frac{
      \left( \log \mathsf{X} \right)^2}{
      2 \log q
    }
  \right)
  .
\end{equation*}
This function  is symmetric in $\mathsf{X} \leftrightarrow
\mathsf{X}^{-1}$,
thus
$\mathsf{s} \, V_R=V_R \, \mathsf{s}$, and commutes
with $\mathsf{T}_1$ and $\mathsf{T}^\vee_1$.
As we have
$\mathsf{s} \, \eth \, V_R
= q^{-\frac{1}{2}} \mathsf{X} \, V_R \, \mathsf{s} \, \eth$,
we get
\begin{equation*}
  V_R \, \mathsf{T}_0
  =
  \left(
    q^{\frac{1}{2}} x^{-1} a(x ) \, \mathsf{s} \, \eth
    + b(x)
  \right) \, V_R
  ,
\end{equation*}
where we have used
$\mathsf{T}_0 = a(x) \, \mathsf{s} \, \eth + b(x)$
in~\eqref{AW_poly_rep} for brevity.
We see that  the expression in the parenthesis coincides with
\begin{equation*}
  \mathsf{T}^\vee_0 =
  q^{-\frac{1}{2}} \mathsf{T}_0^{~-1} \, x
  =
  q^{\frac{1}{2}} x^{-1} a(x) \, \mathsf{s} \, \eth
  + q^{-\frac{1}{2}} x \left( b(x) + t_0-t_0^{~-1} \right) ,
\end{equation*}
when
\begin{equation}
  \label{condition_AW_R}
  \left( t_0-t_2 \right) \left( 1+t_0 \, t_2 \right) = 0
  .
\end{equation}
So
assuming~\eqref{condition_AW_R},
we have $V_R^{~-1} \mathsf{T}^\vee_0 V_R = \mathsf{T}_0$,
and also
\begin{equation*}
  \mathsf{T}_0 V_R \, \mathsf{T}_0
  =\mathsf{T}_0 \mathsf{T}^\vee_0 V_R
  =
  q^{-\frac{1}{2}} x \, V_R
  = V_R \,   q^{-\frac{1}{2}} x
  = V_R \, \mathsf{T}_0 \mathsf{T}^\vee_0 ,
\end{equation*}
which proves
$V_R^{~-1} \mathsf{T}_0 V_R= \mathsf{T}_0 \mathsf{T}^\vee_0
\mathsf{T}_0^{~-1}$.
The automorphism~$\sigma_R$~\eqref{auto_AW}
is thus realized by conjugation of $V_R$.

For~$\sigma_L$~\eqref{auto_AW}, we recall~\eqref{tau_AW_R_and_L} where
the anti-involution~$\epsilon^\prime$~\eqref{AW_epsilon} sends
$\mathsf{X} \mapsto \mathsf{Y}^{-1}$ and $\mathsf{Y} \mapsto
\mathsf{X}^{-1}$.
As in the case of $A_1$-type, the automorphism $\sigma_L$ is also
realized by conjugation of $V_L$.
To conclude, we have the following.
\begin{prop}
  Under the condition~\eqref{condition_AW_R},
  the $SL(2;\mathbb{Z})$ actions~\eqref{auto_AW} are conjugations 
  \begin{equation*}
    \sigma_{\bullet}: h \mapsto
    V_\bullet^{-1} h \,
    V_\bullet
    ,
  \end{equation*}
  where
  \begin{align}
    \label{auto_CC1}
      V_R
      & =
      \exp \left( - \frac{ \left( \log \mathsf{X} \right)^2}{
          2 \log q}
      \right)
      ,
        &
      V_L&
      =
      \exp \left( \frac{\left( \log \mathsf{Y} \right)^2}{
          2 \log q}
      \right)
           .
  \end{align}
\end{prop}



\subsection{Algebra Embedding}
We shall  define $\mathcal{A}_{(r,s)} \in \SH_{q,\boldsymbol{t}}$
associated  to a simple closed curve $\mathbb{c}_{(r,s)}$ on  the
sphere $\Sigma_{0,4}$
with a slope $s/r$,
\begin{equation}
  \label{sphere_c_and_A}
  \mathbb{c}_{(r,s)} \mapsto \mathcal{A}_{(r,s)}
  =
  \ch \left( \mathcal{O}_{(r,s)} \right)
  .
\end{equation}
Amongst them,
we put
\begin{align}
  \mathcal{A}_{(1,0)}
  &=
    \mathsf{X}+\mathsf{X}^{-1}
    =
    \ch \left( \mathsf{T}^\vee_1 \mathsf{T}_1 \right)
    =
    \ch   \left(
    q^{\frac{1}{2}} \mathsf{T}_0  \mathsf{T}_0^\vee
    \right)
    ,
    \notag
  \\
  \mathcal{A}_{(0,1)}
  &=
    \mathsf{Y}+\mathsf{Y}^{-1}
    =
    \ch \left( \mathsf{T}_1 \mathsf{T}_0\right)
    ,
    \notag
  \\
  \mathcal{A}_{(1,1)}
  & =
    \ch \left(
    \mathsf{T}_1 \mathsf{T}_0^\vee
    \right)
    =
    \ch \left(
    q^{-\frac{1}{2}} \mathsf{T}_1 \, {\mathsf{T}_0}^{-1} \, \mathsf{X}
    \right)
    .
    \label{def_A11}
\end{align}
A tedious but straightforward computation proves
  \begin{align}
    \label{a10_a01}
      \mathcal{A}_{(1,0)}    \mathcal{A}_{(0,1)}
      & =
      q^{-\frac{1}{2}} \mathcal{A}_{(1,1)}
      + q^{\frac{1}{2}} \mathcal{A}_{(1,-1)}
      -
      {t}_{03,12}
      ,
      \\
      \mathcal{A}_{(1,1)}    \mathcal{A}_{(0,1)}
      & =
      q^{-\frac{1}{2}} \mathcal{A}_{(1,2)}
      + q^{\frac{1}{2}} \mathcal{A}_{(1,0)}
      -
      {t}_{02,13}
        ,
        \notag
      \\
      \mathcal{A}_{(1,0)}    \mathcal{A}_{(1,1)}
      & =
      q^{-\frac{1}{2}} \mathcal{A}_{(2,1)}
      + q^{\frac{1}{2}} \mathcal{A}_{(0,1)}
      -
      {t}_{01,23}
        .
        \notag
  \end{align}
  Here we have
  \begin{align}
    \mathcal{A}_{(1,-1)}
    & =
      \ch \left(
      q^{\frac{1}{2}} \mathsf{T}_0 \mathsf{T}^\vee_1
      \right)
      ,
      \notag
    \\
    \mathcal{A}_{(1,2)}
    & =
      \ch \left(
      \mathsf{T}_1\mathsf{T}^\vee_0 \mathsf{T}^\vee_1
      \left( \mathsf{T}^\vee_0 \right)^{-1}
      \right)
          =
    \sigma_{L}^{~-1}(\mathcal{A}_{(1,1)})
      ,
      \notag
    \\
    \mathcal{A}_{(2,1)}
    & =
      \ch \left(
      \mathsf{T}_1 \left( \mathsf{T}^\vee_0\right)^{-1}
      \mathsf{T}_0 \mathsf{T}^\vee_0
      \right)
          =
    \sigma_{R}^{~-1}(\mathcal{A}_{(1,1)})
      ,
      \label{def_Ars_others}
  \end{align}
and
\begin{align}
  {t}_{03,12}
  & =
    \left(
    q^{\frac{1}{2}} t_1 - q^{-\frac{1}{2}}t_1^{~-1}
    \right) \left( t_2 - t_2^{~-1} \right)
    +
    \left( t_0 -t_0^{~-1} \right)
    \left(t_3 -t_3^{~-1} \right)
    ,
    \notag
  \\
  {t}_{02,13}
  & =
    \left(
    q^{\frac{1}{2}} t_1 - q^{-\frac{1}{2}}t_1^{~-1}
    \right) \left( t_3 - t_3^{~-1} \right)
    +
    \left( t_0 -t_0^{~-1} \right)
    \left(t_2 -t_2^{~-1} \right)
    ,
    \notag
  \\
  {t}_{01,23}
  & =
    \left(t_2 -t_2^{~-1} \right)
    \left( t_3 -t_3^{~-1} \right)
    +
    \left( t_0 - t_0^{~-1} \right)
    \left(
    q^{\frac{1}{2}} t_1 - q^{-\frac{1}{2}}t_1^{~-1}
    \right)
    .
\end{align}
Furthermore, we can check that
\begin{multline}
  \label{AW_11_11}
  \mathcal{A}_{(1,-1)} \mathcal{A}_{(1,1)}
  =
  q^{-1} \mathcal{A}_{(2,0)_T} +q \mathcal{A}_{(0,2)_T}
  -q^{-\frac{1}{2}} {t}_{02,13} \mathcal{A}_{(1,0)}
  -q^{\frac{1}{2}} {t}_{01,23} \mathcal{A}_{(0,1)}
  + \left( q^{\frac{1}{2}} - q^{-\frac{1}{2}} \right)^2
  \\
  -
  \left( t_0-t_0^{~-1} \right)^2
  -
  \left( q^{\frac{1}{2}} t_1- q^{-\frac{1}{2}}t_1^{~-1} \right)^2
  -
  \left( t_2-t_2^{~-1} \right)^2
  -
  \left( t_3-t_3^{~-1} \right)^2
  \\
  +
  \left( t_0-t_0^{~-1} \right)
  \left( q^{\frac{1}{2}} t_1- q^{-\frac{1}{2}}t_1^{~-1} \right)
  \left( t_2-t_2^{~-1} \right)
  \left( t_3-t_3^{~-1} \right)
  ,
\end{multline}
where we have
in terms of the Chebyshev polynomial~\eqref{Chebysheff} 
\begin{align}
  \mathcal{A}_{(2,0)_T}
  & =
    T_2(\mathcal{A}_{(1,0)})
    =
    \left(
    \mathsf{T}_1^\vee \mathsf{T}_1
    \right)^2
    +
    \left(
    \mathsf{T}_1^\vee \mathsf{T}_1
    \right)^{-2}
    ,
    \notag
    \\
  \mathcal{A}_{(0,2)_T}
  & =
    T_2(\mathcal{A}_{(0,1)})
    =
    \left(
    \mathsf{T}_1 \mathsf{T}_0
    \right)^2
    +
    \left(
    \mathsf{T}_1 \mathsf{T}_0
    \right)^{-2}
    .
\end{align}
Combining~\eqref{a10_a01} with~\eqref{AW_11_11}, we
find that
$\mathcal{A}_{(1,0)}$, $\mathcal{A}_{(0,1)}$, and
$\mathcal{A}_{(1,1)}$ fulfill
the cubic relation~\eqref{center_x_b}.
See also~\cite{Oblom04a,Terwil13a}.
As a result,
we have  
the algebra embedding of the skein
algebra~\eqref{algS04} and~\eqref{center_x_b} as follows.
\begin{theorem}
  \label{thm:embed_CC1}
  We have an algebra embedding
  $\KBS_A(\Sigma_{0,4}) \to \SH_{q,\boldsymbol{t}}$ with
  $A^2=q^{-\frac{1}{2}}$ by
  \begin{equation}
    \label{CC_t_and_x}
    \left(
      \begin{matrix}
        \mathbb{x} \\
        \mathbb{y} \\
        \mathbb{z}
      \end{matrix}
    \right)
    \mapsto
    \left(
      \begin{matrix}
        \mathcal{A}_{(1,0)} \\
        \mathcal{A}_{(0,1)} \\
        \mathcal{A}_{(1,1)}
      \end{matrix}
    \right)
    ,
    \qquad \qquad
    \left(
      \begin{matrix}
        \mathbb{b}_1 \\
        \mathbb{b}_2 \\
        \mathbb{b}_3 \\
        \mathbb{b}_4
      \end{matrix}
    \right)
    \mapsto
    \pm \left(
      \begin{matrix}
        \I \left( t_0 - t_0^{~-1} \right)
        \\
        \I \left( t_2 - t_2^{~-1} \right)
        \\
        \I \left(
          q^{\frac{1}{2}} t_1 - q^{-\frac{1}{2}}t_1^{~-1}
        \right)
        \\
        \I \left( t_3 - t_3^{~-1} \right)
      \end{matrix}
    \right)
    .
  \end{equation}
\end{theorem}

With the  embedding,
the three
term relations~\eqref{AW_three_1} and~\eqref{AW_three_2}
give the representation of $\mathbb{x}$ and $\mathbb{z}$
in terms of the Askey--Wilson polynomials as eigenfunctions of
$\mathbb{y}$~\eqref{AW_eigenfunction}.
As in the case of the skein algebra on the once-punctured
torus~\cite{ChoGel14a}, we obtain a finite-dimensional
representation when we set
$q$ to be a root of unity.

It should be 
noted that
a torus is a two-fold branched cover of a sphere over four-points.
This may correspond
to the fact~\cite{Hikam95g}
that the  DAHA of $C^\vee C_1$-type is constructed using
the
$A_1$-type with the reflection equation~\cite{Skl88}.

The other DAHA operators $\mathcal{A}_{(r,s)}$  associated to
$\mathbb{c}_{(r,s)}$
can be given using the
$SL(2;\mathbb{Z})$ actions on $\SH_{q,\boldsymbol{t}}$.
For instance,
when we
apply the automorphisms~\eqref{auto_AW} to~\eqref{a10_a01}, we get
\begin{gather}
  \mathcal{A}_{(1,n)} \mathcal{A}_{(0,1)}
  =
  q^{-\frac{1}{2}} \mathcal{A}_{(1,n+1)}
  +
  q^{\frac{1}{2}} \mathcal{A}_{(1,n-1)}
  -
  \begin{cases}
    {t}_{02,13}
    ,
    & \text{for odd $n$,}
    \\
    {t}_{03,12}
    ,
    & \text{for even $n$,}
  \end{cases}
  \notag
  \\
  \mathcal{A}_{(1,0)} \mathcal{A}_{(n,1)}
  =
  q^{-\frac{1}{2}} \mathcal{A}_{(n+1,1)}
  +
  q^{\frac{1}{2}} \mathcal{A}_{(n-1,1)}
  -
  \begin{cases}
    {t}_{01,23}
    ,
    & \text{for odd $n$,}
    \\
    {t}_{03,12}
    ,
    & \text{for even $n$,}
  \end{cases}
\end{gather}
where
\begin{align}
  \mathcal{A}_{(1,2k)}
  & =
    \ch \left( \mathsf{Y}^{-k} \mathsf{X}^{-1} \mathsf{T}_1^{~-1} \mathsf{Y}^k
    \mathsf{T}_1 \right)
    ,
    \notag
  \\
  \mathcal{A}_{(1,2k+1)}
  & =
    \ch \left(
    q^{-\frac{1}{2}}
    \mathsf{Y}^{-k-1} \mathsf{T}_1 \mathsf{X} \mathsf{Y}^k
    \mathsf{T}_1
    \right)
    ,
    \label{slope1odd}
  \\
  \mathcal{A}_{(2k,1)}
  & =
    \ch \left(
    \mathsf{T}_1 \mathsf{X}^{-k} \mathsf{T}_1^{~-1}
    \mathsf{Y} \mathsf{X}^k
    \right)
    ,
    \notag
  \\
  \mathcal{A}_{(2k+1,1)}
  & =
    \ch \left(
    q^{-\frac{1}{2}}
    \mathsf{T}_1 \mathsf{X}^{-k} \mathsf{Y}^{-1} \mathsf{T}_1
    \mathsf{X}^{k+1}
    \right) .
    \notag
\end{align}
See~\cite{BaMuPrWiWa18a} for the algorithms to obtain the curve
$\mathbb{c}_{(r,s)}$
from the viewpoint of the skein algebra on the  sphere $\Sigma_{0,4}$.

\subsection{Automorphisms and Braiding}
The relationship between DAHA of $C^\vee C_1$-type and the skein
algebra $\KBS_A(\Sigma_{0,4})$ gives an interpretation of 
the $SL(2;\mathbb{Z})$ action~\eqref{auto_AW} of DAHA.
The mapping class group $\Mod(\Sigma_{0,4})$ is generated by half Dehn
twists, and
it is known
that the $SL(2;\mathbb{Z})$ action corresponds to
the Artin braid group
$B_3$,
which denotes
the subgroup  of $\Mod(\Sigma_{0,4})$
fixing one puncture~\cite{Birman74,FarbMarg11Book}.
As seen from the fact that
\begin{align*}
  \sigma_R : &
             \mathcal{A}_{(0,1)} \mapsto
             \mathcal{A}_{(1,-1)},
             \qquad \qquad
             \mathcal{A}_{(2,1)}  \mapsto
             \mathcal{A}_{(1,1)} ,
  \\
  \sigma_L : &
             \mathcal{A}_{(1,0)} \mapsto
             \mathcal{A}_{(1,-1)},
             \qquad \qquad
             \mathcal{A}_{(1,2)} \mapsto
             \mathcal{A}_{(1,1)}
             ,
\end{align*}
the automorphism $\sigma_{R}$
(resp. $\sigma_L$)
is identified with
a braiding of punctures $\mathbb{b}_1$ and $\mathbb{b}_2$
(resp. $\mathbb{b}_2$ and $\mathbb{b}_4$)
as in Fig.~\ref{fig:auto_braiding},
and
it 
denotes
the Dehn twist about $\mathbb{x}$
(resp. $\mathbb{y}$).

\begin{figure}[htbp]
  \centering
  \includegraphics[scale=1.0]{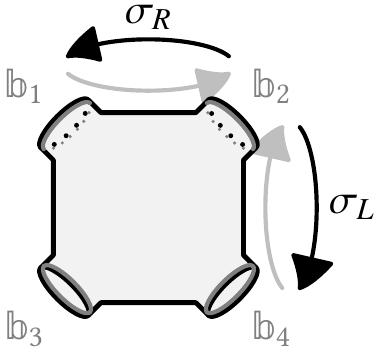}
  \caption{The automorphisms $\sigma_R$ and $\sigma_L$ are braidings of punctures.}
  \label{fig:auto_braiding}
\end{figure}

\subsection{DAHA polynomials for $\mathbb{c}_{(r,s)}$}
For the simple closed curve $\mathbb{c}_{(r,s)}$ with slope $s/r$ on
$\Sigma_{0,4}$, we have assigned  the  DAHA
operator~$\mathcal{A}_{(r,s)} \in \SH_{q,\boldsymbol{t}}$~\eqref{sphere_c_and_A}.
The DAHA polynomial associated to the curve $\mathbb{c}_{(r,s)}$ is
defined by
\begin{equation}
  \label{DAHA_poly_sphere}
  P_n(x,q, \boldsymbol{t}; \mathbb{c}_{(r,s)})
  =
  M_{n-1}(\mathcal{O}_{(r,s)}; q, q)
  (1)
  .
\end{equation}
Especially
\begin{equation}
  P_2(x,q,\boldsymbol{t}; \mathbb{c}_{(r,s)})
  =
  \mathcal{A}_{(r,s)}(1) .
\end{equation}
This is why we use the $A_1$-type
Macdonald polynomial in~\eqref{DAHA_poly_sphere}
rather than the Askey--Wilson polynomial,
$P_1(x;q, \boldsymbol{t})\neq x+x^{-1}$,
but
we may introduce  a new parameter as a $t$-parameter of
the $A_1$-Macdonald
polynomial~\eqref{DAHA_poly_sphere}.
Indeed
a different definition was used in~\cite{IChered16a} as the
$C^\vee C_1$ DAHA polynomial.

We give some explicit forms in the following.
We have
\begin{align*}
  {P}_n(x,q, \boldsymbol{t}; \mathbb{c}_{(1,0)})
  & =
    M_{n-1}(x; q, q)
    ,
  \\
  {P}_n(x,q, \boldsymbol{t}; \mathbb{c}_{(0,1)})
  & =
    M_{n-1}( (t_0 \, t_1)^{-1}; q,  q)
    .
\end{align*}
Using~\eqref{def_A11} and~\eqref{def_Ars_others} we have
\begin{align*}
  {P}_2(x,q, \boldsymbol{t}; \mathbb{c}_{(1,1)})
  & =
    q^{\frac{1}{2}} (t_0 \, t_1)^{-1} \left( x+ x^{-1} \right)
    -
    q^{\frac{1}{2}} t_0^{~-1} \left( t_3 - t_3^{~-1}\right)
    - t_1^{~-1} \left( t_2 - t_2^{~-1} \right) ,
  \\
  {P}_2(x,q, \boldsymbol{t}; \mathbb{c}_{(1,-1)})
  & =
    q^{-\frac{1}{2}} x^{-1} \,
    A(x;\boldsymbol{t}) 
    +
    q^{-\frac{1}{2}} x \, A(x^{-1}; \boldsymbol{t}) 
    - A(x; \tilde{\boldsymbol{t}})
    -A(x^{-1}; \tilde{\boldsymbol{t}})
    + t_1 \, t_2 +
    \left( t_1 \, t_2 \right)^{-1}
    ,
\end{align*}
where
$\tilde{\boldsymbol{t}}=
(t_2, t_1, t_0, t_3)$.



\section{Twice-punctured Torus}
\label{sec:twice_torus}
\subsection{Skein Algebra}
The skein algebra $\KBS_A(\Sigma_{1,2})$ on a twice-punctured torus
was studied in~\cite{BulloPrzyt99a}.
We define simple closed curves
$\mathbb{x}$, $\mathbb{y}$, $\mathbb{x}_u$, $\mathbb{y}_u$,
$\mathbb{b}_3$, $\mathbb{b}_4$
as in
Fig.~\ref{fig:torus_two_punctures},
where $\mathbb{b}_3$ and $\mathbb{b}_4$ are the boundary circles.
We  regard
$\Sigma_{1,2}$  as the surface constructed by
gluing an annulus $S^1\times [0,1]$ with
$\Sigma_{0,4}$ in Fig.~\ref{fig:4-hole_sphere},
$\Sigma_{1,2}=\Sigma_{0,4} \cup S^1 \times [0,1]$,
where
both
$S^1\times \{0\} \approx \mathbb{b}_1$ and
$S^1\times \{1\} \approx \mathbb{b}_2$ are isotopic to $\mathbb{x}_u$.
Then  we have the skein algebra of $\Sigma_{0,4}$-type,~\eqref{algS04} and~\eqref{center_x_b}
with $\mathbb{b}_1 \approx \mathbb{b}_2 \approx \mathbb{x}_u$.
Here $\mathbb{z}$ is generated by $\mathbb{x}$ and $\mathbb{y}$ from
the first identity of~\eqref{algS04}.
On the other hand,
$\Sigma_{1,2}$ is regarded as the surface
given by gluing a once-punctured  torus with a
thrice-punctured sphere,
$\Sigma_{1,2}=\Sigma_{1,1}\cup \Sigma_{0,3}$,
where the boundary circle of $ \Sigma_{1,1}$ is $\mathbb{x}$ and
the three boundary circles of $\Sigma_{0,3}$ are isotopic to $\mathbb{x}$,
$\mathbb{b}_3$, $\mathbb{b}_4$.
Then we have the  algebra~\eqref{torus_xy_z} of $\Sigma_{1,1}$
for $\mathbb{x}_u$, $\mathbb{y}_u$, and $\mathbb{z}_u$.
Here $\mathbb{z}_u$ is generated by $\mathbb{x}_u$ and $\mathbb{y}_u$
by the first identity of~\eqref{torus_xy_z}.
We see that
$\mathbb{x}$, which  is isotopic to the boundary circle
of
$ \Sigma_{1,1}$,
is generated by~\eqref{boundary_b}.
In addition, we need the  consistency
condition~\eqref{skein_twice_puncture} for $\mathbb{y}$ and
$\mathbb{y}_u$ as the skein algebra for $\Sigma_{1,2}$,
which can be checked directly.

\begin{figure}[htbp]
  \centering
  \includegraphics[scale=1.0]{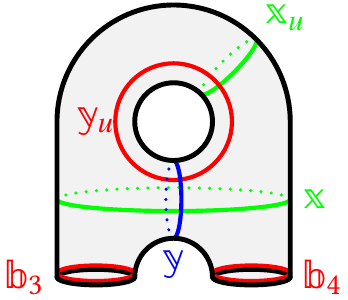}
  \caption{Depicted are simple closed curves on the twice-punctured
    torus $\Sigma_{1,2}$.}
  \label{fig:torus_two_punctures}
\end{figure}

\begin{prop}
  \label{prop:Sigma12}
  The skein algebra $\KBS_A(\Sigma_{1,2})$ is as follows;
  \begin{itemize}
    
  \item $\Sigma_{0,4}$-type,
    \begin{align}
      A^2 \, 
      \mathbb{x} \,  \mathbb{y}
      -A^{-2} \,
      \mathbb{y}  \,  \mathbb{x}
      & =
        \left( A^4-A^{-4} \right) \,
        \mathbb{z}
        + \left( A^2-A^{-2} \right) \,
        \left(
        \mathbb{b}_3
        +
        \mathbb{b}_4
        \right) \mathbb{x}_u
        ,
        \notag
      \\
        A^2 \, 
        \mathbb{y} \,  \mathbb{z}
        -A^{-2} \,
        \mathbb{z}  \,  \mathbb{y}
        & =
        \left( A^4-A^{-4} \right) \,
        \mathbb{x}
        + \left( A^2-A^{-2} \right) \,
        \left(
          \mathbb{x}_u^{~2}
          +
          \mathbb{b}_3 \, \mathbb{b}_4
        \right)
          ,
          \notag
        \\
        A^2 \, 
        \mathbb{z} \,  \mathbb{x}
        -A^{-2} \,
        \mathbb{x}  \,  \mathbb{z}
        & =
        \left( A^4-A^{-4} \right) \,
        \mathbb{y}
        + \left( A^2-A^{-2} \right) \,
        \left(
          \mathbb{b}_3
          +
          \mathbb{b}_4
        \right) \mathbb{x}_u
          ,
      \label{04type}      
    \end{align}
    with
    \begin{multline}
      A^2 \, \mathbb{x} \, \mathbb{y} \, \mathbb{z} =
      A^4 \,  \mathbb{x}^{2}
      + A^{-4} \, \mathbb{y}^{2}
      + A^4 \, \mathbb{z}^{2}
      \\
      +A^2  \,
      \left( \mathbb{x}_u^{~2}
        + \mathbb{b}_3  \, \mathbb{b}_4 \right)
      \, \mathbb{x}
      +A^{-2} \,
      \left( \mathbb{b}_3
        +  \mathbb{b}_4 \right) \mathbb{x}_u
      \, \mathbb{y}
      +A^2 \,
      \left(  \mathbb{b}_4
        +  \mathbb{b}_3 \right) \mathbb{x}_u
      \, \mathbb{z}
      \\
      +
      2 \,  \mathbb{x}_u^{~2}
      +  \mathbb{b}_3^{~2}+  \mathbb{b}_4^{~2}
      +
      \mathbb{x}_u^{~2} \,  \mathbb{b}_3\,  \mathbb{b}_4
      -\left( A^2+A^{-2} \right)^2
      ,
    \end{multline}
    
  \item $\Sigma_{1,1}$-type,
    \begin{align}
      A \, \mathbb{x}_u \, \mathbb{y}_u
      - A^{-1} \,  \mathbb{y}_u \, \mathbb{x}_u
      &=
        \left( A^2-A^{-2} \right)  \, \mathbb{z}_u
        ,
        \notag
      \\
      A  \,  \mathbb{y}_u \,  \mathbb{z}_u
      - A^{-1} \,  \mathbb{z}_u \,  \mathbb{y}_u
      & =
        \left( A^2-A^{-2} \right) \, \mathbb{x}_u
        ,
        \notag
      \\
      A  \, \mathbb{z}_u \, \mathbb{x}_u
      - A^{-1}\,  \mathbb{x}_u  \,  \mathbb{z}_u
      & =
        \left( A^2-A^{-2} \right) \, \mathbb{y}_u
        ,
        \label{torus11_type}
    \end{align}
    with
    \begin{equation}
      \label{11_generate_x}
      \mathbb{x}
      =
      A \, \mathbb{x}_u \, \mathbb{y}_u \, \mathbb{z}_u -
      A^2 \, \mathbb{x}_u^{~2} - A^{-2} \, \mathbb{y}_u^{~2}
      - A^2 \, \mathbb{z}_u^{~2}
      +A^2+A^{-2}
      ,
    \end{equation}

  \item consistency,
    \begin{gather}
        - \mathbb{y}^2 \mathbb{y}_u
        + \left(
          A^2+A^{-2}
        \right)
        \mathbb{y} \, \mathbb{y}_u \, \mathbb{y}
        - \mathbb{y}_u \, \mathbb{y}^2
        =
        \left(
          A^2- A^{-2}
        \right)^2 \mathbb{y}_u,
        \notag
        \\
        - \mathbb{y}_u^{~2} \mathbb{y}
        + \left(
          A^2+A^{-2}
        \right)
        \mathbb{y}_u \,  \mathbb{y} \, \mathbb{y}_u
        - \mathbb{y}  \, \mathbb{y}_u^{~2}
        =
        \left(
          A^2- A^{-2}
        \right)^2 \mathbb{y}
        ,
      \label{skein_twice_puncture}
      \end{gather}
    with
    \begin{align}
      \label{twice_torus_x_y}
      &
        \mathbb{x} \, \mathbb{y}_u
      = \mathbb{y}_u \, \mathbb{x}
        ,
        &&
      \mathbb{x}_u \, \mathbb{y}
      = \mathbb{y} \, \mathbb{x}_u
      .
    \end{align}
  \end{itemize}
  It is noted that the boundary circles $\mathbb{b}_3$ and $\mathbb{b}_4$
  are central.
\end{prop}

We note that the $\Sigma_{1,1}$-type relations~\eqref{torus11_type} are
redundant, and we have
\begin{gather}
  - \mathbb{y}_u^{~2} \mathbb{x}_u
  + \left(
    A^2+A^{-2}
  \right)
  \mathbb{y}_u \mathbb{x}_u\mathbb{y}_u
  - \mathbb{x}_u \mathbb{y}_u^{~2}
  =
  \left(
    A^2- A^{-2}
  \right)^2 \mathbb{x}_u,
  \notag
  \\
  - \mathbb{x}_u^{~2} \mathbb{y}_u
  + \left(
    A^2+A^{-2}
  \right)
  \mathbb{x}_u \mathbb{y}_u\mathbb{x}_u
  - \mathbb{y}_u \mathbb{x}_u^{~2}
  =
  \left(
    A^2- A^{-2}
  \right)^2 \mathbb{y}_u
  .
\end{gather}



\subsection{Construction of  DAHA}
To construct the  DAHA representation 
for $\KBS_A(\Sigma_{1,2})$ in Prop.~\ref{prop:Sigma12},
we shall first  make use of the
DAHA of $C^\vee C_1$-type 
which represents
$\KBS_A(\Sigma_{0,4})$ in Fig.~\ref{fig:4-hole_sphere}.
Due to that the curves
$\mathbb{b}_1 $ and $\mathbb{b}_2$ are set to  be
isotopic,
$\mathbb{b}_1 \approx \mathbb{b}_2 \approx \mathbb{x}_u$,
and that we have the embedding~\eqref{CC_t_and_x} for $\KBS_A(\Sigma_{0,4})$,
we  put
$t_0=t_2=\I x_u$.
Namely the parameters of $C^\vee C_1$-DAHA $\mathcal{H}_{q,\boldsymbol{t}}$~\eqref{AW_operator}  are set to be
\begin{equation}
  \label{t_natural}
  (t_0, t_1, t_2, t_3)
  =
  \left(
    \I \, x_u,
    \I  \, q^{-\frac{1}{2}} x_\ell,
    \I \, x_u ,
    \I  \, x_r
  \right)
  =
  \boldsymbol{t}_\natural
  .
\end{equation}
In the spherical $C^\vee C_1$-DAHA  $\SH_{q,\boldsymbol{t}_\natural}$
with $\boldsymbol{t}_\natural$~\eqref{t_natural},
we assign
$\mathsf{X}+\mathsf{X}^{-1}$ and
$ \mathsf{Y}+\mathsf{Y}^{-1} $
for the curves $\mathbb{x}$ and $\mathbb{y}$ respectively.
Explicitly,
we have  the following representation on $\mathbb{C}[x+x^{-1}]$
satisfying the $\Sigma_{0,4}$-type skein
relations~\eqref{04type},
\begin{align}
  \mathbb{x}
  & \mapsto
    x + x^{-1}
    ,
    \notag
  \\
  \mathbb{y}
  & \mapsto
    -\beta(x,x_u, x_\ell, x_r) \, \eth
    -\beta(x^{-1}, x_u, x_\ell, x_r ) \, \eth^{-1}
    -\varphi(x,x_u, x_\ell, x_r)
    ,
  \label{define_aw_glue}
    \\
    \mathbb{x}_u
    & \mapsto
    x_u+x_u^{~-1}
      ,
      \notag
    \\
    \mathbb{b}_3
    & \mapsto
    x_\ell + x_\ell^{~-1}
      ,
      \notag
    \\
    \mathbb{b}_4
    & \mapsto
    x_r + x_r^{~-1}
      ,
      \notag
\end{align}
where
\begin{align}
  \beta(x, x_u, x_\ell, x_r)
  & =
    \frac{
      \left( x_\ell + q^{\frac{1}{2}} x \, x_r \right)
      \left( q^{\frac{1}{2}} x + x_\ell \, x_r \right)
      \left( q^{\frac{1}{2}} x + x_u^{~2}  \right)
    }{
      q^{\frac{1}{2}}
      \left( 1- q^{\frac{1}{2}} x \right)
      \left( 1-x^2 \right)
      x_\ell \, x_r \, x_u
      },
      \notag
  \\
  \varphi(x,x_u,x_\ell,x_r)
  & =
    -
    \frac{
      x \,
      \left( x_\ell + x_r \right) 
      \left( 1 + x_\ell \, x_r \right)
      \left( 1+ x_u^{~2} \right)
    }{
      \left( 1 - q^{-\frac{1}{2}} x  \right)
      \left( 1 - q^{\frac{1}{2}} x  \right)
      x_\ell \, x_r \, x_u
    }
    .
\end{align}


To give  representations for $\mathbb{x}_u$ and $\mathbb{y}_u$
satisfying~\eqref{torus11_type}
and~\eqref{11_generate_x}, we use the
$A_1$-type  DAHA $\SH_{q_u,t}$.
The skein algebra embeddings in Theorems~\ref{thm:embed_A1} and~\ref{thm:embed_CC1}
suggest to set
\begin{equation}
  \label{qu_and_q}
  q_u= q^{\frac{1}{2}} .
\end{equation}
Recalling that
the boundary circle $\mathbb{b}$ in Fig.~\ref{fig:xyz_torus}
is generated by~\eqref{boundary_b}
and that the DAHA representation gives~\eqref{A1_b_and_t},
we see that~\eqref{11_generate_x} is satisfied by 
\begin{equation}
  \label{t_and_x}
  t=\I  \, q_u^{~\frac{1}{2}}x^{\frac{1}{2}}
  .
\end{equation}
For the consistency
conditions~\eqref{skein_twice_puncture},
we may
take a conjugation of $\SH_{q_u,t}$
by use of a ``gluing function''
$G(x,x_u)$ to be determined
as
\begin{align}
  \mathbb{x}_u
    & \mapsto
     G(x,x_u)^{-1} 
    \left(
      x_u+ x_u^{~-1}
    \right) 
    G(x,x_u)
    =
    x_u+ x_u^{~-1}
      ,
      \notag
    \\
    \mathbb{y}_u
    & \mapsto
    G(x,x_u)^{-1}
    \left(
      \gamma(x,x_u) \, \eth_u
      +
      \gamma(x, x_u^{~-1}) \, \eth_u^{~-1}
    \right) 
    G(x,x_u)
      .
  \label{define_yu_glue}
\end{align}
Here $\eth_u$ is a difference operator for $x_u$,
\begin{equation}
  \eth_u x_u = q_u x_u \eth_u
  ,
\end{equation}
and
$\gamma(x,x_u)$ is for
the Macdonald operator~\eqref{A1_Macdonald} of $\SH_{q_u,\I  q_u^{\frac{1}{2}} x^{\frac{1}{2}}}$
\begin{equation}
  \label{define_gamma}
  \gamma(x,x_u) =
  \frac{
    \I \,  q_u^{~\frac{1}{2}} \, x^{\frac{1}{2}}
    x_u -
    \left(
      \I \, q_u^{~\frac{1}{2}} \, x^{\frac{1}{2}}
    \right)^{-1}
    x_u^{~-1}
  }{
    x_u - x_u^{~-1}
  }
  .
\end{equation}
Under this setting,
the commutativity~\eqref{twice_torus_x_y} is trivial,
and
we have the $\Sigma_{1,1}$-type skein relation~\eqref{torus11_type} 
because of the conjugation of the $A_1$ DAHA representation.
To check the remaining consistency
condition~\eqref{skein_twice_puncture},
we assume that
\begin{equation}
  \mathbb{y}_u
  \mapsto
  \gamma_1(x_u) \,   \eth_u +
  \gamma_2(x,x_u) \,  \eth_u^{~-1}
  ,
  \label{define_gamma_a}
\end{equation}
By brute force computations, we find that the  consistency
conditions~\eqref{skein_twice_puncture}
are  fulfilled
when
\begin{gather}
  \gamma_1(x_u)
  =
  \I \, q^{-\frac{1}{4}}
  \frac{-1}{1-x_u^{~2}} ,
  \notag
  \\
  \gamma_2(x,x_u)
  =
  \I \, q^{\frac{1}{4}}
  \frac{
    \left( 1+ q^{-\frac{1}{2}} x^{-1} x_u^{~2} \right)
    \left( 1 + q^{-\frac{1}{2}} x \, x_u^{~2} \right)
  }{
    1- x_u^{~2}
  }
  .
\end{gather}
The representation~\eqref{define_gamma_a} is indeed the form
of~\eqref{define_yu_glue}, when 
the gluing function is defined  in terms of the quantum dilogarithm
function by
\begin{equation}
  \label{gluing_function}
  G(x,x_u)
  =
  \E^{\frac{\log x \,  \log x_u}{\log q}}
  (-q^{\frac{1}{2}} x \, x_u^{~2};q)_\infty
  .
\end{equation}
We should note that the gluing function $G(x,x_u)$
satisfies the following $q$-difference equations,
\begin{align}
  \label{difference_G}
  &
    \frac{G(q\, x, x_u)}{G(x, x_u)}
    =
    \frac{x_u}{1+q^{\frac{1}{2}} x \, x_u^{~2}}
    ,
  &&
    \frac{
      G(x, q^{\frac{1}{2}} x_u)
    }{
      G(x,x_u)
    }
    =
    \frac{x^{\frac{1}{2}}}{
      1+ q^{\frac{1}{2}} x \, x_u^{~2}
    }
     .
\end{align}

As a result, we obtain the following.

\begin{theorem}
  \label{thm:Sigma12}
  We have an algebra
  embedding~\eqref{define_aw_glue},~\eqref{define_yu_glue},~\eqref{gluing_function}
  of
  $\KBS_A(\Sigma_{1,2}) \to \SH_{q, \boldsymbol{t}_\natural}$ with
  $A^2=q^{-\frac{1}{2}}$
  as operators
  on symmetric polynomial
  $\mathbb{C}[x+x^{-1}]$.
\end{theorem}

It should be stressed  that the above representation preserves the
symmetric Laurent polynomial space
$\mathbb{C}[x+x^{-1}]$.
Due to conjugation by the gluing function $G(x,x_u)$~\eqref{gluing_function},
broken is a  symmetry $x_u \leftrightarrow x_u^{~-1}$.
We can recover the  symmetry of $x_u$ when we discard the  symmetry
$x \leftrightarrow x^{-1}$,
by taking an
inverse conjugation,
$h \mapsto G(x,x_u) \, h \, G(x,x_u)^{-1}$,
for the representations in Theorem~\ref{thm:Sigma12}.
We use $t$ by
\begin{equation}
  \label{t_and_x_2}
  x=-\frac{t^{2}}{q_u}
  ,
\end{equation}
as in~\eqref{t_and_x}, and
a $q$-difference operator for $t$  in $\SH_{q_u, t}$ is
\begin{equation}
  \label{shift_t}
  \eth_t \,  t
  =
  q_u \, t \, \eth_t
  .
\end{equation}
We obtain the following representation which acts on $\mathbb{C}[x_u+x_u^{~-1}]$.
\begin{coro}
  We have an algebra embedding,
  $\KBS_A(\Sigma_{1,2}) \to {\SH_{q_u,t}}$,
  with
  $A^2=q_u^{~-1}$,
  \begin{align}
    \mathbb{x}_u
    & \mapsto
      \mathsf{X}_u + \mathsf{X}_u^{~-1}
      ,
      \notag
    \\
    \mathbb{y}_u
    & \mapsto
      \mathsf{Y}_u + \mathsf{Y}_u^{~-1}
      ,
      \notag
    \\
    \mathbb{b}_3
    & \mapsto
      x_\ell + x_\ell^{~-1} ,
      \notag
    \\
    \mathbb{b}_4
    &
      \mapsto
      x_r + x_r^{~-1},
      \notag
    \\
    \mathbb{y}
    &    \mapsto
      \begin{multlined}[t]
        - \frac{q_u  \, t^2 \left(x_\ell - x_r \, t^2 \right)
          \left(x_\ell x_r - t^2 \right)
        }{
          \left( 1+t^2 \right) \left( q_u^{~2} - t^4 \right)
          x_\ell \, x_r
        }      \,
        \left( \mathsf{X}_u^{~-1} t^{-1} - \mathsf{X}_u \, t \right)
        \left( \mathsf{X}_u \, t^{-1} - \mathsf{X}_u^{~-1} t \right)
        \eth_t
        \\
        - \frac{ q_u  \left( q_u^{~2} x_r - x_\ell  \,  t^2 \right)
          \left( q_u^{~2} - x_\ell \,  x_r \, t^2 \right)
        }{
          \left( t^2+q_u^{~2} \right)
          \left(   q_u^{~2} - t^4 \right)
          x_\ell \, x_r
        } \, \eth_t^{~-1}
        -
        \frac{
          q_u \,  t^2 \left(x_\ell+x_r\right)
          \left(1+x_\ell  \, x_r \right)
        }{
          \left(1+t^2 \right)
          \left( t^2 + q_u^{~2} \right)
          x_\ell  \, x_r 
        } \,
        \left( \mathsf{X}_u + \mathsf{X}_u^{~-1} \right)
        .
      \end{multlined}
    \label{rep11_xu}
  \end{align}
  Here $\mathsf{X}_u$ and $\mathsf{Y}_u$ defined
  by~\eqref{A1_represent} acting on the Laurent polynomial of $x_u$
  constitute the DAHA $\SH_{q_u,t}$.
  The representation~\eqref{rep11_xu}  preserves  the symmetric polynomials
  $\mathbb{C}[x_u+x_u^{~-1}]$,
  and they are explicitly written as follows;
  \begin{align}
    \mathbb{x}_u
    & \mapsto
      x_u+x_u^{~-1} ,
      \notag
    \\
    \mathbb{y}_u
    & \mapsto
      \frac{t \, x_u - t^{-1} x_u^{-1}}{x_u - x_u^{~-1}} \,\eth_u
      +
      \frac{t^{-1} x_u - t \, x_u^{~-1}}{x_u - x_u^{~-1}} \,
      \eth_u^{~-1}
      ,
      \notag
    \\
    \mathbb{b}_3
      & \mapsto
      x_\ell + x_\ell^{~-1} ,
      \notag
    \\
    \mathbb{b}_4
    & \mapsto
      x_r + x_r^{~-1},
      \notag
    \\
    \mathbb{y}
    &    \mapsto
      \begin{multlined}[t]
        - \frac{q_u \left(x_\ell - x_r \, t^2 \right)
          \left(x_\ell \, x_r -  t^2 \right)
          \left(1 - x_u^{~2} t^2 \right)
          \left(x_u^{~2}- t^2 \right)
        }{
           \left( 1+t^2 \right) \left( q_u^{~2} -t^4 \right)
          x_\ell \, x_r\,  x_u^{~2}
        } \,
        \eth_t
        \\
        - \frac{q_u  \left( q_u^{~2} x_r - x_\ell  \, t^2 \right)
          \left(  q_u^{~2}- x_\ell \, x_r \, t^2 \right)
        }{
           \left( t^2+q_u^{~2} \right)
          \left(  q_u^{~2} - t^4 \right)
          x_\ell \, x_r
        } \, \eth_t^{~-1}
        -
        \frac{
          q_u \,  t^2 \left(x_\ell+x_r\right)
          \left(1+x_\ell \,  x_r \right)
          \left( 1+ x_u^{~2} \right)
        }{
          \left(1+t^2 \right)
          \left( t^2 + q_u^{~2} \right)
          x_\ell \, x_r \,  x_u
        }
        .
      \end{multlined}
  \end{align}
\end{coro}

\subsection{DAHA Polynomial of  Simple Closed Curves on $\Sigma_{1,2}$}

The  representation in Theorem~\ref{thm:Sigma12} is on
$\mathbb{C}[x+x^{-1}]$, while the representation~\eqref{rep11_xu} is
on $\mathbb{C}[x_u+x_u^{~-1}]$.
Using Theorem~\ref{thm:Sigma12}, we can assign
an operator $\mathcal{A}$ acting on $\mathbb{C}[x+x^{-1}]$ for a simple closed curve $\mathbb{c}$
on $\Sigma_{1,2}$,
\begin{equation}
  \mathbb{c}\mapsto \mathcal{A}=\ch (\mathcal{O}) .
\end{equation}
When $\mathbb{c}$ is given by Dehn twists from  $\mathbb{y}$,
$\mathbb{c}=\mathscr{T}(\mathbb{y})$,
we have $\mathcal{O}=\gamma(\mathsf{Y})$ due to that
$\mathbb{y}\mapsto\mathsf{Y}+\mathsf{Y}^{-1}\in \SH_{q,\boldsymbol{t}_\natural}$.
The automorphism  $\gamma$ is induced from $\mathscr{T}$,
and it is
written as the  conjugation by $U_\bullet$
and $V_\bullet$.
We then define the DAHA polynomial for $\mathbb{c}$ by
\begin{equation}
  Q_n(x,x_u,x_\ell, x_r,q; \mathbb{c})
  =
  M_{n-1}(\mathcal{O}; q, q)
  (1)
  .
\end{equation}
Especially we have
\begin{equation*}
  Q_2(x,x_u,x_\ell, x_r, q ; \mathbb{c})
  = \mathcal{A}(1).
\end{equation*}
For example, we have
\begin{gather*}
  Q_2(x,x_u, x_\ell, x_r, q; \mathbb{y}_u)
  =
  \I \, q^{-\frac{1}{2}} \frac{ x_u^{~2}}{1-x_u^{~2}}
  \left( x+x^{-1} \right)
  -
  \I \, q^{-\frac{3}{2}}
  \frac{ q (1-q)- x_u^{~4}}{1-x_u^{~2}} ,
  \\
  Q_2(x,x_u, x_\ell, x_r, q; \mathbb{y})
  =
  - q^{-\frac{1}{2}} x_u  \, x_\ell
  - q^{\frac{1}{2}} x_u^{~-1} \, x_\ell^{~-1}
  .
\end{gather*}
These  are in $\mathbb{C}(q,x_u,x_\ell,x_r)[x+x^{-1}]$, and
we do not know a relationship with the previously known quantum
polynomial invariants at this stage.

We shall pay attention to the representation~\eqref{rep11_xu} which
acts on $\mathbb{C}[x_u+x_u^{~-1}]$.
We assume that a simple closed curve $\mathbb{c}$ on $\Sigma_{1,2}$ is
given from either
$\mathbb{x}_u$
or
$\mathbb{c}_{(r,s)}$ on subsurface $\Sigma_{0,4} \subset
\Sigma_{1,2}$
by the Dehn twist
$\mathscr{T}$, where $\mathscr{T}$ is
generated by $\mathscr{T}_{\mathbb{x}_u}$ and
$\mathscr{T}_{\mathbb{y}_u}$.
Then we have
\begin{equation*}
  \mathbb{c} \mapsto
  \begin{cases}
    \ch (\gamma(\mathsf{X}_u))
    ,
    & \text{when $\mathbb{c}=\mathscr{T}(\mathbb{x}_u)$,}
    \\
    \ch (\gamma(\mathcal{O}_{(r,s)}))
    ,
    & \text{when $\mathbb{c}=\mathscr{T}(\mathbb{c}_{(r,s)})$,}
  \end{cases}
\end{equation*}
where
$\gamma$ denotes automorphisms of DAHA $\mathcal{H}_{q_u,t}$ induced
from $\mathscr{T}$.
We define
\begin{equation}
  {P}_n(x_u, x_\ell, x_r, q_u, t; \mathbb{c})
  =
  \begin{cases}
    M_{n-1}(
    \gamma(\mathsf{X}_u) ; q_u, t)
    (1) ,
    & \text{when $\mathbb{c}=\mathscr{T}(\mathbb{x}_u)$,}
    \\
    M_{n-1}(
    \gamma(\mathcal{O}_{(r,s)}) ; q_u, t)
    (1) ,
    & \text{when $\mathbb{c}=\mathscr{T}(\mathbb{c}_{(r,s)})$.}    
  \end{cases}
\end{equation}

We show some concrete examples.
In the following~$\tau_{\bullet(u)}$ denotes
the automorphisms~\eqref{tau_for_A1} for~$\mathsf{X}_u$
and~$\mathsf{Y}_u$ of~$A_1$-DAHA,
which correspond to the Dehn twists about the curve $\mathbb{x}_u$ and
$\mathbb{y}_u$ respectively.
The first example is
\begin{equation*}
 \mathbb{c}^\prime_{(2k+1,2)}
  =
  \left(
    \mathscr{T}_{\mathbb{x}_u}^{~-k} \circ
    \mathscr{T}_{\mathbb{y}_u}^{~2}
  \right)
  ( \mathbb{x}_u),
\end{equation*}
  In the DAHA representation we have the automorphism
  $ \tau_{R(u)}^{~k} \circ \tau_{L(u)}^{~2} $ on
  $\mathsf{X}_u$~\eqref{tau_tau_X}
  to obtain the same results for torus knots
  $\mathbb{c}_{(2k+1,2)}$  on the once-punctured torus,
  \begin{equation}
    \mathbb{c}^\prime_{(2k+1,2)}
    \mapsto
    \widehat{M}_{(2k+1,2)}^{(0)}(x_u; q_u, t),
  \end{equation}
  where $\widehat{M}_{(2k+1,2)}^{(0)}(x;q, t)$ is defined
  in~\eqref{M0_2k1}.
  We thus have
  \begin{equation*}
    {P}_2
    (x_u, x_\ell, x_r, q_u, t; \mathbb{c}^\prime_{(2k+1,2)})
    =
    \widehat{M}_{(2k+1,2)}^{(0)}(x_u; q_u, t) (1)
    ,
  \end{equation*}
  which reduces to
  \begin{equation}
    {P}_2
    \left(
      x_u=-q, x_\ell=-q^{-1}, x_r=-q^{-1}, q , t=q;
      \mathbb{c}^{\prime}_{(2k+1,2)}
    \right)
    =
    1-q^{4k}-q^{4k+2}-q^{4k+4} ,
  \end{equation}
  which is the Jones polynomial for torus knot $T_{(2k+1,2)}$ as in~\eqref{P2_and_torus_knot}.

  As the second example, we treat
  \begin{equation*}
   \mathbb{c}^{\prime\prime}_{(2k+1,2)}
  =
  \left(
    \mathscr{T}_{\mathbb{x}_u}^{~-k} \circ
    \mathscr{T}_{\mathbb{y}_u}^{~2}
  \right)
  (\mathbb{y}),
\end{equation*}
We need
the action
$  \tau_{R(u)}^{~k} \circ \tau_{L(u)}^{~2} $
on the  representation~\eqref{rep11_xu} of
$\mathbb{y}$ arising from the DAHA
$\mathcal{H}_{q,\boldsymbol{t}_\natural}$
of $C^\vee C_1$ type.
  Recalling~\eqref{auto_and_Dehn},
  we have
  \begin{multline}
    \label{tau_on_yaw}
    \mathbb{c}^{\prime\prime}_{(2k+1,2)}
    \mapsto
    \frac{q_u  \, t^2 \left(x_\ell - x_r t^2 \right)
      \left(x_\ell x_r - t^2 \right)
    }{
      \left( 1+t^2 \right) \left( q_u^{~2} -t^4 \right)
      x_\ell \, x_r
    }      \,
    \sh \left(
       t\, q_u^{~k-1} (\mathsf{X}_u^{~k} \mathsf{Y}_u )^{2}
      \mathsf{X}_u
    \right) \,
    \sh (t^{-1}\mathsf{X}_u) \,
    \eth_t
    \\
    - \frac{q_u \left( q_u^{~2} x_r - x_\ell \,  t^2 \right)
      \left( q_u^{~2} - x_\ell \,  x_r \, t^2 \right)
    }{
      \left( t^2+q_u^{~2} \right)
      \left(  q_u^{~2}-t^4 \right)
      x_\ell \, x_r
    } \, \eth_t^{~-1} \,
    \frac{1}{
      \sh \left(t^{-1}\mathsf{X}_u \right)
    } \, \sh
    \left( t^{-1} q_u^{~k-1}(\mathsf{X}_u^{~k}\mathsf{Y}_u)^{2} \mathsf{X}_u \right)
    \\
    -
    \frac{
      q_u \,  t^2 \left(x_\ell+x_r\right)
      \left(1+x_\ell  \, x_r \right)
    }{
      \left(1+t^2 \right)
      \left( t^2 + q_u^{~2} \right)
      x_\ell  \, x_r 
    } \,
    \ch
    \left(
      q_u^{~k-1} (\mathsf{X}_u^{~k} \mathsf{Y}_u)^{2}\mathsf{X}_u
    \right)
    .
  \end{multline}
  This can be written as an operator on $\mathbb{C}[x_u+x_u^{~-1}]$
  as
  \begin{multline}
    \label{tau_on_yaw_2}
    \mathbb{c}^{\prime\prime}_{(2k+1,2)}
    \mapsto
    -
    \frac{q_u t^2 \left(x_\ell - x_r t^2 \right)
      \left(x_\ell x_r- t^2 \right)
    }{
      \left( 1+t^2 \right) \left( q_u^{~2}-t^4 \right)
      x_\ell \, x_r
    }      \,
    \widehat{M}_{(2k+1,2)}^{(+)}(x_u; q_u,t) \,
    \eth_t
    \\
    - \frac{q_u \left( q_u^{~2} x_r - x_\ell \,  t^2 \right)
      \left( q_u^{~2} - x_\ell \,  x_r \, t^2 \right)
    }{
       \left( t^2+q_u^{~2} \right)
      \left( q_u^{~2}-t^4 \right)
      x_\ell \, x_r
    } \, \eth_t^{~-1} \,
    \widehat{M}_{(2k+1,2)}^{(-)}(x_u;q_u,t) \,
    \\
    -
    \frac{
      q_u \,  t^2 \left(x_\ell+x_r\right)
      \left(1+x_\ell  \, x_r \right)
    }{
      \left(1+t^2 \right)
      \left( t^2 + q_u^{~2} \right)
      x_\ell  \, x_r 
    } \,
    \widehat{M}_{(2k+1,2)}^{(0)}(x_u;q_u,t)
    ,
  \end{multline}
  where we use~\eqref{M0_2k1} and
  the operators $\widehat{M}_{(2k+1,2)}^{(\pm)}(x;q,t)$ are
  given by
  \begin{align}
    \notag
    &
      \begin{multlined}[b]
        \widehat{M}_{(2k+1,2)}^{(+)}(x;q,t)
        =
        \left\{
          \left( q \, x \right)^{2k}
          \frac{
            \left( 1-t^2 x^2 \right) \left( 1- q^2 t^2 x^2 \right)
            \left(1-q^4 t^2 x^2 \right)
          }{
            q \, t^4 \left(1-x^2 \right) \left( 1-q^2 x^2 \right)
          } 
        \right\} \eth^2
        + \left\{ \text{$x \to x^{-1}$}
        \right\} \, \eth^{-2}
        \\
        + \frac{q \left( 1+q^2 \right) \left( 1-t^2 \right)
          \left(t^2-x^2 \right) \left( 1- t^2 x^2\right)
        }{
          t^4 \left( q^2-x^2 \right) \left(1-q^2 x^2 \right)
        }
        ,
      \end{multlined}
    \\
    &
      \widehat{M}_{(2k+1,2)}^{(-)}(x;q,t)
      =
      \left\{
      \left(q \, x \right)^{2k}
      \frac{q \, x^2 \left(-1+q^2 t^2 x^2 \right)}{
      \left( 1- x^2 \right) \left( 1- q^2 x^2 \right)
      } 
      \right\} \eth^{2}
      +
      \left\{ \text{$x\to x^{-1} $} \right\} \, \eth^{-2}
      +
      \frac{q \left(1+q^2 \right) \left( 1-t^2 \right) x^2}{
        \left( q^2 - x^2 \right)
        \left(1-q^2 x^2 \right)
      }
      .
      \label{Mmp_2k1}
      %
  \end{align}
  Here $\eth$ is a $q$-difference operator for $x$, and
  $\left\{ x\to x^{-1} \right\}$ in the second term
  denotes  the coefficient $\{\cdots\}$
  of
  $\eth^2$ in the first term
  replacing $x$ by $x^{-1}$ so that the operators
  $\widehat{M}^{(\pm)}(x;q,t)$ preserve the symmetric Laurent polynomial space
  $\mathbb{C}[x+x^{-1}]$.
  Then we obtain the DAHA polynomial for   $\mathbb{c}^{\prime\prime}_{(2k+1,2)}$
  by acting~\eqref{tau_on_yaw_2} on~$1$.
  This is  indeed a deformation of the Jones polynomial for torus knot
  $T_{(2k+1,2)}$;
  it
 reduces to the Jones polynomial for $\mathbb{c}_{(2k+1,2)}$
  up to framing when we take specific values for deformation parameters,
  \begin{equation}
    {P}_2
    \left(
      x_u=-q, x_\ell=-q^{-1}, x_r=-q^{-1}, q, t=q;
      \mathbb{c}^{\prime\prime}_{(2k+1,2)}
    \right)
    =
    1-q^{4k}-q^{4k+2}-q^{4k+4} .
  \end{equation}

\section{Genus-Two Torus}
\label{sec:genus-two}
\subsection{Skein Algebra}
We shall construct
the skein algebra $\KBS_A(\Sigma_{2,0})$ based on the previous
sections.
We define simple closed curves on $\Sigma_{2,0}$ as in
Fig.~\ref{fig:genus_two_surface}.
We regard
the  twice-punctured torus $\Sigma_{1,2}$ in
Fig.~\ref{fig:torus_two_punctures}
as a  subsurface of $\Sigma_{2,0}$, and
we can construct
$\Sigma_{2,0}=\Sigma_{1,2} \cup S^1 \times [0,1]$
where both
$S^1 \times \{ 0\} \approx \mathbb{b}_3$ and
$S^1 \times \{ 1\} \approx \mathbb{b}_4$
are isotopic to $\mathbb{x}_d$.
As a reduction of the $C^\vee C_1$-type DAHA~\eqref{algS04},
the simple closed curves $\mathbb{x}$,  $\mathbb{y}$, $\mathbb{z}$ 
constitute the following
skein algebra~\eqref{04type_genus2} of~$\Sigma_{0,4}$-type.
Here  $\mathbb{z}$ is generated by $\mathbb{x}$ and $\mathbb{y}$ as in~\eqref{04type_genus2}.
On the other hand,
we can regard $\Sigma_{2,0}$ as a union of two once-punctured tori,
$\Sigma_{2,0}=\Sigma^u_{1,1} \cup \Sigma^d_{1,1}$,
where $\Sigma^u_{1,1}\cap \Sigma^d_{1,1} \approx \mathbb{x}$,
and
the curves
$\mathbb{x}_u$ and $\mathbb{y}_u$
(resp. $\mathbb{x}_d$ and $\mathbb{y}_d$)
satisfy
the skein algebra~\eqref{torus11_type_genus2}
on the once-puncture torus $\Sigma^u_{1,1}$
(resp. $\Sigma^d_{1,1}$).
Two sets of curves $\{ \mathbb{x}_\circledcirc,
\mathbb{y}_\circledcirc,
\mathbb{z}_\circledcirc \}$
for $\circledcirc\in \{u, d\}$
fulfill the skein algebra of
$\Sigma_{1,1}$ whose boundary circle is isotopic to $\mathbb{x}$.
Further we need skein relations for
$\widetilde{\mathbb{y}}$ in Fig.~\ref{fig:torus_two_punctures}, and
they can be written in~\eqref{condition_ya}.

\begin{figure}[htbp]
  \centering
      \includegraphics[scale=1.]{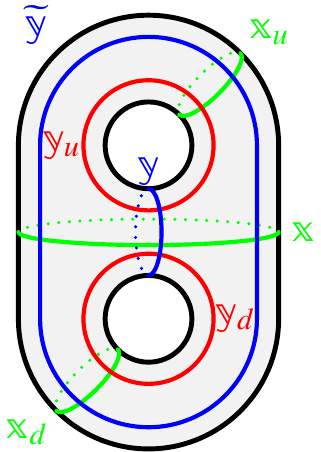}
  \caption{Simple closed curves on the genus-two torus $\Sigma_{0,2}$ are given.}
  \label{fig:genus_two_surface}
\end{figure}

In summary, we have the following skein algebra.
See also~\cite{ArthaShaki17a}.
\begin{prop}
  The skein algebra $\KBS_A(\Sigma_{2,0})$
  is
  \begin{itemize}
  \item $\Sigma_{0,4}$-type,
    \begin{align}
      \label{04type_genus2}      
        A^2 \, 
        \mathbb{x} \,  \mathbb{y}
        -A^{-2} \,
        \mathbb{y}  \,  \mathbb{x}
        & =
        \left( A^4-A^{-4} \right) \,
        \mathbb{z}
        + 2 \left( A^2-A^{-2} \right) \,
        \mathbb{x}_d \,
        \mathbb{x}_u
        ,
        \\
        A^2 \, 
        \mathbb{y} \,  \mathbb{z}
        -A^{-2} \,
        \mathbb{z}  \,  \mathbb{y}
        & =
        \left( A^4-A^{-4} \right) \,
        \mathbb{x}
        + \left( A^2-A^{-2} \right) \,
        \left(
          \mathbb{x}_u^{~2}
          +
          \mathbb{x}_d^{~2}
        \right)
          ,
          \notag
        \\
        A^2 \, 
        \mathbb{z} \,  \mathbb{x}
        -A^{-2} \,
        \mathbb{x}  \,  \mathbb{z}
        & =
        \left( A^4-A^{-4} \right) \,
        \mathbb{y}
        +
        2 \left( A^2-A^{-2} \right) \,
        \mathbb{x}_d
        \, \mathbb{x}_u
          ,
          \notag
    \end{align}
    with
    \begin{multline}
      A^2 \, \mathbb{x} \, \mathbb{y} \, \mathbb{z} =
      A^4 \,  \mathbb{x}^{2}
      + A^{-4} \, \mathbb{y}^{2}
      + A^4 \, \mathbb{z}^{2}
      +A^2  \,
      \left( \mathbb{x}_u^{~2}
        + \mathbb{x}_d^{~2}   \right)
      \, \mathbb{x}
      + 2 \, A^{-2} \,
      \mathbb{x}_d     \,    \mathbb{x}_u
      \, \mathbb{y}
      +2 \, A^2 \,
      \mathbb{x}_d \,    \mathbb{x}_u
      \, \mathbb{z}
      \\
      +
      2 \,  \mathbb{x}_u^{~2}
      + 2 \,  \mathbb{x}_d^{~2}
      +
      \mathbb{x}_u^{~2} \,  \mathbb{x}_d^{~2}
      -\left( A^2+A^{-2} \right)^2
      ,
    \end{multline}

  \item $\Sigma_{1,1}$-type,
    \begin{align}
      \label{torus11_type_genus2}
      A \, \mathbb{x}_\circledcirc \, \mathbb{y}_\circledcirc
      - A^{-1} \,  \mathbb{y}_\circledcirc \, \mathbb{x}_\circledcirc
      &=
        \left( A^2-A^{-2} \right)  \, \mathbb{z}_\circledcirc
        ,
      \\
      A  \,  \mathbb{y}_\circledcirc \,  \mathbb{z}_\circledcirc
      - A^{-1} \,  \mathbb{z}_\circledcirc \,  \mathbb{y}_\circledcirc
      & =
        \left( A^2-A^{-2} \right) \, \mathbb{x}_\circledcirc
        ,
        \notag
      \\
      A  \, \mathbb{z}_\circledcirc \, \mathbb{x}_\circledcirc
      - A^{-1}\,  \mathbb{x}_\circledcirc  \,  \mathbb{z}_\circledcirc
      & =
        \left( A^2-A^{-2} \right) \, \mathbb{y}_\circledcirc
        ,
        \notag
    \end{align}
    with
    \begin{equation}
      \label{11_generate_x_genus2}
      \mathbb{x}
      =
      A \, \mathbb{x}_\circledcirc \, \mathbb{y}_\circledcirc \, \mathbb{z}_\circledcirc -
      A^2 \, \mathbb{x}_\circledcirc^{~2} - A^{-2} \, \mathbb{y}_\circledcirc^{~2}
      - A^2 \, \mathbb{z}_\circledcirc^{~2}
      +A^2+A^{-2}
      ,
    \end{equation}
    where $\circledcirc\in \{u, d\}$,

  \item consistency,
    \begin{gather}
      \label{consistency_genus2}
        - \mathbb{y}^2 \mathbb{y}_\circledcirc
        + \left(
          A^2+A^{-2}
        \right)
        \mathbb{y} \, \mathbb{y}_\circledcirc \, \mathbb{y}
        - \mathbb{y}_\circledcirc \, \mathbb{y}^2
        =
        \left(
          A^2- A^{-2}
        \right)^2 \mathbb{y}_\circledcirc ,
        \\
        - \mathbb{y}_\circledcirc^{~2} \mathbb{y}
        + \left(
          A^2+A^{-2}
        \right)
        \mathbb{y}_\circledcirc \,  \mathbb{y} \, \mathbb{y}_\circledcirc
        - \mathbb{y}  \, \mathbb{y}_\circledcirc^{~2}
        =
        \left(
          A^2- A^{-2}
        \right)^2 \mathbb{y}
        ,
        \notag
      \end{gather}
      with
    \begin{align}
      \label{genus2_x_y}
      &
        \mathbb{x} \, \mathbb{y}_\circledcirc
        = \mathbb{y}_\circledcirc \, \mathbb{x}
        ,
        &&
           \mathbb{x}_\circledcirc \, \mathbb{y}
           = \mathbb{y} \, \mathbb{x}_\circledcirc
           ,
    \end{align}
    where $\circledcirc\in \{u,d \}$,
    and
    \begin{align*}
      &
        \mathbb{x}_u \, \mathbb{x}_d = \mathbb{x}_d \, \mathbb{x}_u
        ,
        &&
        \mathbb{x}_u \, \mathbb{x}_d = \mathbb{x}_d \, \mathbb{x}_u
        ,
      \\
      &
        \mathbb{y}_u \, \mathbb{x}_d = \mathbb{x}_d \, \mathbb{y}_u
        ,
        &&
        \mathbb{y}_u \, \mathbb{x}_d = \mathbb{x}_d \, \mathbb{y}_u
           ,
    \end{align*}

  \item skein relations for $\widetilde{\mathbb{y}}$,
    \begin{gather}
      \label{condition_ya}
        \widetilde{\mathbb{y}} \, \mathbb{y}_\circledcirc
        = \mathbb{y}_\circledcirc \, \widetilde{\mathbb{y}}
        ,
        \\
        - \widetilde{\mathbb{y}}^{2} \,  \mathbb{x}_\circledcirc
        +
        \left( A^2+ A^{-2} \right)
        \widetilde{\mathbb{y}} \, \mathbb{x}_\circledcirc \,
        \widetilde{\mathbb{y}}
        - \mathbb{x}_\circledcirc \, \widetilde{\mathbb{y}}^2
        =
        \left( A^2 - A^{-2} \right)^2 \mathbb{x}_\circledcirc
        ,
        \notag
        \\
        - \mathbb{x}_\circledcirc^{~2} \, \widetilde{\mathbb{y}}
        +\left( A^2 + A^{-2} \right)
        \mathbb{x}_\circledcirc \, \widetilde{\mathbb{y}} \, \mathbb{x}_\circledcirc
        - \widetilde{\mathbb{y}} \, \mathbb{x}_\circledcirc^{~2}
        =
        \left( A^2 - A^{-2} \right)^2 \widetilde{\mathbb{y}}
        .
        \notag
      \end{gather}
      where $\circledcirc\in \{u, d\}$,
    and
    \begin{gather}
      \label{skein_y_yall}
        \widetilde{\mathbb{y}} \, \mathbb{x}_u \, \mathbb{y}_u
        - \mathbb{y}_u \, \mathbb{x}_u \, \widetilde{\mathbb{y}}
        =
        \mathbb{y} \, \mathbb{y}_d \, \mathbb{x}_d
        - \mathbb{x}_d \, \mathbb{y}_d \, \mathbb{y}
        ,
        \\
        \widetilde{\mathbb{y}}  \, \mathbb{y}
        = \mathbb{y} \, \widetilde{\mathbb{y}}
        ,
        \notag
      \end{gather}
  \end{itemize}
\end{prop}

Note that the above $\Sigma_{1,1}$-type
relations~\eqref{torus11_type_genus2} are redundant, and we have
\begin{gather}
  \label{genus_skein_1}
    - \mathbb{y}_\circledcirc^{~2} \mathbb{x}_\circledcirc
    + \left(
      A^2+A^{-2}
    \right)
    \mathbb{y}_\circledcirc \mathbb{x}_\circledcirc \mathbb{y}_\circledcirc
    - \mathbb{x}_\circledcirc \mathbb{y}_\circledcirc^{~2}
    =
    \left(
      A^2- A^{-2}
    \right)^2 \mathbb{x}_\circledcirc,
    \\
    - \mathbb{x}_\circledcirc^{~2} \mathbb{y}_\circledcirc
    + \left(
      A^2+A^{-2}
    \right)
    \mathbb{x}_\circledcirc \mathbb{y}_\circledcirc \mathbb{x}_\circledcirc
    - \mathbb{y}_\circledcirc \mathbb{x}_\circledcirc^{~2}
    =
    \left(
      A^2- A^{-2}
    \right)^2 \mathbb{y}_\circledcirc
    .
    \notag
  \end{gather}


\subsection{Polynomial Representation}
To give a representation,
we make use of the DAHA representation for the twice-punctured torus in
the previous section.
We   glue the two punctures of $\Sigma_{1,2}$ in
Fig.~\ref{fig:torus_two_punctures}
by
setting
\begin{equation}
  x_\ell = x_r
  = x_d
  .
\end{equation}
It  should be remarked
that four parameters
$\boldsymbol{t}$  of DAHA of $C^\vee C_1$-type
$\SH_{q,\boldsymbol{t}}$
(see Fig.~\ref{fig:4-hole_sphere})
are now set to be
\begin{equation}
  \label{t_star}
  (t_0,t_1,t_2,t_3)
  =
  \left(
    \I \, x_u, \I \, q^{-\frac{1}{2}} x_d ,
    \I \, x_u , \I \, x_d
  \right)=
    \boldsymbol{t}_{\star}
  ,
\end{equation}
which means that we have glued the punctures,
$\mathbb{b}_1$ with $\mathbb{b}_2$, and
$\mathbb{b}_3$ with $\mathbb{b}_4$, together.
In gluing $\mathbb{b}_1$ with $\mathbb{b}_2$ in the previous
section, we have employed DAHA of $A_1$-type
$\SH_{q_u,\I q_u^{1/2} x^{1/2}}$  so that
$\mathbb{x}$ is generated from $\mathbb{x}_u$ and $\mathbb{y}_u$
as the boundary circle~\eqref{boundary_b} of $\Sigma_{1,1}^{u}$.
As we have $\Sigma_{1,1}^{u}\cap \Sigma_{1,1}^{d} \approx \mathbb{x}$,
another DAHA $\SH_{q_u,\I q_u^{1/2} x^{1/2}}$ plays the role of $\Sigma_{1,1}^d$
so  that
$\mathbb{x}$ is also generated from $\mathbb{x}_d$ and $\mathbb{y}_d$
as the boundary circle~\eqref{boundary_b} of $\Sigma_{1,1}^{d}$.
Thus
it is natural to use the gluing function~\eqref{gluing_function} for $\mathbb{y}_d$,
and to put
as follows;
\begin{align}
  \mathbb{x}
  &  \mapsto
    x + x^{-1}
    ,
    \notag
  \\
  \mathbb{y}
  &  \mapsto
    - \beta(x,x_u,x_d) \, \eth
    - \beta(x^{-1}, x_u , x_d) \, \eth^{-1}
    -
    \varphi(x,x_u,x_d)
    ,
    \notag
  \\
  \mathbb{x}_u
  &  \mapsto
    x_u + x_u^{~-1}
    ,
    \notag
  \\
  \mathbb{y}_u
  &      \mapsto
    \begin{multlined}[t]
      G(x,x_u)^{-1} 
      \left( \gamma(x,x_u) \, \eth_u
        +
        \gamma(x,x_u^{~-1}) \, \eth_u^{~-1}
      \right) G(x,x_u)
      \\
      =\gamma_1(x_u) \, \eth_u +
      \gamma_2(x,x_u) \,
      \eth_u^{~-1}
      ,
    \end{multlined}
  \notag
  \\
  \mathbb{x}_d
  &  \mapsto
    x_d + x_d^{~-1}
    ,
    \notag
  \\
  \mathbb{y}_d
  & \mapsto
    \begin{multlined}[t]
      G(x,x_d)^{-1} 
      \left( \gamma(x,x_d) \, \eth_d
        +
        \gamma(x,x_d^{~-1}) \, \eth_d^{~-1}
      \right)  G(x,x_d)
      \\
      =
      \gamma_1(x_d) \, \eth_d + \gamma_2(x,x_d) \,  \eth_d^{~-1}
      ,
    \end{multlined}
  \label{glue_genus_two}
\end{align}
where
$\gamma(x,x_u)$, $\gamma_1(x_u)$,
and $\gamma_2(x, x_u)$ are defined in~\eqref{define_gamma}
and~\eqref{define_gamma_a}.
The functions $\beta(x,x_u,x_d)$ and $\varphi(x,x_u,x_d)$ come from
the Askey--Wilson operator~\eqref{AW_operator}
with $\boldsymbol{t}_\star$~\eqref{t_star},
\begin{gather*}
  \beta(x,x_u,x_d)
  =
  \frac{
    q^{-\frac{1}{2}} + x
  }{
    \left( 1- q^{\frac{1}{2}} x \right)
    \left( 1 - x^2 \right)
  } 
  \left(x_u + q^{\frac{1}{2}} x \, x_u^{~-1} \right) \,
  \left(x_d + q^{\frac{1}{2}} x \, x_d^{~-1} \right)
  ,
  \\
  \varphi(x,x_u,x_d)
  =
  \frac{2}{
    \left( q^{-\frac{1}{2}} - x^{-1} \right)
    \left(1 -q^{\frac{1}{2}} x \right)
  } 
  \left( x_u + x_u^{~-1} \right)
  \left( x_d + x_d^{~-1} \right)
  .
\end{gather*}
We  should note that $\mathbb{y}$ is represented by the Askey--Wilson
difference operator~\eqref{AW_operator},
while $\mathbb{y}_u$ and $\mathbb{y}_d$
are   the (conjugated)
$A_1$-type Macdonald difference operators~\eqref{A1_Macdonald}.
The $q$-difference operators $\eth_u$ and $\eth_d$ are
for
$x_u$ and $x_d$ respectively
\begin{gather*}
  \eth_u \, \eth_d = \eth_d \, \eth_u
  ,
  \\
  \eth_u \,
  f(x,x_u,x_d) = f(x,q_u x_u, x_d)
  ,
  \\
  \eth_d \,
  f(x,x_u,x_d) =
  f(x,x_u,q_d x_d )
  ,
\end{gather*}
where we mean  $q_u=q_d=q^{\frac{1}{2}}$~\eqref{qu_and_q}.
Note that the representation for $\mathbb{y}$ is symmetric
$x_u\leftrightarrow x_d$, and that the associated Askey--Wilson polynomial for
$\boldsymbol{t}_\star$~\eqref{t_star} is
written in a symmetric form by use of the Sears'
transformation formula~\cite{GaspRahm04} to~\eqref{def_AW_poly}
as
\begin{equation}
  P_m(x; q, x_u, x_d)
  =
  (-1)^m q^{- \frac{m}{2}}
  \frac{
    \left( \frac{q}{x_u^{~2}} ,
      \frac{q}{x_d^{~2}} , q; q \right)_m
  }{
    \left( \frac{q^{m+1}}{x_d^{~2} x_u^{~2}}; q \right)_m
  } \,
  {}_4 \phi_3
  \left[
    \begin{matrix}
      q^{-m} ,
      \
      \frac{q^{m+1}}{x_d^{~2} x_u^{~2}} ,
      \
      -q^{\frac{1}{2}} x ,
      \
      - q^{\frac{1}{2}} x^{-1}
      \\
      \frac{q}{x_d^{~2}} ,
      \
      \frac{q}{x_u^{~2}} ,
      \
      q
    \end{matrix}
    ;
    q, q
  \right]
  .
\end{equation}

By construction,
both skein relations of
the $\Sigma_{0,4}$-type~\eqref{04type_genus2}
and
the $\Sigma_{1,1}$-type~\eqref{torus11_type_genus2}
are fulfilled by
the above representation~\eqref{glue_genus_two}.
Moreover the above Askey--Wilson operator for $\mathbb{y}$ is symmetric
in $x_u \leftrightarrow x_d$, and
the consistency
conditions~\eqref{consistency_genus2} are satisfied thanks to
the results for the skein algebra on   $\Sigma_{1,2}$ in the previous section.
Hence a
non-trivial  is for the simple closed curve~$\widetilde{\mathbb{y}}$ in
Fig.~\ref{fig:torus_two_punctures}.
To give a representation, we suppose that
\begin{equation}
  \label{ansatz_for_yall}
  \widetilde{\mathbb{y}}
  \mapsto
  \sum_{\varepsilon\in \{-, 0, +\}}
  \sum_{ \varepsilon_u, \varepsilon_d=\pm 1}
  a_{\varepsilon, \varepsilon_u, \varepsilon_d}(x,x_u,x_d) \,
  \eth^{\varepsilon} \,
  \eth_u^{~\varepsilon_u} \,
  \eth_d^{~\varepsilon_d}
  ,
\end{equation}
and  consider a condition for
the first identity of~\eqref{skein_y_yall}.
Equating each coefficient of $\eth \, \eth_u^{~2} \, \eth_d^{-1}$,
$\eth\, \eth_u^{~-2}\,\eth_d^{-1}$,
$\eth \, \eth_d^{-1}$,
we have the following functional equations
for~$a_{+ \pm -}(x,x_u,x_d)$;
\begin{gather*}
  \frac{
    a_{++-}(x, q_u x_u, x_d)
  }{
    a_{++-}(x, x_u, x_d)
  }
  =
  \frac{\gamma_1(  q_u x_u)}{
    \gamma_1(x_u)}
  ,
  \\
  \frac{
    a_{+--}(x, q_u^{~-1} \, x_u, x_d)
  }{
    a_{+--}(x, x_u, x_d)
  }
  =
  \frac{
    \gamma_2(q x, q_u^{~-1}\, x_u)
  }{
    \gamma_2(x,x_u)
  }
  ,
  \\
  \begin{multlined}[b]
    \left( q_u^{~-1} x_u+q_u x_u^{~-1}\right)
    \left\{
      \gamma_1( q_u^{~-1} x_u) \,
      a_{+--}(x,x_u,x_d)
      -
      \gamma_2(x,x_u) \, a_{++-}(x,q_u^{~-1} x_u, x_d)
    \right\}
    \\
    -
    \left( q_u  \, x_u+ q_u^{~-1}  x_u^{~-1}\right)
    \left\{
      \gamma_1(x_u) \, a_{+--}(x, q_u \, x_u, x_d)
      -
      \gamma_2(q \, x, q_u x_u ) \,
      a_{++-}(x,x_u,x_d)
    \right\}
    \\
    =
    \I \,  \frac{
      (1-q)  \left( 1+q^{\frac{1}{2}} x  \right)
      \left( q^{\frac{1}{2}} x+x_d^{~2} \right)
      \left( q^{\frac{3}{2}} x+x_d^{~2} \right)
      \left( q^{\frac{1}{2}} x+x_u^{~2} \right)
    }{
      q^{\frac{9}{4}} x \left( 1-x^2 \right)
      \left( 1-x_d^{~2} \right) x_u
    } .
  \end{multlined}
\end{gather*}
The first two are solved to be
\begin{gather*}
  a_{+--}(x,x_u,x_d)
  =
  \tilde{a}_{+--}(x,x_d) \,
  \frac{
    \left( q^{\frac{1}{2}} x + x_u^{~2} \right)
    \left( q^{\frac{3}{2}} x + x_u^{~2} \right)
  }{
    1-x_u^{~2}
  } 
  ,
  \\
  a_{++-}(x,x_u,x_d)
  =
  \tilde{a}_{++-}(x,x_d) \, \frac{1}{1-x_u^{~2}}
  .
\end{gather*}
Due to that functions $\tilde{a}_{+++}$ and $\tilde{a}_{+-+}$ do not
depend on $x_u$,
we can solve them from the above third equations,
\begin{gather*}
  \tilde{a}_{+--}(x,x_d)
  =
  \frac{
    \left( 1+q^{\frac{1}{2}} x \right)
    \left( q^{\frac{1}{2}} x + x_d^{~2} \right)
    \left( q^{\frac{3}{2}} x + x_d^{~2} \right)
  }{
    q^{\frac{5}{2}} x \left( 1-x^2 \right)
    \left( 1-q^{\frac{1}{2}}x \right)
    \left( 1-x_d^{~2} \right)
  },
  \\
  \tilde{a}_{++-}(x,x_d)
  = -q \, x \,
  \tilde{a}_{+--}(x,x_d) .
\end{gather*}
In this manner, we get
$a_{\pm, \varepsilon_u, \varepsilon_d}(x,x_u,x_d)$.

For $a_{0,\pm,\pm}(x,x_u,x_d)$, we see that
a sum for $\varepsilon \neq 0$ in~\eqref{ansatz_for_yall} commute with
both~$\mathbb{y}_u$ and~$\mathbb{y}_d$.
Consequently
we can suppose that a sum for $\varepsilon=0$ in~\eqref{ansatz_for_yall}
has a form of 
\begin{equation*}
  \psi(x) \,
  \left(
    \gamma_1(x_d) \, \eth_d + \gamma_2(x, x_d) \, \eth_d^{~-1}
  \right)
  \left(
    \gamma_1(x_u) \, \eth_u + \gamma_2(x, x_u) \, \eth_u^{~-1}
  \right) .
\end{equation*}
A commutativity between $\mathbb{y}$ and
$\widetilde{\mathbb{y}}$ solves $\psi(x)$, and as a result we obtain
\begin{multline}
  \label{define_ya}
  \widetilde{\mathbb{y}}
  \mapsto
  \left(
     \kappa(x_d) \, \eth_d+    \lambda(x,x_d) \, \eth_d^{~-1}
  \right)
  \left(
     \kappa(x_u) \, \eth_u+    \lambda(x,x_u) \, \eth_u^{~-1}
  \right) \, \omega(x) \, \eth
  \\
  +
  \left(
    \kappa(x_d) \, \eth_d +    \lambda(x^{-1},x_d) \, \eth_d^{~-1}
  \right)
  \left(
    \kappa(x_u) \, \eth_u +    \lambda(x^{-1},x_u) \, \eth_u^{~-1}
  \right) \, \omega(x^{-1}) \, \eth^{-1}
  \\
  +
  \psi(x) \,
  \left(
    \gamma_1(x_d) \, \eth_d + \gamma_2(x, x_d) \, \eth_d^{~-1}
  \right)
  \left(
    \gamma_1(x_u) \, \eth_u + \gamma_2(x, x_u) \, \eth_u^{~-1}
  \right)
  ,
\end{multline}
where
\begin{align}
  \omega(x) & =
  \frac{x \left( 1 + q^{\frac{1}{2}} x \right)}{
    q^{\frac{1}{2}} \left( 1 - x^2 \right)
    \left( 1 - q^{\frac{1}{2}} x  \right)
  }
  ,
              &
  \psi(x) & = \frac{2 \, x}{
    \left( 1 - q^{-\frac{1}{2}} x \right)
    \left( 1 - q^{\frac{1}{2}} x \right)
  }
  ,
  \\
  \lambda(x,x_u) &=
  \frac{
    \left( q^{\frac{1}{2}} x + x_u^{~2} \right)
    \left( q^{\frac{3}{2}} x + x_u^{~2} \right)
  }{
    q\, x  \left(1 - x_u^{~2} \right)
  } ,
                   &
  \kappa(x_u) &= \frac{-1}{1-x_u^{~2}}
  .
                \notag
\end{align}
With the setting~\eqref{define_ya},
we can check the relations~\eqref{condition_ya} by tedious
computations.

As a result, we have obtained the following theorem.
\begin{theorem}
  We have a representation of
  $\KBS_A(\Sigma_{2,0})$
  by~\eqref{glue_genus_two},~\eqref{define_ya}
  with  $A^2=q^{-\frac{1}{2}}$.
\end{theorem}


This representation denotes the difference
operators on $\mathbb{C}[x+x^{-1}]$.
As representation on $\mathbb{C}[x_u+x_u^{-1}, x_d+x_d^{-1}]$,
we take a conjugation
$h \mapsto
G\,h \, G^{-1}$
with
\begin{equation}
  \label{GGGud}
  G=G(x,x_u) \, G(x,x_d) .  
\end{equation}
Using~\eqref{t_and_x_2},
we have the following representation.

\begin{coro}
  We have a representation of
  $\KBS_A(\Sigma_{2,0}) $
  with $A^2 = q_u^{~-1}$ by
  \begin{align}
    &
      \mathbb{x}_\circledcirc
      \mapsto
      \mathsf{X}_\circledcirc + \mathsf{X}_\circledcirc^{~-1}
      ,
      \notag
    \\
    &
      \mathbb{y}_\circledcirc
      \mapsto
      \mathsf{Y}_\circledcirc + \mathsf{Y}_\circledcirc^{~-1},
      \qquad \text{for $\circledcirc \in \{u, d \}$},
    \\
    \label{y_SH_t_qu}
    &
      \begin{multlined}[b]
        \mathbb{y}
        \mapsto
        -\frac{ q_u \, t^4 \left( 1-t^2 \right)}{
          \left( 1 + t^2 \right) \left(  q_u^{~2} - t^4 \right)}
        \left\{        
          \prod_{\circledcirc\in \{u,d \}}
          \sh \left( t \, \mathsf{X}_\circledcirc \right)
          \sh \left( t^{-1} \, \mathsf{X}_\circledcirc \right)
        \right\} \eth_t
        \\
        -\frac{q_u^{~3} \left( q_u^{~2} -t^2 \right)}{
           \left( q_u^{~2} +t^2 \right)
          \left(  q_u^{~2} -t^4 \right)
        } \, \eth_t^{~-1}
        -
        \frac{2 \, q_u \, t^2}{
          \left( t^2 + q_u^{~2}  \right) \left( t^2+1 \right)
        } \,
        \left\{
          \prod_{\circledcirc\in \{u,d\}}
          \ch \left( \mathsf{X}_\circledcirc \right)
        \right\}
        ,
      \end{multlined}
    \\
    \label{ya_SH_t_qu}
    &
      \begin{multlined}[b]
        \widetilde{\mathbb{y}}
        \mapsto
        \frac{q_u \, t^4 \left( 1-t^2 \right)}{
          \left( 1+t^2 \right) \left( q_u^{~2} -t^4 \right)}
        \left\{
          \prod_{\circledcirc\in\{u,d\}}
          t^{-1}
          \sh \left( t^{-1} \mathsf{Y}_\circledcirc \right)
          \sh \left(t^{-1} \mathsf{X}_\circledcirc \right)
        \right\} \eth_t
        \\
        +
        \frac{q_u^{~3} \left(q_u^{~2} - t^2 \right)}{
          \left( q_u^{~2} + t^2 \right)
          \left(  q_u^{~2}-t^4 \right)}
        \left\{
          \prod_{\circledcirc\in \{u,d \}}
          \frac{t}{
            \sh \left( t^{-1} \mathsf{X}_\circledcirc \right)}
          \sh \left( t \, \mathsf{Y}_\circledcirc \right)
        \right\}
        \eth_t^{~-1}
        \\
        -
        \frac{2 \, q_u \, t^2}{
          \left(  q_u^{~2} +t^2  \right)
          \left(  1+t^2  \right)
        } \,
        \left\{
          \prod_{\circledcirc \in \{u, d\}}
          \ch \left( \mathsf{Y}_\circledcirc \right)
        \right\}
        .
      \end{multlined}
  \end{align}
Here
$\left\{
  \mathsf{X}_\circledcirc,
  \mathsf{Y}_\circledcirc,
  \mathsf{T}_\circledcirc
\right\}$
are generators of $\mathcal{H}_{q_u,t}$, and
these representations preserve
symmetric space  $\mathbb{C}[x_u+x_u^{~-1}, x_d+x_d^{~-1}]$.
Recalling~\eqref{lower_raise},
these are explicitly written as operators on the symmetric polynomial space,
\begin{align}
  &
    \mathbb{x}_\circledcirc
    \mapsto
    x_\circledcirc + x_\circledcirc^{~-1}
    ,
    \notag
  \\
  &
    \mathbb{y}_\circledcirc
    \mapsto
    \frac{ t\, x_\circledcirc - t^{-1} x_\circledcirc^{~-1}}{
    x_\circledcirc - x_\circledcirc^{~-1}
    }\, \eth_\circledcirc
    +
    \frac{t^{-1} x_\circledcirc - t \, x_\circledcirc^{~-1}}{
    x_\circledcirc - x_\circledcirc^{~-1}
    }
    \, \eth_\circledcirc^{~-1}
    ,
     \qquad
    \text{for $\circledcirc \in \{u, d\},$}
  \\
  &
    \mathbb{x}
    \mapsto
    -q_u t^{-2}-q_u^{~-1} t^2
    ,
    \notag
  \\
  &
    \begin{multlined}[b]
      \mathbb{y}
      \mapsto
      -\frac{ q_u \, t^4 \left( 1-t^2 \right)}{
        \left( 1+ t^2 \right) \left(  q_u^{~2} - t^4 \right)}
      \left\{        
        \prod_{\circledcirc\in \{u,d \}}
        \left( t \, x_\circledcirc - t^{-1}x_\circledcirc^{~-1}\right)
        \left( t^{-1} \, x_\circledcirc - t \, x_\circledcirc^{~-1}\right)
      \right\} \eth_t
      \\
      -\frac{q_u^{~3} \left( q_u^{~2} - t^2 \right)}{
         \left( q_u^{~2} + t^2 \right)
        \left(  q_u^{~2} -t^4 \right)
      } \, \eth_t^{~-1}
      -
      \frac{2 \, q_u \, t^2}{
        \left( q_u^{~2} + t^2 \right) \left( 1+t^2 \right)
      } \,
      \prod_{\circledcirc\in \{u,d\}}
      \left( x_\circledcirc + x_\circledcirc^{~-1} \right)
      ,
    \end{multlined}
  \\
  &
    \begin{multlined}
      \widetilde{\mathbb{y}}
      \mapsto
      \frac{q_u \, t^4\left( 1- t^2 \right)}{
        \left( 1+ t^2 \right) \left( q_u^{~2} - t^4 \right)}
      \left\{
        \prod_{\circledcirc\in\{u,d\}}
        \left(
          \frac{
            \left( 1 - t^2 x_\circledcirc^{~2} \right)
            \left( 1-q_u^{~2} t^2 x_\circledcirc^{~2} \right)
          }{
            q_u\, t^2 x_\circledcirc \left( x_\circledcirc^{~2} -1 \right)
          } \,
          \eth_\circledcirc
          -
          \frac{
            \left( t^2-x_\circledcirc^{~2} \right)
            \left( t^2 q_u^{~2} - x_\circledcirc^{~2} \right)
          }{
            q_u\, t^2 x_\circledcirc \left( x_\circledcirc^{~2} -1 \right)
          } \, 
          \eth_\circledcirc^{~-1}
        \right)
      \right\} \eth_t
      \\
      +
      \frac{q_u^{~3} \left(q_u^{~2} - t^2 \right)}{
        \left( q_u^{~2} +t^2 \right)
        \left(  q_u^{~2} -t^4 \right)}
      \left\{
        \prod_{\circledcirc\in \{u,d \}}
        \frac{x_\circledcirc}{x_\circledcirc^{~2}-1}
        \left( \eth_\circledcirc - \eth_\circledcirc^{~-1} \right)
      \right\}
      \eth_t^{~-1}
      \\
      -
      \frac{2 \, q_u \, t^2}{
        \left( q_u^{~2} + t^2 \right)
        \left( 1 + t^2 \right)
      } \,
      \left\{
        \prod_{\circledcirc \in \{u, d\}}
        \left(
          \frac{ t\, x_\circledcirc - t^{-1} x_\circledcirc^{~-1}}{
            x_\circledcirc - x_\circledcirc^{~-1}
          }\, \eth_\circledcirc
          +
          \frac{t^{-1} x_\circledcirc - t \, x_\circledcirc^{~-1}}{
            x_\circledcirc - x_\circledcirc^{~-1}
          }
          \, \eth_\circledcirc^{~-1}
        \right)
      \right\}
      .
    \end{multlined}
\end{align}
\end{coro}


We should note that the representation~\eqref{ya_SH_t_qu} of
$\widetilde{\mathbb{y}}$
can be
recovered by use of the automorphisms of $\mathcal{H}_{q_u,t}$.
As shown  in Fig.~\ref{fig:yall_auto}, we have
$\widetilde{\mathbb{y}}
=\mathscr{T}_{\mathbb{x}_u}\mathscr{T}_{\mathbb{x}_d}^{~-1}
\mathscr{T}_{\mathbb{y}_u}\mathscr{T}_{\mathbb{y}_d}^{~-1}(\mathbb{y})$,
and the DAHA operator for
$\mathbb{y}=\mathbb{c}_{(0,1)}$ is
$\mathbb{c}_{(0,1)} \mapsto \mathcal{A}_{(0,1)}=\ch( \mathsf{Y})$ in~\eqref{y_SH_t_qu}.
Using~\eqref{auto_and_Dehn}, an action
$\tau_{R(u)}^{~-1} \tau_{R(d)} \tau_{L(u)} \tau_{L(d)}^{~-1}$
is
\begin{equation*}
  \left(
    \begin{smallmatrix}
      \mathsf{X}_u \\
      \mathsf{X}_d \\
      \eth_t
    \end{smallmatrix}
  \right)
  \mapsto
  \left(
    \begin{smallmatrix}
      q^{-1} \mathsf{X}_u^{~-1} \mathsf{Y}_u \mathsf{X}_u
      \\
      \mathsf{Y}_d^{~-1}
      \\
      \frac{1}{
        \sh \left(
          t^{-1} q^{-1} \mathsf{X}_u^{~-1} \mathsf{Y}_u
          \mathsf{X}_u
        \right)
      }
      \sh \left( t^{-1} \mathsf{X}_u \right)
      \frac{1}{
        \sh \left( t^{-1} \mathsf{Y}_d^{~-1} \right)
      }
      \sh  \left( t^{-1} \mathsf{X}_d \right) \,
      \eth_t
    \end{smallmatrix}
  \right) .
\end{equation*}
We then  recover~\eqref{ya_SH_t_qu} from~\eqref{y_SH_t_qu} as operators
on $\SH_{q_u,t}\times \SH_{q_u,t}$.

As depicted  in
Fig.~\ref{fig:yall_auto}, the simple closed
curve $\widetilde{\mathbb{y}}$ is also given
from the curve $\mathbb{c}_{(1,1)}$ by
$\widetilde{\mathbb{y}}
= \mathscr{T}_{\mathbb{x}_u}\mathscr{T}_{\mathbb{x}_d}\mathscr{T}_{\mathbb{y}_u}\mathscr{T}_{\mathbb{y}_d}
(\mathbb{c}_{(1,1)})$.
The associated operator 
$\mathbb{c}_{(1,1)}\mapsto \mathcal{A}_{(1,1)}
=
\ch(\mathsf{T}_1 \mathsf{T}^\vee_0)
\in \SH_{q,\boldsymbol{t}_\star}$ 
is explicitly written as the operator on $\mathbb{C}[x+x^{-1}]$,
\begin{multline*}
  \left. \mathcal{A}_{(1,1)} \right|_{\text{sym}}
  =
  -\frac{
    x \left( 1+q^{\frac{1}{2}} x \right)
    \left(q^{\frac{1}{2}} x \, x_u^{~-1}+ x_u \right)
    \left(q^{\frac{1}{2}} x \, x_d^{~-1}+ x_d \right)
  }{
    \left( 1-x^2 \right)
    \left( 1-q^{\frac{1}{2}}x \right)}
  \,
  \eth
  \\
  -
  \frac{
    \left( q^{\frac{1}{2}} + x \right)
    \left( q^{\frac{1}{2}} x_u^{~-1} + x \, x_u \right)
    \left( q^{\frac{1}{2}} x_d^{~-1} + x \, x_d \right)
  }{
    x \left( x^2-1 \right)
    \left( x-q^{\frac{1}{2}} \right)
  } \,
  \eth^{-1}
  +\frac{2 q^{\frac{1}{2}} x
    \left( x_u+x_u^{~-1} \right)
    \left( x_d+x_d^{~-1} \right)
  }{\left( q^{\frac{1}{2}}-x \right)
    \left( 1-q^{\frac{1}{2}}x \right)}
  ,
\end{multline*}
which follows from~\eqref{def_A11} with $\boldsymbol{t}_\star$~\eqref{t_star}.
We take a conjugation
$\mathcal{A}_{(1,1)}   \mapsto G \, \mathcal{A}_{(1,1)} \, G^{-1}$
with the gluing function~\eqref{GGGud},
to have an operator on
$\mathbb{C}[x_u+x_u^{~-1}, x_d+x_d^{~-1}]$.
Making a 
change of variables~\eqref{t_and_x_2},
we
get
\begin{multline}
  G \left. \mathcal{A}_{(1,1)} \right|_{\text{sym}} G^{-1}
  =
  \frac{
    q_u \, t^6 (1-t^2)
  }{ \left( q_u^{~2} - t^4  \right)  \left( 1+t^2 \right)}
  \prod_{\circledcirc \in \{u,d\}}
  \sh\left( t \, \mathsf{X}_\circledcirc \right)
  \sh \left( t^{-1} \, \mathsf{X}_\circledcirc \right) \,
  \eth_t
  \\
  +
  \frac{q_u^{~5} \left(  q_u^{~2} - t^2 \right)}{
    t^2 \left( q_u^{~2} - t^4 \right)
    \left( q_u^{~2} + t^2 \right)
  } \eth_t^{~-1}
  -
  \frac{2 \, q_u \, t^2}{
    \left( q_u^{~2}+ t^2 \right)
    \left( 1+ t^2 \right)
  }
  \prod_{\circledcirc \in \{ u, d \}}
  \ch   \left(  \mathsf{X}_\circledcirc  \right)
  .
\end{multline}
Applying the automorphism
$
\tau_{R(u)}^{~-1}\tau_{R(d)}^{~-1}
\tau_{L(u)}\tau_{L(d)}
$, 
we obtain~\eqref{ya_SH_t_qu} as well.

\begin{figure}[htbp]
  \centering
  \begin{tikzpicture}
    \node (BB) at (0,0) {
      \makecell{
        \includegraphics[scale=.7]{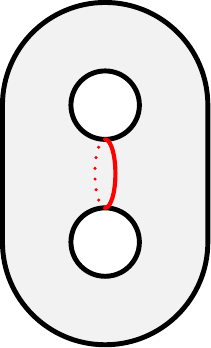}
      }};
    \node[right=16mm of BB] (BB1) {
      \makecell{
        \includegraphics[scale=.7]{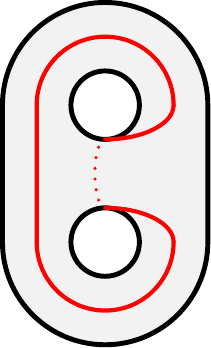}
      }};
    \node[right=16mm of BB1] (BB2) {
      \makecell{
        \includegraphics[scale=.7]{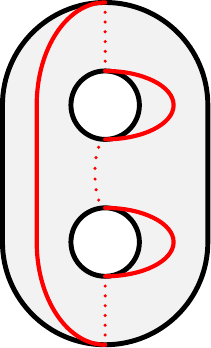}
      }};
    \node[below=3mm of BB] (AA)  {
      \makecell{
        \includegraphics[scale=.7]{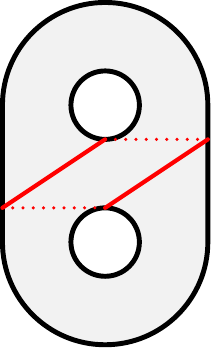}
      }};
    \node[right=16mm of AA] (AA1) {
      \makecell{
        \includegraphics[scale=.7]{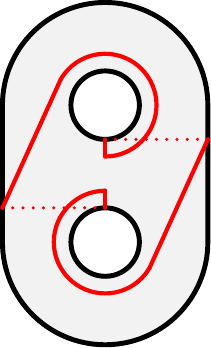}
      }};
    \node[right=16mm of AA1] (AA2) {
      \makecell{
        \includegraphics[scale=.7]{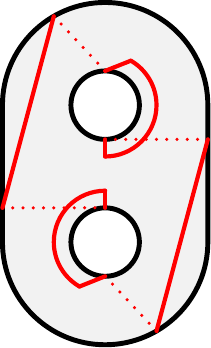}
      }};
    \node[right=8mm  of BB2] (AA3) {
      \makecell{
        \includegraphics[scale=.7]{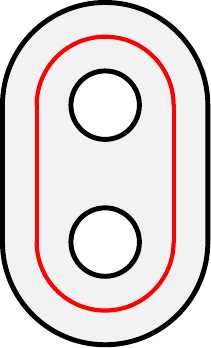}
      }};
    \draw[->] (BB) edge node[below] {\scriptsize
      $\mathscr{T}_{\mathbb{y}_u} \, \mathscr{T}_{\mathbb{y}_d}^{~-1}$}
    (BB1);
    \draw[->] (BB1) edge node[below] {\scriptsize
      $\mathscr{T}_{\mathbb{x}_u} \mathscr{T}_{\mathbb{x}_d}^{~-1}$
    } (BB2);
    \draw[white] (BB2) edge node[black] {$\approx$} (AA3);
    \draw[->] (AA) edge node[below] {\scriptsize $\mathscr{T}_{\mathbb{y}_u} \, \mathscr{T}_{\mathbb{y}_d}$} (AA1);
    \draw[->] (AA1) edge node[below] {\scriptsize
      $\mathscr{T}_{\mathbb{x}_u} \mathscr{T}_{\mathbb{x}_d}$
    } (AA2);
    \draw[white] (AA2) edge node[black,sloped] {$\approx$} (AA3);
  \end{tikzpicture}
  \caption{A simple closed curve $\widetilde{\mathbb{y}}$ is given
    from curves with slope-$1/0$ and $1/1$
    by Dehn twists.
  }
  \label{fig:yall_auto}
\end{figure}
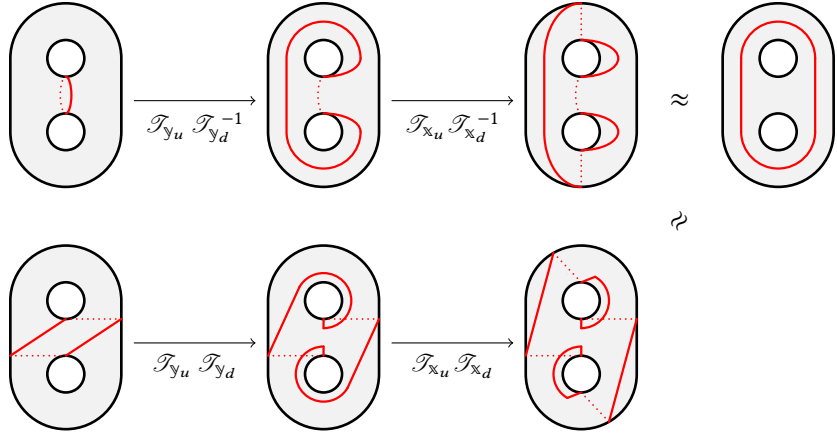

As a simple closed curve~$\widetilde{\mathbb{y}}$ does not intersect
with both~$\mathbb{y}_u$ and~$\mathbb{y}_d$,
the Dehn twists~$\mathscr{T}_{\mathbb{y}_u}$
and~$\mathscr{T}_{\mathbb{y}_d}$
have trivial
actions on~$\widetilde{\mathbb{y}}$,
$\widetilde{\mathbb{y}}=
\mathscr{T}_{\mathbb{y}_u}(\widetilde{\mathbb{y}})
=
\mathscr{T}_{\mathbb{y}_d}(\widetilde{\mathbb{y}})$.
Our representation for $\widetilde{\mathbb{y}}$
is indeed invariant under $\tau_{L(\circledcirc)}$
due to the following
proposition.

\begin{prop}
  \label{coro:A1_invariant_op}
  The operators of $\SH_{q,t}$,
  \begin{equation*}
    t^{-1} \sh (t^{-1}\mathsf{Y}) \,
    \sh (t^{-1}\mathsf{X}) \,
    \eth_t
    ,
  \end{equation*}
  and
  \begin{equation*}
    \frac{t}{\sh(t^{-1}\mathsf{X}) } \,
    \sh(t \,\mathsf{Y}) \,
    \eth_t^{~-1} ,
  \end{equation*}
  are invariant under $\tau_L$.
\end{prop}

\begin{proof}
  It is straightforward to see the invariance of the first operator
  using~\eqref{tau_for_A1}
  and~\eqref{partial_t_Tau}.
  For the second, we recall that
  it denotes the lowering operator
  $\mathsf{K}^{(-)}$ which does
  not depend on $t$, and that it commutes with $\eth_t$.
  We then have
  \begin{equation*}
    \eth_t^{~-1} \frac{t}{\sh \left( t^{-1} \mathsf{X} \right)}
    \sh \left( t \, \mathsf{Y}\right)
    \mapsto
    \eth_t^{~-1} \frac{1}{\sh(t^{-1}\mathsf{X})} \,
    \sh (t^{-1}q^{-\frac{1}{2}} \mathsf{Y} \, \mathsf{X} ) \,
    \frac{t}{
      \sh (t^{-1}q^{-\frac{1}{2}} \mathsf{Y} \, \mathsf{X} )
    } \,
    \sh ( t\, \mathsf{Y} ) .
  \end{equation*}
  which proves the invariance under $\tau_L$.
\end{proof}



\subsection{Double-Torus Knots}
A simple closed curve on a genus two Heegaard surface in $S^3$ is
called a double-torus knot.
A construction of a  non-trivial knot was studied
in~\cite{PHill99a,HillMura00a}, but
a classification  of the double-torus knots  is far from complete.

We shall propose a  DAHA polynomial for the double-torus knot.
We  assign a difference operator
for the simple
closed curve $\mathbb{c}$ as follows.
We suppose that $\mathbb{c}$ is given from ${\mathbb{c}}_{(r,s)}$ by
Dehn twists $\mathscr{T}$ which is generated by
$\mathscr{T}_{\mathbb{x}_u}$, $\mathscr{T}_{\mathbb{y}_u}$,
$\mathscr{T}_{\mathbb{x}_d}$, and $\mathscr{T}_{\mathbb{y}_d}$,
\begin{equation*}
  \mathbb{c}=\mathscr{T}(\mathbb{c}_{(r,s)}) ,
\end{equation*}
and that
a curve~${\mathbb{c}}_{(r,s)}$ does not intersect
with~$\mathbb{x}_u$,~$\mathbb{x}_d$,
Here
we mean that ${\mathbb{c}}_{(r,s)}$ is a simple closed curve on
$\Sigma_{0,4} \subset \Sigma_{2,0}$ with slope $s/r$,
when we regard $\Sigma_{2,0}=\Sigma_{0,4} \cup S^1\times[0,1] \cup
S^1\times[0,1]$ and
the boundary circles in Fig.~\ref{fig:4-hole_sphere} are
$\mathbb{b}_1 \approx \mathbb{b}_2 \approx \mathbb{x}_u$ and
$\mathbb{b}_3 \approx \mathbb{b}_4 \approx \mathbb{x}_d$.
Using the $C^\vee C_1$-type  DAHA $\SH_{q,\boldsymbol{t}_\star}$
with $\boldsymbol{t}_\star$~\eqref{t_star},
we have the $q$-difference operator
\begin{equation*}
  \mathbb{c}_{(r,s)} \mapsto
  \mathcal{A}_{(r,s)}=
  \ch( \mathcal{O}_{(r,s)} )
  \in \SH_{q,\boldsymbol{t}_\star}
  .
\end{equation*}
We then take a conjugation
$G \, \mathcal{A}_{(r,s)} \, G^{-1}$
by the gluing function~\eqref{GGGud},
and make a change of variables~\eqref{t_and_x_2}.
Applying the automorphism $\gamma$ associated to the Dehn twist
$\mathscr{T}$, we have the difference operator
\begin{equation*}
  \mathbb{c} \mapsto \gamma( G\, \mathcal{A}_{(r,s)} \, G^{-1} )
  .
\end{equation*}
We then  define the DAHA polynomial of
$\mathbb{c}$
by
\begin{align}
  \label{define_our_Poly}
  P_n(t, q_u, x_u, x_d; \mathbb{c})
  & =
    \gamma
    \left(
    G \, M_{n-1}(\mathcal{O}_{(r,s)}; q, q) \,
    G^{-1}
    \right) (1)
  \\
  & =
    \gamma
    \left(
    G \, S_{n-1}( \mathcal{A}_{(r,s)}) \, G^{-1}
    \right)
    (1)
    .
    \notag
\end{align}
It is noted that we can  further  take a conjugation by $\Phi$ to be determined,
\begin{equation*}
  \mathbb{c} \mapsto
  \Phi(q_u, t) \, \gamma( G\, \mathcal{A}_{(r,s)} \, G^{-1} ) \,
  \Phi(q_u,  t)^{-1} ,
\end{equation*}
and can modify the definition of the DAHA polynomial as
\begin{equation}
  \label{def_modify_P}
  P_n(t, q_u, x_u, x_d, \Phi; \mathbb{c})
  =
  \gamma \left(
    \Phi(q_u,t) \,
    G \, S_{n-1}(\mathcal{A}_{(r,s)}) \, G^{-1} \,
    \Phi(q_u, t)^{-1}
  \right) (1) .
\end{equation}
We expect that, by suitably choosing the function $\Phi$,
there may exist a relationship between our DAHA
polynomial and a Poincar{\'e} polynomial of knot homology
(see, \emph{e.g.},~\cite{FujiGukoSulk13a}), but we do not
know at this stage.

To see  a relationship with the Jones polynomial,
we pay attention to the constant term $\eth_t^0$ in the operators
$ S_{n-1}(\mathcal{A}_{(r,s)})$.
This extraction is realized in~\eqref{def_modify_P} by putting $t=q_u$ with a condition
\begin{equation}
  \label{Phi_zero}
  \Phi(q_u,q_u)=0 .
\end{equation}
For simplicity, we write
the  reduced DAHA polynomial as the constant term $\eth_t^0$
of
$ S_{n-1}(\mathcal{A}_{(r,s)}) $ as
\begin{equation}
  \label{colored_P_bar}
  \begin{aligned}[t]
    \overline{P}_n(q_u, x_u,x_d; \mathbb{c})
    &
    =
    \gamma
    \left(
      \left.
        \Phi(q_u, t)
        G \, M_{n-1}(\mathcal{O}_{(r,s)}; q, q) \,
        G^{-1}
        \Phi(q_u,t )^{-1}
      \right|_{t=q_u}
    \right)(1)
    \\
    & =
    \gamma
    \left(
      \left.
        \Const \left(
          S_{n-1}(\mathcal{A}_{(r,s)})
        \right)
      \right|_{t=q_u}
    \right)
    (1) ,
  \end{aligned}
\end{equation}
where $\Const(\bullet)$ is the $\eth_t^0$ term of~$\bullet$.
We do not need to take a conjugation by the gluing function~$G$ to pick up the constant
term.

\begin{conj}
  \label{conj:poly_reduced}
  The DAHA invariant $\overline{P}_n(q_u,x_u,x_d; \mathbb{c})$ for
  a simple closed curve $\mathbb{c}$ on $\Sigma_{2,0}$
  coincides with  the $n$-colored Jones polynomial for the
  double-torus
  knot $\mathbb{c}$
  up to framing
  when $x_u=x_d=q_u$.
\end{conj}

\begin{figure}[htbp]
  \centering
  \begin{tikzpicture}
    \node (AA) at (0,0) {
      \makecell{
        \includegraphics[scale=.7]{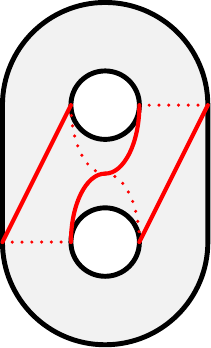}
      }};
    \node[right=16mm of AA] (AB) {
      \makecell{
        \includegraphics[scale=.7]{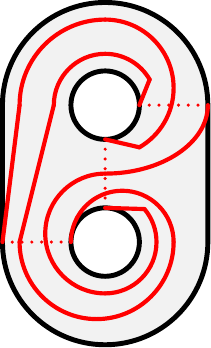}
      }};
    \draw[->] (AA) edge node[below] {
      \scriptsize $\mathscr{T}_{\mathbb{y}_u} \,
      \mathscr{T}_{\mathbb{y}_d}^{~-1}$
      } (AB);
  \end{tikzpicture}
  \caption{Shown is a curve $\mathbb{c}_{(1,2):(1,-1)}$ as an example
    of a  Dehn twist on a simple closed curve
    $\mathbb{c}_{(1,2)}$
    with slope $2$.
    In this case, we get  the figure-eight knot $4_1$ in $S^3$.}
  \label{fig:slope2}
\end{figure}

We give explicit  examples.
When $\mathscr{T}$ is generated by
$\mathscr{T}_{\mathbb{x}_u}$, $\mathscr{T}_{\mathbb{y}_u}$,
$\mathscr{T}_{\mathbb{x}_d}$, $\mathscr{T}_{\mathbb{y}_d}$,
the Dehn twist $\mathscr{T}(\mathbb{c}_{(1,1)})$
is a
connected sum of torus knots.
To have a hyperbolic knot,
we consider a simple closed curve
$\mathbb{c}=\mathscr{T}(\mathbb{c}_{(1,2)})$ given from the slope-$2/1$ curve.
Let  $\mathbb{c}_{(1,2):(k,\ell)}$
be a simple curve on $\Sigma_{2,0}$
given by 
$\mathscr{T}_{\mathbb{y}_u}^{~k} \mathscr{T}_{\mathbb{y}_d}^{~\ell}(\mathbb{c}_{(1,2)})$.
See Fig.~\ref{fig:slope2} for a case of $(k,\ell)=(1,-1)$.
For small $k$ and $\ell$, the simple closed curve
$\mathbb{c}_{(1,2):(k,\ell)}$
is identified with a knot of Rolfsen's notation as
in Table~\ref{tab:knot_k_l}~\cite{SnapPy,PHill99a}.
The curve $\mathbb{c}_{(1,2):(1,p)}$  is the twist knot
$K_p$ in $S^3$.

\begin{table}[tbhp]
  \centering
  \newcolumntype{C}{>{$}c<{$}}
  \resizebox{.98\textwidth}{!}{
    \begin{tabular}[]{C*{12}C}
      \toprule
      \text{knot}
      & 3_1 & 4_1 & 5_2 & 6_1 & 7_2 & 7_4
      & 8_1 & 8_3 & 9_2
      &9_5  &  10_1  & 10_3
      \\
      \midrule
      (k,\ell)
      &   (1,1) & (1,-1)  & (1,2) & (-1,2) & (1,3)
                                    & (-2,-2) & (-1,3)
            &(-2,2)
                  & (1,4) & (2,3) &(-1,4) & (-2,3)
      \\
      \bottomrule
    \end{tabular}
  }
  \caption{Examples of double-torus knots
    $\mathbb{c}_{(1,2):(k,\ell)}$.
    The so-called twist knot corresponds to the case of $k=1$.}
  \label{tab:knot_k_l}
\end{table}

The operator associated to~$\mathbb{c}_{(1,2)}$
on $\Sigma_{0,4}$,
$\mathbb{c}_{(1,2)} \mapsto \mathcal{A}_{(1,2)} \in \SH_{q,\boldsymbol{t}_\star}$,
is given
in~\eqref{def_Ars_others}, which preserves $\mathbb{C}[x+x^{-1}]$.
At $\boldsymbol{t}_\star$~\eqref{t_star} it is
written as
\begin{equation}
  \left. \mathcal{A}_{(1,2)} \right|_{\text{sym}}
  =
  A_{(1,2)}^{[0]}(q, x, x_u, x_d)
  +
  \sum_{j=1}^2 \sum_{\varepsilon= \pm 1}
  A_{(1,2)}^{[j]}(q, x^\varepsilon, x_u, x_d) \, \eth^{\varepsilon \, j}
  ,
\end{equation}
where
\begin{align*}
  &
    \begin{multlined}[b]
      A_{(1,2)}^{[2]}(q, x, x_u, x_d)
      \\
      =
      \frac{
        x \left( 1+ q^{\frac{1}{2}}x    \right)  \left(1+q^{\frac{3}{2}} x \right)
      }{
        \left( 1-x^2 \right)    \left( 1 -q^2 x^2 \right)
        \left( 1-q^{\frac{1}{2}} x \right) \left( 1-q^{\frac{3}{2}} x
        \right)
      }
      \,
      \prod_{\circledcirc \in \{u, d\}}
      \left( x_\circledcirc + q^{\frac{1}{2}} x \, x_\circledcirc^{~-1}
      \right)
      \left( x_\circledcirc + q^{\frac{3}{2}} x \,
        x_\circledcirc^{~-1}  \right)
      ,
    \end{multlined}
  \\
  &
    \begin{multlined}[b]
      A_{(1,2)}^{[1]}(q, x, x_u, x_d)
      \\
      =
      - \frac{2 q^{\frac{1}{2}} x \left( 1+ q^{\frac{1}{2}}x
        \right)^2}{
        \left( 1-x^2 \right)    \left( q^{\frac{1}{2}} -x \right)
        \left( 1-q^{\frac{1}{2}} x \right) \left( 1-q^{\frac{3}{2}} x
        \right)
      }
      \,
      \prod_{\circledcirc \in \{u, d\}}
      \left( x_\circledcirc + x_\circledcirc^{~-1} \right)
      \left( x_\circledcirc + q^{\frac{1}{2}} x \, x_\circledcirc^{~-1}
      \right)
      ,
    \end{multlined}
  \\
  &
    \begin{multlined}[t]
      A_{(1,2)}^{[0]}(x, q, x_u, x_d)
      \\
      =
      \frac{q^{\frac{3}{2}} x
        \left\{
          q \left( 1+x^2 \right)^2+
          4  \left( 1+q^2 \right)x^2+
          3 \,q^{\frac{1}{2}}  \left( 1+q \right)
          \left( 1+x^2 \right) x
        \right\}
      }{
        \left( q^{\frac{1}{2}} - x \right)
        \left( 1-q^{\frac{1}{2}} x \right)  \left(q^2-x^2 \right)
        \left( 1-q^2 x^2 \right)
      }
      \left(x_u+x_u^{~-1} \right)^2  \left( x_d+x_d^{~-1} \right)^2
      \\
      +
      \frac{q^{\frac{1}{2}} x
        \left\{
          2\, q^{\frac{3}{2}} \left(1+x^2 \right)
          + \left(1+q \right) \left(1+q^2 \right) x
        \right\}
      }{
        \left( q^2 -x^2 \right)  \left( 1-q^2 x^2 \right)
      }
      \left\{
        \left( x_u+x_u^{~-1} \right)^2
        +  \left( x_d+x_d^{~-1} \right)^2
      \right\}
      \\
      +
      \frac{q \left( 1-q \right)^2  x \left( 1+x^2 \right)}{
        \left( q^2-x^2\right) \left( 1-q^2 x^2 \right)
      } .
    \end{multlined}
\end{align*}
We take a conjugation by the gluing function $G$~\eqref{GGGud}, and then replace $x$ with $t$~\eqref{t_and_x_2} to obtain an
operator on $\SH_{q_u,t}\times \SH_{q_u,t}$ as
\begin{multline}
  \label{A12_conjugate}
  G \left.  \mathcal{A}_{(1,2)} \right|_{\text{sym}} G^{-1}
  =
  a_{(1,2)}^{[2]}(q_u,t) \,
  \left[
    \sh \left( t^{-1}\mathsf{X}_u\right)  \sh \left( t\,  \mathsf{X}_u\right) 
    \sh \left(t^{-1} \mathsf{X}_d \right) \sh \left( t \, \mathsf{X}_d \right)
    \eth_t
  \right]^2
  \\
  + a_{(1,2)}^{[1]}(q_u,t) \,
  \ch \left( \mathsf{X}_u \right)
  \ch \left( \mathsf{X}_d \right)
  \left[
    \sh \left( t^{-1} \mathsf{X}_u \right) \sh \left( t \,  \mathsf{X}_u \right)
    \sh \left( t^{-1} \mathsf{X}_d \right) \sh \left( t \,  \mathsf{X}_d \right)
    \eth_t
  \right]
  \\
  +
  a_{(1,2)}^{[0]}(q_u, t, \mathsf{X}_u, \mathsf{X}_d)
  +
  a_{(1,2)}^{[1]}(q_u, q_u \,t^{-1}) \,
  \ch  \left( \mathsf{X}_u\right)
  \ch \left( \mathsf{X}_d \right)\, \eth_t^{~-1}
  +
  a_{(1,2)}^{[2]}(q_u, q_u\, t^{-1}) \, \eth_t^{~-2}
  .
\end{multline}
Here we have
\begin{align*}
  a_{(1,2)}^{[2]}(q_u, t)
  & =
    \frac{
    - q_u^{~5} t^{10} \left( 1-t^2 \right) \left( 1-q_u^{~2} t^2 \right)
    }{
    \left( q_u^{~2}-t^4 \right)
    \left( 1+t^2 \right) \left( 1-q_u^{~2} t^4 \right)
    \left( 1+q_u^{~2} t^2 \right)
    }
    ,
  \\
  a_{(1,2)}^{[1]}(q_u, t)
  & =
    \frac{
    2 \,  q_u^{~3} t^6 \left( 1-t^2 \right)^2
    }{
    \left( 1+t^2 \right)
    \left( q_u^{~2}-t^4 \right)  \left(q_u^{~2}+t^2\right)
    \left( 1+q_u^{~2} t^2 \right)
    }
    ,
\end{align*}
and
\begin{multline}
  a_{(1,2)}^{[0]}(q_u, t, \mathsf{X}_u, \mathsf{X}_d)
  \\
  =
  \frac{- q_u^{~3} t^2
    \left\{
      t^4  \left(4-3t^2+t^4 \right)
      +q_u^{~4} \left( 1-3t^2+4t^4 \right)
      +q_u^{~2} t^2
      \left( -3+2t^2-3t^4 \right)
    \right\}
  }{
    \left( q_u^{~6}-t^4 \right)  \left( 1-q_u^{~2} t^4 \right)
    \left( 1+t^2 \right)  \left(  q_u^{~2} +t^2 \right)
  }
  \,
  \ch^2 \left( \mathsf{X}_u \right)
  \ch^2 \left( \mathsf{X}_d \right)
  \\
  +
  \frac{
    q_u  \, t^2 \, \left\{
      t^2(1+q_u^{~6}) +q_u^{~4} (t^2-2) +
      q_u^{~2} t^2 (1-2 t^2)
    \right\}
  }{
    \left( q_u^{~6} -t^4 \right)  \left( 1-q_u^{~2} t^4 \right)
  }
  \,
  \left\{
    \ch^2  \left( \mathsf{X}_u \right)
    +
    \ch^2 \left( \mathsf{X}_d \right)
  \right\}
  \\
  +
  \frac{
    -q_u  \, t^2  \left( 1-q_u^{~2} \right)^2  \left( q_u^2 + t^4 \right)
  }{
    \left( q_u^{~6}-t^4 \right)  \left(  1-q_u^{~2} t^4 \right)
  } .
\end{multline}

The reduced DAHA polynomial $\overline{P}_2$ for the simple closed curve
$\mathbb{c}_{(1,2):(k,\ell)}$
is given by
applying automorphism $\tau_{R(u)}^{~k} \tau_{R(d)}^{~\ell}$
to~\eqref{A12_conjugate}.
At $t=q_u$, the constant term $\eth^0$ of $\mathcal{A}_{(1,2)}$
reduces to a symmetric bilinear form,
\begin{equation}
  \begin{aligned}[b]
    &a_{(1,2)}^{[0]}( q_u, q_u, \mathsf{X}_u, \mathsf{X}_d)
    \\
    & =
    -\frac{q_u \left(1-q_u^2 \right)}{
       1-q_u^4    } \,
    \left(
      \frac{
        \left( 1-q_u^4 \right)^2}{
        q_u^2 \left( 1-q_u^2 \right)^2}
      -
      \frac{1-q_u^2}{q_u^2 \left( 1-q_u^4 \right)}
      \prod_{\circledcirc\in\{u,d\}}
      \left( 1-q_u^2 \mathsf{X}_\circledcirc^{~ 2} \right)
      \left( 1-q_u^2 \mathsf{X}_\circledcirc^{~ -2} \right)
    \right)
    \\
        & =
    -\frac{q_u \left(1-q_u^2 \right)}{
       1-q_u^4    } \,
    \left( 1, S_2(\ch  \mathsf{X}_u  ) \right) \,
    \begin{pmatrix}
      1 & 1 \\
      1 & - \frac{q_u^{2}(1-q_u^2)}{1-q_u^6}
    \end{pmatrix}
    \begin{pmatrix}
      1 \\
      S_2( \ch \mathsf{X}_d )
    \end{pmatrix}
    ,
  \end{aligned}
\end{equation}
where we have used the Chebyshev polynomial of the second kind~\eqref{Macdo_q_equal_t}.
As we have for $\SH_{q,t}$
\begin{equation*}
  \ch \left(
    \left( q^{-\frac{k}{2}} \mathsf{Y}^k \, \mathsf{X} \right)^2
  \right) (1)
  =
  \left( q \, t \right)^{2k} M_2(x; q,t)
  -\frac{ \left( 1-t^2 \right) \left( 1+q^2 \right)}{1-q^2 t^2} ,
\end{equation*}
we obtain
\begin{align}
  &\overline{P}_2(q_u, x_u,x_d; \mathbb{c}_{(1,2):(k,\ell)})
    =
    a_{(1,2)}^{[0]}(q_u, q_u,
    q_u^{-\frac{k}{2}} \mathsf{Y}_u^{~k} \mathsf{X}_u,
    q_u^{-\frac{\ell}{2}} \mathsf{Y}_u^{~\ell} \mathsf{X}_u
    ) (1)
    \notag
  \\
  & =
    -\frac{q_u \left(1-q_u^2 \right)}{
     1-q_u^4    } \,
    \left( 1, q_u^{4k } \, S_2(\ch  x_u  ) \right) \,
    \begin{pmatrix}
      1 & 1 \\
      1 & - \frac{q_u^{2}(1-q_u^2)}{1-q_u^6}
    \end{pmatrix}
    \,
    \begin{pmatrix}
      1 \\
      q_u^{4 \ell}  S_2( \ch x_d )
    \end{pmatrix}
  .
\end{align}
We have checked that $\overline{P}_2(q_u,q_u,q_u;
\mathbb{c}_{(1,2):(k,\ell)})$ for $(k,\ell)$ in Table~\ref{tab:knot_k_l}
agrees with the Jones polynomial up to
framing 
(see, \emph{e.g.},~\cite{KnotInfo2018,KnotAtlas}).

We compute the colored reduced DAHA
polynomial $\overline{P}_{n>2}$~\eqref{colored_P_bar}.
For  $n=3$, we need the difference operator
$S_2(\mathcal{A}_{(1,2)})=\mathcal{A}_{(1,2)}^{~2} -1$.
Using~\eqref{A12_conjugate}, we find that the operator
which survives at $t=q_u$ is a term $\eth^0$
given by
\begin{multline*}
  a_{(1,2)}^{[2]}(q_u,q_u) \, a_{(1,2)}^{[2]}(q_u, q_u^{-2})
    \prod_{\circledcirc \in \{u, d\} }
  \sh( q_u^{-2} \mathsf{X_\circledcirc} )
  \sh( q_u^{-1} \mathsf{X_\circledcirc} )
  \sh( q_u \mathsf{X_\circledcirc} )
  \sh( q_u^{2} \mathsf{X_\circledcirc} )
  \\
  + a_{(1,2)}^{[1]}(q_u, q_u) \, a_{(1,2)}^{[1]}(q_u, q_u^{-1})
  \prod_{\circledcirc \in \{ u, d\} }
  \ch^2 (\mathsf{X}_\circledcirc)
  \sh( q_u^{-1} \mathsf{X_\circledcirc} )
  \sh( q_u \mathsf{X_\circledcirc} )
  +
  \left(
    a_{(1,2)}^{[0]}(q_u, q_u, \mathsf{X}_u, \mathsf{X}_d)
  \right)^2
  -1
  \\
  =
  \frac{q  \left( 1-q \right)}{
    1-q^3}
  \Biggl\{
  \frac{ \left( 1-q^3 \right)^2
  }{ q^2 \left( 1-q \right)^2}
  -
  \frac{1-q^4}{q^3 \left( 1-q^2 \right)}
  \prod_{\circledcirc \in \{u,d \}}
  \left( 1-q \mathsf{X}_\circledcirc^{~ 2}  \right)
  \left( 1-q \mathsf{X}_\circledcirc^{~ -2} \right)
  \\
  +\frac{
    \left( 1-q \right) \left( 1-q^2 \right)
  }{
    q^3 \left( 1-q^4 \right) \left( 1-q^5 \right)
  }
  \prod_{\circledcirc \in \{u,d\}}
  \left(1-q \mathsf{X}_\circledcirc^{~ 2} \right) 
  \left( 1-q \mathsf{X}_\circledcirc^{~ -2} \right)
  \left( 1-q^2 \mathsf{X}_\circledcirc^{~ 2} \right)
  \left( 1-q^2 \mathsf{X}_\circledcirc^{~ -2} \right)
  \Biggr\}
  ,
\end{multline*}
where $q=q_u^{~2}$.
The $3$-colored polynomial for $\mathbb{c}_{(1,2):(k,\ell)}$ is given
by replacing $\mathsf{X}_u$ (resp. $\mathsf{X}_d$) by
$q_u^{-\frac{k}{2}} \mathsf{Y}_u^{~k} \, \mathsf{X}_u$
(resp.
$q_u^{-\frac{\ell}{2}} \mathsf{Y}_d^{~\ell} \,  \mathsf{X}_d$).
The DAHA $\SH_{q,t}$ proves
\begin{equation}
  \label{macdo_2j_auto}
  M_{j}(q^{-\frac{k}{2}} \mathsf{Y}^k \mathsf{X} ; q, t)(1)
  =
  \left( q^{\frac{1}{2} j^2 } t^{ j  } \right)^k \,
  M_{j}(x; q, t) ,
\end{equation}
which gives
\begin{multline}
  \overline{P}_3(q_u, x_u, x_d; \mathbb{c}_{(1,2):(k,\ell)})
  =
  \frac{q_u^2  \left( 1-q_u^2 \right) }{1-q_u^6}
  \times
  \\
  \left(
    1, q_u^{4k} S_2 ( \ch x_u), q_u^{12k} S_4 (\ch x_u)
  \right)
  \begin{pmatrix}
    1 & 1 & 1 \\
    1 & \frac{
      \left(1-q_u^2 \right)  \left(1- q_u^{12}\right)}{
      \left( 1-q_u^6 \right)  \left( 1-q_u^8 \right)} &
    - \frac{q_u^2 \left( 1-q_u^4 \right)}{1-q_u^8}
    \\
    1 & - \frac{q_u^2 (1-q_u^4)}{1-q_u^8} &
    \frac{q_u^6 \left( 1-q_u^2 \right) \left( 1-q_u^4 \right)}{
      \left( 1-q_u^8 \right) \left( 1-q_u^{10} \right)
    }
  \end{pmatrix}
  \begin{pmatrix}
    1 \\
    q_u^{4\ell} S_2(\ch {x}_d ) \\
    q_u^{12 \ell} S_4 (\ch {x}_d)
  \end{pmatrix}
  .
\end{multline}
We have checked that the results for knots in Table~\ref{tab:knot_k_l}
coincide with the known $N=3$ colored Jones polynomial
in~\cite{KnotAtlas}.

The DAHA polynomial for  $n=4$  is given similarly  from the  difference operator
$S_3(\mathcal{A}_{(1,2)})=
\mathcal{A}_{(1,2)}^{~3} - 2\mathcal{A}_{(1,2)}$.
Based on these explicit computations,
we  have the following observation for the $N$-colored reduced DAHA 
polynomial $\overline{P}_N$.
We define
\begin{gather}
  \label{define_base_s}
  \mathbf{s}_N({x})
  =
  \left(
    S_{2j}(x+x^{-1})
  \right)_{0 \leq j \leq N-1}
  ,
  \\
  \mathbf{v}_N(q,{x})
  =
  \left(
    \left(q \, {x}^2, q\, {x}^{-2} ; q\right)_{j}
  \right)_{0 \leq j \leq N-1}
  .
\end{gather}
These are bases of the symmetric Laurent polynomial space
$\mathbb{C}[x+x^{-1}]$ of
even power, and we have
\begin{equation}
  \label{v_and_s}
  \mathbf{v}_N(q, x )
  =
  \mathbf{s}_N(x) \, \mathbf{B}_N(q) ,
\end{equation}
where the $N\times N$ triangular matrix $\mathbf{B}_N(q)$ is defined by
\begin{equation*}
  \left( \mathbf{B}_N(q) \right)_{j,k}
  =
  \begin{cases}
    \displaystyle
    (-1)^j q^{\frac{1}{2} j(j+1)}
    \frac{
      (q;q)_{2k+1}}{
      (q;q)_{k-j} \, (q;q)_{k+j+1}
    },
    & \text{for $0 \leq j \leq k \leq N-1$,}
    \\
    0,  & \text{otherwise.}
  \end{cases}
\end{equation*}
One sees that these bases were used in~\cite{KHabiro06b,GMasb03a} for
a cyclotomic expansion
of the colored Jones polynomial.

\begin{conj}
  \label{conj:constant_term}
  The constant term $\eth^0$ of $S_{N-1}(\mathcal{A}_{(1,2)})$
  at $t=q_u=q^{\frac{1}{2}}$
  is given by
  \begin{equation}
    \Const\left(
      \left.
        S_{N-1}(
        \mathcal{A}_{(1,2)}
        )
      \right|_{t=q_u}
    \right)
    =
    (-1)^{N-1} q^{\frac{1}{2}(N-1)}
    \frac{1-q}{1-q^N} \,
    \mathbf{v}_N(q, \mathsf{X}_u) \,
    \mathbf{T}_N(q; \mathbb{c}_{(1,2)} ) \,
    \mathbf{v}_N(q, \mathsf{X}_d)^\top
    ,
    \label{const_slope12}
  \end{equation}
  where  the diagonal matrix $\mathbf{T}_N$ is
  \begin{equation}
    \mathbf{T}_N(q; \mathbb{c}_{(1,2)})
    =
    \diag
    \left(
      (-1)^{k-1} q^{\frac{1}{2} k(k+1-2N)}
      \frac{
        (q^{2k} ; q)_{N+1-2k}
      }{
        (q^{k}; q)_{N+1-2k}
      } \,
      \frac{
        (q^{2k}; q)_{N-k}}{
        (q^{k}; q)_{N-k}
      }
    \right)_{1 \leq k \leq N}
    .
  \end{equation}
\end{conj}

For computations of the reduced DAHA polynomial of
$\mathscr{T}_{\mathbb{y}_u}^k
\mathscr{T}_{\mathbb{y}_d}^\ell(\mathbb{c}_{(1,2)})$,
we need  $\mathbf{v}_N(q, \tau_L^k (\mathsf{X}_\circledcirc))(1)$
in~\eqref{const_slope12}.
By changing bases using~\eqref{v_and_s} to the Chebyshev polynomial, we have
$\mathbf{s}_N(\tau_L^k(\mathsf{X}_\circledcirc))(1)
=
\mathbf{s}_N(q_u^{~ -\frac{k}{2}} \mathsf{Y}_\circledcirc^k \mathsf{X}_\circledcirc)(1)$ at $t=q_u$,
which can be computed  by $\SH_{q,t}$ as~\eqref{macdo_2j_auto}.
To conclude,  Conjecture~\ref{conj:poly_reduced}
of the reduced DAHA polynomial for $\mathbb{c}_{(1,2):(k,\ell)}$ is
read under Conjecture~\ref{conj:constant_term} as follows.

\begin{conj}
  The $N$-colored Jones polynomial for
  $\mathbb{c}_{(1,2):(k,\ell)}=\mathscr{T}_{\mathbb{y}_u}^k
  \mathscr{T}_{\mathbb{y}_d}^\ell(\mathbb{c}_{(1,2)})$
  coincides up to framing
  with the  reduced DAHA polynomial
  $\overline{P}_N(q_u, x_u, x_d;  \mathbb{c})$
  at $x_u=x_d=q_u=q^{\frac{1}{2}}$,
  where we have
  \begin{multline}
    \label{reduce_DAHA_kl}
    \overline{P}_N(q_u, x_u, x_d; \mathbb{c}_{(1,2):(k,\ell)})
    =
    (-1)^{N-1}q^{\frac{1}{2}(N-1)} \frac{ 1-q}{1-q^N}
    \\
    \times
    \mathbf{s}_N(x_u) \,
    \diag \left( q^{j(j+1)k} \right)_{0\leq j\leq N-1} \,
    \mathbf{B}_N(q) \,
    \mathbf{T}_N(q; \mathbb{c}_{(1,2)}) \,
    \mathbf{B}_N(q)^\top \,
    \diag \left( q^{j(j+1)\ell} \right)_{0 \leq j \leq N-1} \,
    \mathbf{s}_N(x_d)^\top
    .
  \end{multline}
\end{conj}
One can check the conjecture for small $N$ and the knots in
Table~\ref{tab:knot_k_l}.
Furthermore,
as an explicit form of the $N$-colored Jones polynomial for twist knot
$K_p$ is given
in~\cite{GMasb03a},
we have checked the equality for small  $p$ and $N$,
\begin{equation}
  \overline{P}_N(q_u, q_u, q_u; \mathbb{c}_{(1,2):(1,p)})
  =
  (-1)^{N-1}
  \frac{
    q^{N/2} - q^{-N/2}
  }{
    q^{1/2} - q^{-1/2}
  } \,
  J_N(q;K_p) .
\end{equation}
Here the colored Jones polynomial is read as
\begin{equation}
  \label{Masbaum_Jones}
  J_N(q; K_p)
  =
  \sum_{n=0}^\infty q^n
  (q^{1-N}, q^{1+N};q)_n
  \sum_{j=0}^n
  (-1)^j q^{j(j+1)p+ \frac{1}{2} j(j-1)}
  \left( 1-q^{2j+1} \right)
  \frac{
    (q;q)_n}{
    (q;q)_{n+j+1} (q;q)_{n-j}
  } ,
\end{equation}
which is normalized such that $J_N(q; \text{unknot})=1$.
See also~\cite{GukNawSabStoSul16a} where
a similar expression with~\eqref{reduce_DAHA_kl} was given for
the Poincar{\'e} polynomial for  knot homology $\mathbb{c}_{(1,2):(k,\ell)}$.

\section{Concluding Remarks}

We have clarified topological aspects of rank-1 DAHAs,
$A_1$-type and
$C^\vee C_1$-type.
Motivated by the DAHA--Jones polynomial for torus knots by Cherednik~\cite{IChered13a},
we have proposed to combine these two rank-1 DAHAs as a
representation of the skein algebra on the genus-two surface $\Sigma_{2,0}$.
We have constructed the reduced  DAHA polynomial
$\overline{P}_n(q_u,x_u, x_d; \mathbb{c})$
for a simple closed curve $\mathbb{c}$ on
$\Sigma_{2,0}$, and have clarified the relationship with the colored
Jones polynomial for double-torus knot.
In this paper, we have mainly studied double twist knots
$\mathscr{T}_{\mathbb{y}_u}^k \mathscr{T}_{\mathbb{y}_d}^\ell
(\mathbb{c}_{(1,2)})$, and in Appendix~\ref{sec:other_slope2}
we give some results on other simple closed curves given from
$\mathbb{c}_{(1,2)}$.

In Appendix~\ref{sec:example_slope3}, we give some  results concerning
simple closed curves derived from
$\mathbb{c}_{(1,3)}$, and discuss a relationship with the colored
Jones polynomial.
These results  indicate an importance of the Askey--Wilson operator at
$\boldsymbol{t}_\star$~\eqref{t_star} associated to a simple closed
curve $\mathbb{c}_{(r,s)}$ in studies of quantum invariants of knot.
The symmetric bilinear form~\eqref{reduce_DAHA_kl} for those  knot families seems to be promising for
other quantum polynomial  invariants of knots.
Also
studies of
the
(non-reduced)  DAHA polynomial $P_n(t,q_u,x_u,x_d, \Phi ; \mathbb{c})$,
higher rank cases,
and  skein algebra on higher genus surfaces
will reveal
a fruitful structure~\cite{HikamiWork}.
Also the  cluster algebraic construction of the skein
algebra~\cite{KHikami17a} will be useful  in DAHA.


We have studied the reduced DAHA polynomial for a simple curve $\mathbb{c}=\mathscr{T}(\mathbb{c}_{(r,s)})$
where the Dehn twist $\mathscr{T}$ is
generated by
$\mathscr{T}_{\mathbb{x}_u}$,
$\mathscr{T}_{\mathbb{y}_u}$,
$\mathscr{T}_{\mathbb{x}_d}$, and
$\mathscr{T}_{\mathbb{y}_d}$.
In this case, we need only to apply the $SL(2;\mathbb{Z})$ actions of $A_1$-type
DAHA to the $q$-difference operator which is defined in terms of the
$C^\vee C_1$-type DAHA.
Even in the case that
we further have  the Dehn twist $\mathscr{T}_{\mathbb{y}}$
about $\mathbb{y}$,
it is possible to define  the DAHA polynomial using $V_L$
corresponding to the Dehn twist about $\mathbb{y}$,
as
we have constructed the conjugation~\eqref{auto_CC1} for the
$C^\vee C_1$-type DAHA.
Unfortunately such computations are much  involved, and
it remains for
future studies.

\appendix
\section{Other Simple Closed Curves derived from $\mathbb{c}_{(1,2)}$}
\label{sec:other_slope2}

The Dehn twist $\mathscr{T}_{\mathbb{x}_\circledcirc}^{-k}
\mathscr{T}_{\mathbb{y}_\circledcirc}$ induces
the automorphism of $\SH_{q,t}$,
$\tau_R^{k} \tau_L (\mathsf{X})
=q^{\frac{1}{2}(k-1)} \mathsf{X}^k \mathsf{Y}
\mathsf{X}$.
At $t=q$,
we have
\begin{equation}
  \label{ex12_xky}
  S_{2j}(\ch (q^{\frac{1}{2}(k-1)} \mathsf{X}^k \mathsf{Y}
  ) )(1)
  =
  \sum_{i=-j}^j S_{2(k+1)i}(\ch x) \, q^{2i ( (k+1)i+1)} .
\end{equation}
Here  as we have
$S_n(x+x^{-1})=\frac{x^{n+1}-x^{-n-1}}{x-x^{-1}}$ for $n\geq 0$,
the   Chebyshev polynomial for negative integers means
\begin{equation*}
  S_{-n}(\ch x)=-S_{n-1}(\ch x).  
\end{equation*}

The Dehn twist $\mathscr{T}_{\mathbb{y}_\circledcirc}^k
\mathscr{T}_{\mathbb{x}_\circledcirc}
\mathscr{T}_{\mathbb{y}_\circledcirc}$ induces the $\SH_{q,t}$
automorphism,
$\tau_L^k \tau_R \tau_L(\mathsf{X})
=
q^{-k} \mathsf{Y}^k \mathsf{X} \mathsf{Y}^{k+1} \mathsf{X}$.
At $t=q$, we have
\begin{equation}
  \label{ex12_yxy}
  S_{2j}(
  \ch(q^{-k} \mathsf{Y}^k \mathsf{X} \mathsf{Y}^{k+1} \mathsf{X})
  )(1)
  =
  \sum_{i=0}^{2j} (-1)^i q^{(2k+1) i(i+1)} S_{2i}(\ch x).
\end{equation}

As we have discussed,
the reduced DAHA polynomials can be given from
the constant term~\eqref{const_slope12} by applying the automorphisms.
For instance, when
a simple closed  curve $\mathbb{c}$  is given from 
$\mathbb{c}_{(1,2)}$ as
$\mathscr{T}_{\mathbb{y}_u}^{k} \mathscr{T}_{\mathbb{x}_u}^{-1}
\mathscr{T}_{\mathbb{y}_u}
\mathscr{T}_{\mathbb{y}_d}^\ell$, the reduced DAHA polynomial is
then given by
\begin{multline}
  \overline{P}_N(q_u, x_u, x_d;
  \mathscr{T}_{\mathbb{y}_u}^{k} \mathscr{T}^{-1}_{\mathbb{x}_u}
  \mathscr{T}_{\mathbb{y}_u}
  \mathscr{T}_{\mathbb{y}_d}^\ell(\mathbb{c}_{(1,2)}))
  \\
  =
  (-1)^{N-1}q^{\frac{1}{2}(N-1)} \frac{ 1-q}{1-q^N}
  \,
  \mathbf{s}_{2N-1}(x_u) \,
  \left( (-1)^i q^{\frac{1}{2}(2k+1)i(i+1)} \right)_{
    0 \leq i \leq 2 j \leq 2(N-1)}
  \\
    \times
    \mathbf{B}_N(q) \,
    \mathbf{T}_N(q; \mathbb{c}_{(1,2)}) \,
    \mathbf{B}_N(q)^\top \,
    \diag \left( q^{j(j-1)\ell} \right)_{1 \leq j \leq N} \,
    \mathbf{s}_N(x_d)^\top
    .
  \end{multline}
  We have checked the polynomial at $x_u=x_d=q_u$ for
  $(k,\ell)=(-1,1)$
  (resp. $(-1,-1)$)
  coincidences with the $N$-colored Jones polynomial
for $\overline{6_1}$ (resp. $\overline{5_2}$).
We can give
the formula for the curves obtained by the above Dehn twists for $\mathsf{X}_\circledcirc$
by applying~\eqref{ex12_xky}  and~\eqref{ex12_yxy}

\section{Double-torus knots  from $\mathbb{c}_{(1,3)}$}
\label{sec:example_slope3}
We give some results on the reduced DAHA polynomial for double-torus
knots, which are given from  the simple closed curve
$\mathbb{c}_{(1,3)}$ on $ \Sigma_{2,0}$.
The difference operator
$\mathcal{A}_{(1,3)} \in \SH_{q,\boldsymbol{t}_\star}$
associated to
the curve
$\mathbb{c}_{(1,3)}$
is given in~\eqref{slope1odd}.
At~$\boldsymbol{t}_\star$~\eqref{t_star},
we can compute it explicitly as follows.
\begin{multline}
  G \,
  \left. \mathcal{A}_{(1,3)} \right|_{\text{sym}}  \,
  G^{-1}
  =
  a_{(1,3)}^{[3]}(q_u, t) \,
  \left[
     \sh (t^{-1}\mathsf{X}_u)     \sh(t \mathsf{X}_u)
    \sh (t^{-1} \mathsf{X}_d)     \sh(t \mathsf{X}_d)\,
    \eth_t
  \right]^3
  \\
  +
  a_{(1,3)}^{[2]}(q_u, t) \, \ch (\mathsf{X}_u) \ch (\mathsf{X}_d) \,
  \left[
     \sh (t^{-1}\mathsf{X}_u)     \sh(t \mathsf{X}_u)
    \sh (t^{-1} \mathsf{X}_d)     \sh(t \mathsf{X}_d)\,
    \eth_t
  \right]^2
  \\
  +
  a_{(1,3)}^{[1]}(q_u, t, \mathsf{X}_u, \mathsf{X}_d) \,
  \left[
    \sh (t^{-1}\mathsf{X}_u)     \sh(t \mathsf{X}_u)
    \sh (t^{-1} \mathsf{X}_d)     \sh(t \mathsf{X}_d)\,
    \eth_t
  \right]
  \\
  +
  a_{(1,3)}^{[0]}(q_u, t, \mathsf{X}_u, \mathsf{X}_d)
  +  a_{(1,3)}^{[1]}(q_u, q_u \, t^{-1} ,\mathsf{X}_u, \mathsf{X}_d)
  \, \eth_t^{~-1}
  \\
  +a_{(1,3)}^{[2]}(q_u, q_u \, t^{-1}) \,  \ch (\mathsf{X}_u)
  \ch (\mathsf{X}_d) \, \eth_t^{~ -2}
  +  a_{(1,3)}^{[3]}(q_u, q_u \, t^{-1}) \,
  \eth_t^{~ -3}
  ,
\end{multline}
where $G$ is the gluing function~\eqref{GGGud}, and
\begin{align*}
  &
  a_{(1,3)}^{[3]}(q_u, t)
    =
    \frac{
    q_u^{~ 13} t^{14} (1-t^2) (1-q_u^2 t^2) (1-q_u^4 t^2)
    }{
    (1+t^2)(1+q_u^2 t^2) (1+q_u^4 t^2) (q_u^2-t^4)
    (1-q_u^2   t^4) (1-q_u^6 t^4)}
    ,
  \\
  &
    a_{(1,3)}^{[2]}(q_u, t)
    =
    \frac{
    -2 q_u^7 t^{10} (1-t^2) (1-q_u^2 t^2)
    \left( 1- (1+q_u^2)t^2\right)
    }{
    (1+t^2) (q_u^2+t^2) (1+q_u^2 t^2) (1+q_u^4 t^2) (q_u^2-t^4)
    (1-q_u^2 t^4)
    }
    ,
  \\
  &
    \begin{multlined}[t]
      a_{(1,3)}^{[1]}(q_u, t, \mathsf{X}_u, \mathsf{X}_d)
      =
      \frac{q_u^3 t^6 \left( 1-t^2 \right)
      }{
        \left( 1+t^2 \right)  \left( q_u^2+t^2 \right)
        \left( 1+q_u^2 t^2 \right) \left( q_u^2-t^4 \right)
        \left( q_u^6 - t^4 \right) \left( 1-q_u^6 t^4 \right)
      }
      \\
      \times
      \Biggl\{
      -q_u^2 
      \bigl(
      t^4(1+q_u^8)(-4+3t^2)
      - q_u^2 t^2(1+q_u^4)(-3+5t^2-4t^4+t^6)
      \\
      + q_u^4(-1+4t^2-5t^4+6t^6-4t^8)
      \bigr) \ch^2(\mathsf{X}_u) \ch^2(\mathsf{X}_d)
      \\
      -(q_u^2+t^2)( 1+q_u^2 t^2)
      \biggl(
      t^2(1+q_u^8)+q_u^2(1+q_u^4)t^2(1-2t^2)
      +q_u^4(-2+t^2+t^6)
      \biggr) \left( \ch^2(\mathsf{X}_u) + \ch^2(\mathsf{X}_d) \right)
      \\
      +(q_u^2+t^2)(1+q_u^2 t^2) (1-q_u^2)^2 (t^4+q_u^2+q_u^4 t^4)
      \Biggr\} ,
    \end{multlined}
  \\
  &
    \begin{multlined}[t]
      a_{(1,3)}^{[0]}(q_u, t, \mathsf{X}_u, \mathsf{X}_d)
      =
      \frac{2 q_u t^2}{
        \left( 1+t^2 \right)  \left( q_u^2+t^2 \right)
        \left( q_u^4+t^2 \right)
        \left( 1+q_u^2 t^2 \right)  \left( q_u^6-t^4 \right)
        \left( 1-q_u^2 t^4 \right)
      } \ch (\mathsf{X}_u) \ch (\mathsf{X}_d)
      \\
      \times
      \Biggl\{
      q_u^4 t^2
      \biggl(
      t^4(3-3t^2+t^4)
      +q_u^2 t^2 (-3+5t^2-4t^4+t^6)
      +q_u^4 (1-4t^2+5t^4-3t^6)
      \\
      +q_u^6(1-3t^2+3t^4)
      \biggr) \ch^2(\mathsf{X}_u) \ch^2 (\mathsf{X}_d)
      \\
      - q_u^2 \biggl(
      t^6+ q_u^2 t^6 (3-t^2-2t^4+t^6)+q_u^4 t^4(-1+t^4-2t^6)
      +q_u^6 t^2 (-2+t^2-t^6)
      \\
      + q_u^8 (1-2t^2-t^4+3t^6)+q_u^{10}t^6
      \biggr) \left( \ch^2 (\mathsf{X}_u) + \ch^2 (\mathsf{X}_d)
      \right)
      \\
      +
      \biggl(
      t^6(1+t^2) +
      q_u^2 t^4(1+t^2+2t^4-t^8)
      - q_u^4 t^4 (-2-5t^2+2t^4+3t^6+t^8)
      \\
      -
      q_u^6 t^4(2+t^2-t^4+3t^6+t^8)
      -
      q_u^8 (1+3t^2-t^4+t^6+2t^8)
      +
      q_u^{10}(-1-3t^2-2t^4+5t^6+2t^8)
      \\
      +q_u^{12}(-1+2t^4+t^6+t^8)
      + q_u^{14} t^4 (1+t^2)
      \biggr)
      \Biggr\} .
    \end{multlined}
\end{align*}
The constant  term $\eth^0$ in $\mathcal{A}_{(1,3)}$
at $t=q_u$ is written as
\begin{multline}
  a_{(1,3)}^{[0]}( q_u, q_u ,\mathsf{X}_u, \mathsf{X}_d)
  \\
  =
  -\frac{q_u \left( 1-q_u^2 \right)}{1-q_u^4}
  \Biggl\{
  1
  +
  \frac{ \left( 1-q_u^2 \right) \left(1-q_u^4\right)
  }{
    \left( 1- q_u^6 \right) \left( 1-q_u^8 \right)}
  \prod_{\circledcirc \in \{u,d \}}
  \left(1-q_u^{~2} \mathsf{X}_\circledcirc^2 \right)
  \left( 1-q_u^{~2} \mathsf{X}_\circledcirc^{~ -2} \right)
  \Biggr\}
  \prod_{\circledcirc\in\{u,d\}}
  \left( \mathsf{X}_\circledcirc+\mathsf{X}_\circledcirc^{-1} \right)
  .
  \label{const_a_13}
\end{multline}

Let $\mathbb{c}_{(1,3):(k,\ell)}$ be a simple closed curve on $\Sigma_{2,0}$
given by
$
\mathscr{T}_{\mathbb{y}_u}^k \mathscr{T}_{\mathbb{y}_d}^\ell
(\mathbb{c}_{(1,3)})
$.
Applying the DAHA automorphisms to~\eqref{const_a_13},
the reduced  DAHA polynomial is given by
\begin{equation}
  \overline{P}_2(q_u, x_u, x_d; \mathbb{c}_{(1,3):(k,\ell)})
  =
  a_{(1,3)}^{[0]}(q_u, q_u,
  q_u^{-\frac{k}{2}} \mathsf{Y}_u^{~k} \mathsf{X}_u,
  q_u^{-\frac{\ell}{2}} \mathsf{Y}_d^{~\ell} \mathsf{X}_d
  )(1)
  .
\end{equation}
As $\mathsf{X}_\circledcirc$ and $\mathsf{Y}_{\circledcirc}$ are
generators of $\SH_{q_u,t}$,
we can make use of~\eqref{macdo_2j_auto} of $\SH_{q,t}$ in the
computation.
As a result, we obtain
a symmetric bilinear form for the reduced DAHA polynomial
\begin{multline}
  \label{poly_for_13kl}
  \overline{P}_2(q_u,x_u, x_d; \mathbb{c}_{(1,3):(k, \ell)})
  =
-  \frac{q_u \left( 1-q_u^2 \right)}{1-q_u^4}
  \\
  \times
  \left(
    q_u^{\frac{3}{2}k} S_1(\ch {x}_u)
    ,
    q_u^{\frac{15}{2} k} S_3(\ch {x}_u)
  \right) \,
  \begin{pmatrix}
    1 + \frac{
      \left( 1-q_u^2 \right) \left( 1-q_u^8 \right)}{
      \left( 1-q_u^4 \right) \left( 1-q_u^6 \right)}
    &
    - \frac{
      q_u^2 \left( 1-q_u^2 \right)}{
      1-q_u^6
    }
    \\
    - \frac{q_u^2 \left( 1-q_u^2 \right)}{
      1-q_u^6
    }
    &
    \frac{q_u^4 \left( 1-q_u^2 \right)
      \left( 1-q_u^4 \right)}{
      \left( 1-q_u^6 \right) \left( 1-q_u^8 \right)
    }
  \end{pmatrix}
  \begin{pmatrix}
    q_u^{\frac{3}{2}\ell} S_1(\ch {x}_d ) \\
    q_u^{\frac{15}{2}\ell} S_3(\ch {x}_d)
  \end{pmatrix}
  .
\end{multline}
We see that the curve $\mathbb{c}_{(1,3):(k,\ell)}$ with
$(k,\ell)=(1,-1)$ denotes a connected sum
of trefoils $3_1 \# \overline{3_1}$
(the square knot).
Also
SnapPy~\cite{SnapPy} tells us that
the curves $\mathbb{c}_{(1,3):(k,\ell)}$ for $(k,\ell)=(1,1)$
and $(1,2)$ are  $\overline{9_{46}}$ and $k7_{125}$ respectively.
We have checked
that
the reduced DAHA polynomial  $\overline{P}_2$~\eqref{poly_for_13kl}
for these closed curves
coincide with the Jones polynomials for $3_1\# \overline{3_1}$,
$\overline{9_{46}}$,
and $k7_{125}$
in~\cite{KnotInfo2018,ChamKofmPatt04a}
at $x_u=x_d=q_u$.

For the $n=3$ colored reduced DAHA polynomial $\overline{P}_3$,
the constant term of
$S_2(\mathcal{A}_{(1,3)})
=
\mathcal{A}_{(1,3)}^{~ 2}-1$ at $t=q_u=q^{\frac{1}{2}}$ is computed as
\begin{multline*}
  a_{(1,3)}^{[3]}(q_u,q_u) \, a_{(1,3)}^{[3]}(q_u, q_u^{~ -3})
  \prod_{\circledcirc \in \{u,d\}}
  \prod_{j=1}^3
  \sh ( q_u^{~ -j} \mathsf{X}_\circledcirc)
  \sh ( q_u^{~ j} \mathsf{X}_\circledcirc)
  \\
  +
  a_{(1,3)}^{[2]}(q_u,q_u) \, a_{(1,3)}^{[2]}(q_u,q_u^{~ -2})
 \prod_{\circledcirc \in \{u,d\}}
 \ch^2(\mathsf{X}_\circledcirc)
 \prod_{j=1}^2
  \sh ( q_u^{~ -j} \mathsf{X}_\circledcirc)
  \sh ( q_u^{~ j} \mathsf{X}_\circledcirc)
  \\
  +
  a_{(1,3)}^{[1]}(q_u,q_u, \mathsf{X}_u, \mathsf{X}_d) \,
  a_{(1,3)}^{[1]}(q_u,q_u^{~ -1}, \mathsf{X}_u, \mathsf{X}_d)
  \prod_{\circledcirc \in \{u,d\}}
  \sh ( q_u^{-1} \mathsf{X}_\circledcirc)
  \sh ( q_u \mathsf{X}_\circledcirc)
  \\
  + \left( a_{(1,3)}^{[0]}(q_u, q_u, \mathsf{X}_u, \mathsf{X}_d) \right)^2
  -1
  \\
  =
  \frac{
    q \left( 1-q \right)^2  \left( 1-q^2\right) }{
    \left( 1-q^5 \right) \left( 1-q^6 \right)
    \left( 1-q^7 \right)}
  \prod_{\circledcirc \in \{u, d\}}
  (q \mathsf{X}_\circledcirc^{~ 2}, q \mathsf{X}_\circledcirc^{~ -2}; q)_2
  \left(  \ch( \mathsf{X}_\circledcirc^{~ 2} ) +1 + \frac{q}{1+q^2}
  \right)
  \\
  + \frac{
    q \left( 1-q \right) \left( 1-q^2 \right)^2}{
    \left( 1-q^3 \right) \left( 1-q^4 \right)^2
  }
  \prod_{\circledcirc \in \{u, d \}}
  \left(1-q \mathsf{X}_\circledcirc^{~ 2} \right)
  \left( 1-q \mathsf{X}_\circledcirc^{~ -2}  \right)
  \ch^2 ( \mathsf{X}_\circledcirc )
  + \frac{q \left( 1-q \right)}{1-q^3}
  \prod_{\circledcirc \in \{u, d \}}
  \left( \ch(\mathsf{X}_\circledcirc^{~ 2})+ 1 \right) .
\end{multline*}
Applying the $SL(2;\mathbb{Z})$ actions on $\mathsf{X}_\circledcirc$,
we obtain  a symmetric bilinear form for
the reduced DAHA polynomial of
$\mathbb{c}_{(1,3):(k,\ell)}$
as
\begin{multline}
  \overline{P}_3(q_u, x_u, x_d; \mathbb{c}_{(1,3):(k,\ell)})
  =
  \frac{q \left( 1-q \right)}{1-q^3} \,
  \mathbf{s}_4(x_u) \,
  \diag\left( 1,q^{2k}, q^{6k}, q^{12k} \right)
  \\
  \times
  \left(
    \begin{smallmatrix}
      1 &  \frac{  \left(1-q \right)
        \left( 1-q^6 \right)}{ \left( 1-q^3 \right)
        \left( 1-q^4 \right)}
      &
      - \frac{q \left( 1-q^2 \right)}{1-q^4}
      & 0
      \\
      \frac{ \left( 1-q \right) \left( 1-q^6 \right)
      }{
        \left( 1-q^3 \right) \left( 1-q^4 \right)}
      & \frac{\left( 1-q \right) \left( 1-q^2 \right)
        \left( 3+3q^2+4q^3+3q^4+3q^6 \right)}{
        \left( 1-q^4 \right) \left( 1-q^5 \right)}
      &
      - \frac{q \left( 1-q \right)
        \left( 1-q^2 \right)
        \left( 2+q+q^3+2q^4 \right)
      }{ \left( 1-q^4 \right)
        \left( 1-q^5 \right)}
      & \frac{q^3 \left( 1-q \right)
        \left( 1-q^2 \right)}{
        \left( 1-q^4 \right) \left( 1-q^5 \right)}
      \\
      - \frac{q \left( 1-q^2 \right)}{1-q^4}
      &
      - \frac{q \left(1-q \right)
        \left( 1-q^2 \right)
        \left( 2+q+q^3+2q^4 \right)
      }{
        \left( 1-q^4 \right) \left( 1-q^5 \right)}
      &
      \frac{q^2 \left( 1-q^2 \right)^2
        \left( 1-q^3 \right)
        \left(2+q+q^3+2q^4 \right)}{
        \left( 1-q^4 \right) \left( 1-q^5 \right)
        \left( 1-q^6 \right)
      }&
      - \frac{q^4 \left( 1-q^2 \right)^2 \left(1-q^3 \right)
      }{
        \left(1-q^4 \right) \left( 1-q^5 \right)
        \left( 1-q^6 \right)}
      \\
      0 &
      \frac{q^3 \left( 1-q \right) \left( 1-q^2 \right)
      }{ \left( 1-q^4 \right) \left( 1-q^5 \right)}
      &
      - \frac{q^4 \left( 1-q^2 \right)^2 \left( 1-q^3 \right)}{
        \left( 1-q^4 \right) \left( 1-q^5 \right)
        \left( 1-q^6 \right)}
      &
      \frac{q^6 (1-q) \left( 1-q^2 \right) \left( 1-q^3 \right)}{
        \left( 1-q^5 \right) \left( 1-q^6 \right)
        \left( 1-q^7 \right)
      }
    \end{smallmatrix}
  \right)
  \\
  \times \diag \left(1, q^{2\ell}, q^{6\ell}, q^{12\ell} \right) \,
  \mathbf{s}_4(x_d)^\top ,
\end{multline}
where
$\mathbf{s}_N(x)$ is defined in~\eqref{define_base_s}.
We have checked
that the results for $(k,\ell)=(1,- 1)$ and $(1,1)$
agree up to framing factor with the
$n=3$ colored Jones polynomials for the square knot
$3_1 \# \overline{3_1}$ and
$\overline{9_{46}}$ in~\cite{KnotInfo2018} respectively when $x_u=x_d=q_u=q^{\frac{1}{2}}$ as expected.


\end{document}